\documentclass[acmsmall,screen,dvipsnames]{acmart}

\usepackage{mylhs2tex}

\usepackage{amssymb}
\usepackage{xcolor}
\usepackage{tikz}
\usetikzlibrary{tikzmark}
\usetikzlibrary{shapes,shapes.geometric,backgrounds,patterns,shapes.misc,positioning}
\usetikzlibrary{arrows,automata,calc}
\usepackage{wrapfig}
\usepackage{listings}
\usepackage{xspace}
\usepackage[inline]{enumitem}
\usepackage{quiver}
\usepackage{mathtools}
\usepackage{bbm}

\usepackage{cleveref}
\Crefformat{figure}{#2Fig.~#1#3}
\Crefformat{section}{\S#2#1#3}
\Crefformat{subsection}{\S#2#1#3}
\Crefformat{subsubsection}{\S#2#1#3}
\usepackage{lipsum}

\setlist[enumerate]{leftmargin=5.5mm}
\setlist[itemize]{leftmargin=3.5mm}
\usepackage{graphicx}
\usepackage{pgf}
\usepackage{multirow}

\usepackage{algorithm}
\usepackage[noend]{algpseudocode}
\usepackage{soul} %
\usepackage{subcaption}
\usepackage{longtable}
\usepackage{mathtools}

\usepackage{thmtools} 
\usepackage{thm-restate}
\newtheorem{theorem}{Theorem}[section]

\newtheorem{lemma}[theorem]{Lemma}

\usepackage{diagbox}

\usepackage{ottalt}

\newcommand{\ottafterlastrule}{\\}

  \renewottcommands[ott]

\renewcommand\ottaltinferrule[4]{
  \inferrule*[narrower=1,lab=#1,#2]
  {#3}
  {#4}
}

\usepackage{tcolorbox}
\tcbuselibrary{skins,listings,breakable}

\newtcolorbox{mycode}[1][gray!70]{
    enhanced,
    colback=white,
    colframe=white,
    overlay={\begin{tcbclipinterior}\fill[#1]
    ([xshift=1.2em,yshift=1.2em]frame.south west)
    rectangle 
    ([xshift=0.9em,yshift=-1.5em]frame.north west);\end{tcbclipinterior}}
}

\RequirePackage{calc}
\newcommand{\keyword}[1]{\textcolor{BlueViolet}{\textbf{#1}}}
\newcommand{\idd}[1]{\textsf{\textsl{#1}}}

\newcommand{\mykeyword}[1]{\textcolor{BlueViolet}{\idd{#1}}}
\newcommand{\varid}[1]{{\idd{#1}}}
\newcommand{\effect}[1]{\textcolor{Sepia}{\idd{#1}}}
\newcommand{\conid}[1]{\textcolor{OliveGreen}{\idd{#1}}}

\newcommand{\glocal}[3]{\langle #2 \rangle^{#1}_{#3}} 
\newcommand{\sglocal}[3]{\langle\!\langle #2 \rangle\!\rangle^{#1}_{#3}} 

\definecolor{listing-data}{RGB}{71,19,169}
\definecolor{listing-type}{RGB}{211,54,130}
\definecolor{listing-keyword}{RGB}{19,42,169} %
\definecolor{listing-number}{RGB}{71,19,169}

\makeatletter
\patchcmd{\@@addmarginpar}{\ifodd\c@@page}{\ifodd\c@@page\@@tempcnta\m@@ne}{}{}
    \makeatother
    \reversemarginpar

\newcommand{\cname}{$\lambda C$\xspace}
\newcommand{\BEQ}{::=}
\newcommand{\BOR}{\mid}

\newcommand{\real}{\mathbf{loss}}
\newcommand{\unit}{\mathbf{unit}}

\newcommand{\bool}{\mathbf{bool}}

\newcommand{\type}{\!:\!}
\newcommand{\etype}{\,!\,}
\newcommand{\loss}{\mathbf{loss}}

\newcommand{\myG}{\Gamma}
\newcommand{\op}{\mathit{op}}
\newcommand{\Op}{\mathit{Op}}

\newcommand\parr{ \mathit{par} }
\newcommand\inn{ \mathit{in} }
\newcommand\outt{ \mathit{out} }

\usepackage{mathpartir}
\renewcommand\MathparLineskip{\lineskip=0cm}

\newcommand{\Infer}[4][1]{\ensuremath{\inferrule*[narrower=#1,lab={#2}]{#3}{#4}}}

\makeatletter %
\def\arcr{\@arraycr}
\makeatother

\newcommand{\local}[2]{ \ssmathhl{\langle\!\langle} #2 \ssmathhl{\rangle\!\rangle_{#1}}}

\newcommand{\return}{\mathbf{return}}

\newcommand{\kk}{k}

\newcommand{\kl}{l}
\newcommand{\effl}{\ell}
\newcommand{\con}{c}

\newcommand{\wh}[2]{\mathbf{with}\; #1\; \mathbf{handle} \; #2}
\newcommand{\wph}[3]{\mathbf{with}\; #1\; \mathbf{from} \; #2 \; \mathbf{handle} \, #3}

\newcommand{\reset}[1]{\mathbf{reset}\, #1}

\definecolor{sscolor}{HTML}{0c78bf}
\definecolor{epscolor}{HTML}{687c92}
\newcommand{\sshl}[1]{ {\color{sscolor}{\small \textit{#1}}} }

\newcommand{\ssmathhl}[1]{ {\color{sscolor}{#1} } }

\newcommand{\myg}[1][]{{\color{sscolor} \mathrm{g_{#1}}}}
\newcommand{\og}{\overline{\myg}}

\newcommand{\sem}[2][black]{ {\color{#1} \left [\! \left | {\color{black}#2} \right |\!\right ]} }

\newcommand{\myskip}{$\mbox{}$\\[-1em]}

\newcommand{\R}{R}

\newcommand{\ER}[3]{\mathbf{R}_{#1}(#2|#3)}

\newcommand{\eqdef}{=_{\small{\mathrm{def}}}}
\newcommand{\mydotminus}{\mathbin{\ooalign{\hss\raise1ex\hbox{.}\hss\cr
  \mathsurround=0pt$-$}}}
  
\newcommand{\xxrightarrow}[1]{\xrightarrow{\raisebox{-0.6ex}[0ex][0ex]{\scriptsize $#1$}}}

\usepackage{stackengine}
\stackMath

\newcommand{\evalto}[1]{\overset{\ssmathhl{#1}}{\longrightarrow}}

\newcommand{\bevalto}[1]{\overset{\ssmathhl{#1}}{\Longrightarrow}}

\setcopyright{rightsretained}


\title{Handling the Selection Monad (Full Version)}

\author{Gordon Plotkin}
\affiliation{
  \institution{Google DeepMind}
  \city{Mountain View}
  \country{United States}
}
\email{plotkin@google.edu}
\author{Ningning Xie}
\affiliation{
  \institution{Google DeepMind and University of Toronto}
  \city{Toronto}
  \country{Canada}
}
\email{ningningxie@cs.toronto.edu}

\begin{document}
\begin{abstract}
The selection monad on a set consists of selection functions. These select an element from the set, based on
a loss (dually, reward) function giving the loss resulting from a choice of an element. 
Abadi and Plotkin used the monad to model a language with operations 
making choices of computations
taking account of the  loss  that would arise  from each  choice.
However, their choices were optimal, and they asked if they could instead be programmer provided.

In this work, we present a novel design enabling programmers to  do so.
We present a version of algebraic effect handlers
enriched by computational ideas inspired by the selection monad.
Specifically, as well as the usual delimited continuations, our new kind of handlers additionally have access to \textit{choice continuations},
that give the possible future losses.
In this way programmers can write operations implementing optimisation algorithms that are aware of the losses arising from their possible choices.

We give an operational semantics for a higher-order model language \cname,
and establish desirable properties including progress, type soundness, and termination for a subset with a mild hierarchical constraint on allowable operation types.
We give this subset a selection monad denotational semantics, and prove
soundness and adequacy results.
We also present a Haskell implementation and give a variety of programming examples.
\end{abstract}

\maketitle

\section{Introduction}
\label{sec:introduction}

\newcommand{\myread}{\mathit{read}}
\newcommand{\myint}{\mathit{Int}}
\newcommand{\argmax}{\mathsf{argmax}}
\newcommand{\argmin}{\mathsf{argmin}}
\newcommand{\geval}{\mathsf{eval}}

The \textit{selection monad}~\cite{EM2010,EMO2010,EMO2010MSFP,EO2010,EsO2011,EO2015} has been used to explain fundamental phenomena in various
areas of logic, including game theory, proof theory, and computational
interpretations; it has also been used in connection with CPS transformations and with algorithm design~\cite{HED15,HG22}.
The monad has the form $ S(X) = (X \to R) \to X$, where
$R$, typically an ordered set such as the real numbers, can be thought of as a set of \textit{losses}. A computation, meaning an element of
 $S(X)$, is a \textit{selection function} that, given a \textit{loss function} from $X$ to
 $R$, picks
an element of $X$. For example, the well-known selection function $\argmin$ takes a
loss function and returns an element that minimizes its value.\footnote{Dually, we can think of $R$ as a set of \emph{rewards}, and recall the $\argmax$ function that picks a maximising element; the two viewpoints are equivalent in case $R$ has a negation function, as with the reals. Below we talk only of losses.} 
The selection monad can be combined with other \textit{auxiliary} monads $T$ to produce \textit{augmented} selection monads $ S_T(X) = (X \to R) \to T(X)$~\cite{EO2015,AP2019}. This generalization proves useful when combining  the selection monad with additional effects.

The connection
between programming languages with decision-making operations and the selection monad
was investigated by~\citet{AP21,AP23}.
They considered a language for a selection monad augmented by the writer monad $W(X) = R \times X$.
It had  a binary \textit{choice operation} which could choose between two computations based on the losses the two choices  entailed. Losses were reported by a loss operation. (They also considered a probabilistic extension.)

{A key question they left open was how to empower programmers with the ability to 
define their own choice operations, with the choice of computations (i.e., the selection strategy)
again based on the losses the possible choices of computation entail.
} 
Specifically, \citeauthor{AP23}
used $\mathsf{argmax}$ to model binary choice so that selections were optimal (or optimal-in-expectation).
Implementing optimal selection 
implies perfect knowledge of
all available choices and the ability to consistently make the best decision.
As they acknowledge,
this is not at all a reasonable assumption for applications such as, for example, learning algorithms~\cite{ruder2016overview,smartChoices2018,GoodBengCour16} or game-playing, 
since
programmers generally do not have efficient, or perhaps any, access to optimal choices. 
Considerations of this sort motivated their question asking for a mechanism allowing programmers to employ their own selection strategies.

We give one such way, answering their key open question.
To do so, we develop a suitable version of \textit{algebraic effect
  handlers}~\cite{pretnar2013handling}. Algebraic effect handlers  provide a flexible
mechanism for modular programming with user-defined effects. They have
been explored in various languages including OCaml~\cite{KC2021OCaml},
C/C++~\cite{effects2024c,effects2022cpp}, and WebAssembly~\cite{phipps2023continuing}.
Algebraic effect handlers achieve modularity by decoupling syntax from semantics: the
syntax is defined by user-specified effect operations, while their semantics are
determined by handlers.
A handler's operation receives
the operation's \textit{argument} and a \textit{delimited continuation} which captures the evaluation context, i.e., from where the operation is
performed to where its \textit{result} can be used. %
Handlers can manipulate the continuation in various ways.
For instance, a handler may choose to not resume the continuation, effectively implementing an exception mechanism, 
or it could resume the the continuation multiple times enabling backtracking and non-deterministic behaviors.
However, standard algebraic effect handlers cannot handle operations based on a program's loss information,
unless this information is provided explicitly as an argument to the operation or returned as part of the final result of the continuation
--- but both options are restrictive.
Specifically,
when performing an operation, complete loss information may not be immediately available.
Similarly, returning the loss as part of the final result of the continuation would require all continuations (and functions) to return loss-result pairs, which is impractical.

\textit{We propose a novel language design that 
empowers programmers to
 write %
 choice 
operations that can choose  computations 
based on the losses the possible choices of computation entail.
In this way they can employ their own selection strategies.}
Our fundamental insight  is to combine 
algebraic effect handlers with computational ideas inspired by the selection monad.
We achieve this by providing effect handlers with   \textit{choice
  continuations} (as well as the usual delimited continuations). 
  {These are a kind of loss continuation that yields the loss  
   arising from an operation's possible result.} %
   Handlers can then define choice operations which compute selections from these choice continuations %
   and use the result for their choice of computation, e.g.,   as an argument to a delimited continuation.
   
    Notably, 
   choice  continuations are not delimited, but have scopes that can be controlled by an additional localising construct
that determines how much loss information is accessible. Such scopes can vary from inside the handler %
to beyond. As in~\citet{AP21},  losses are prescribed using a \textit{loss} effect.

We make several contributions:\\[-1.6em]
\begin{itemize}

\item %
We connect up the
selection monad with effect handlers via \cname, a higher-order model language incorporating the new kind of handlers with their  choice continuations. We present a novel loss-continuation-based operational semantics for the language. 

\item  %
We give \cname a selection-monad-based
denotational semantics  and establish  soundness and adequacy theorems for both big-step and giant-step  operational semantics (Theorems~\ref{thm:sound}, \ref{thm:adequate}, and \ref{thm:giant}).
We thereby both show that  our computational ideas are in accord with the monadic ones which inspired them and establish a theoretical foundation for our
extended effect handlers.

\item  %
We provide a library implementation in Haskell, following the operational semantics,
and present programming examples, demonstrating expressiveness.

\item  %
Our work presents a novel combination of delimited and choice
continuations. Our techniques may extend to other combinations
of different kinds of continuations.
\end{itemize}
The rest of the paper is structured as follows. In \Cref{sec:overview}, we
review the selection monad and algebraic effect handlers, and illustrate our new language design.
In \Cref{sec:source} we present
\cname, our higher-order model language, extending effect handlers with loss
primitives and choice continuations. 
We give  \cname  a  deterministic, progressive, and type-safe
 small-step operational semantics  (\Cref{thm:dps}).
We prove that termination holds for a subset of the language with a hierarchical
constraint on handler interfaces (\Cref{thm:term}); this is needed for the results in 
\Cref{sec:denotational}.
In \Cref{sec:implementation} we give a Haskell library,
and give programming examples illustrating potential applications of the new language design. 
We hope our examples offer fresh insights in the  effect handlers application space.
In \Cref{sec:denotational}, we connect up the selection monad with our
computational contributions. We give the hierarchical part of \cname
a selection-monad-based denotational semantics and establish our soundness and adequacy theorems.
For space reasons, some rules and all proofs are provided in the appendix.
Finally, we discuss related and future work in \Cref{sec:conclusion}.
\section{Overview}
\label{sec:overview}
We first introduce the selection monad, and algebraic effect handlers.
We then present our language design, specifically how  computational ideas inspired by the selection monad are combined with algebraic effect handlers.

\subsection{The Selection Monad} \label{sec:selmon}

While the selection monad $S(X) = (X \to R) \to X$ is available in any cartesian closed category, we
focus on the category of sets. We assume $R$ is a commutative monoid $(\R,+,0)$
(for example, the reals, or a finite product of the reals). This will be needed for the loss operation.
 As we have said, a computation $F \in S(X)$ acts as
a \textit{selection function} taking a \textit{loss function}
$\gamma\in (X \to R)$ and picking an element $F(\gamma) \in X$. The loss associated to  $F \in S(X)$, given a loss function $\gamma\type X \to R$, is defined to be $\ER{}{F}{\gamma} \eqdef \gamma (F(\gamma))$. (We remark that the selection monad is closely connected to the more familiar continuation monad $C(X)   = (X \to R)\to R$. For, given $F$ in $S(X)$,
$\lambda \gamma.\, \ER{}{F}{\gamma}$ is in $C(X)$.)

So, for example, taking $R$ to be the real numbers (with their usual addition), for finite sets $X$, $\argmin_X: (X \to R) \to X$ is an example selection function. Given a loss function $\gamma: X \to R$, $\argmin_X(\gamma)$ is an element $x$ of $X$ minimising $\gamma(x)$ (we assume available some way to choose when there is more than one such element).   Then  $\ER{}{\argmin_X}{\gamma}$ is just the minimum value that $\gamma$ obtains.

To fully specify the selection monad we give its Kleisli triple structure, viz the units $(\eta_S)_X: X \to S(X)$ and the Kleisli extensions $f^{\dagger_S}: S(X) \to S(Y)$ of maps $f\type X \to S(Y)$. (As explained in~\cite{benton2000monads} these correspond, modulo currying, to Haskell's monadic $\return$ and bind ($\mathbf{>\!>\!=}$) operations.) The units are given by: $\eta_S(x) = \lambda \gamma.~x$ (Here, and below, we omit the objects $X$, when writing units.) 
For the Kleisli extension, first associate  a \emph{loss continuation
  transformer} $\tilde{f}\type (Y \to R) \to (X \to R)$ to $f$ by
\[ \tilde{f}(\gamma) = \lambda x\in X.\ER{}{f(x)}{\gamma} \]
where we use $f(x)$ as the $Y$ selection function. (Connecting again to the continuation monad, note that, modulo currying, $\tilde{f}$ has type $X \to C(Y)$.) %
Then  the Kleisli extension is given by:
\[ f^{\dagger_S} (F) = \lambda \gamma \in Y \to R.~{f}~{(} F~{(}
  \tilde{f}~\gamma {)}  {)}~{\gamma}
\]

\noindent
Here the loss function $\gamma$ on $Y$ is transformed into a loss function $\tilde{f}(\gamma)$ on $X$, which
is then used by $F$ to select an element $x = F(\tilde{f}(\gamma))$ of $X$. Finally, $f$ uses $x$ and the original
loss function $\gamma$ to select an element of $Y$. As always, the Kleisli structure determines the monad's functorial action by the formula $S(f) = (\eta_S \circ {f})^{\dagger_S}$, which latter, in this case, is
$\lambda \gamma \in Y\to R.\, f(\gamma \circ f)$. 

Continuing the example, Kleisli extension allows us to solve one-move games with evaluation function
$\geval: X \times Y \to R$. Suppose $f:X \to S(X \times Y)$ is defined by:
\[f(x)(\gamma) = (x,\argmin(\lambda y.\, \gamma(x,y))\]
Then $f^{\dagger_S} (\argmax) (\geval)$ is a minimax pair $(x_0,y_0)$ for $\geval$, with 
$x_0 \in X$  maximising all possible $\geval(x,y)$, and $y_0 \in Y$  minimising all possible $\geval(x_0,y)$.

Turning to augmented monads $S_T(X) = (X \to R) \to T(X)$, as an example, take $T$ to be the writer monad $W(X) = \R \times X$, with $R$ the reals. An example (augmented) selection function is then the ``loss-recording" version of $\argmin$ that sends $\gamma$ to
$(\gamma(\argmin(x)),\argmin(\gamma))$.
The unit of $S_W$ is $\eta_{S_{W}}(x) = \lambda \gamma.~(0, x)$. The loss
associated to $F \in S_W(X)$ and $\gamma\type X \to \R$ is the sum of the loss incurred by $F$ and the loss incurred by the
loss function:
$\ER{W}{F}{\gamma} =  \pi_0(F(\gamma)) + \gamma(\pi_1(F(\gamma)))$.
The Kleisli extension $f^{\dagger_{S_{W}}}:S_W(X) \to S_W(Y)$ for $f:X \to
S_W(Y)$ is then defined as below,  where the losses incurred by $F$ and $f$ are added
up (and  where $\tilde{f}$ is defined similarly to the above):
\[
  \arraycolsep=0pt
  \begin{array}{ll}
f^{\dagger_{S_{W}}} (F) = \lambda \gamma: Y \to R.\;~ &
let~\langle r_1, x  \rangle = {(}F~{(\tilde{f}~\gamma)}{)}
~in~ \\
                                                   &
let~\langle r_2, y  \rangle = {(}{f~x~\gamma}{)} ~
                                                 in~\langle  r_1 + r_2, y \rangle
                                                   \end{array}
\]
The functorial action in this case is: $S_W(f) = \lambda \gamma.\, W(f)(f \circ \gamma)$. In general (see, e.g.,~\citet{AP23}) augmented selection monads $S_T$ are available when $R$ forms a $T$-algebra $\alpha\type T(R) \to R$. In the case of the writer monad $W(X)$ just considered, $\alpha \type  W(R) \to R = +$.

For the denotational semantics of our model language \cname   we use  families $F_\epsilon(W(X))$ of auxiliary monads  and loss sets $F_\epsilon(\R)$, with $R$ the reals, parameterized by multisets $\epsilon$ of certain effects $\ell$, where $F_\epsilon(X)$ is a free algebra monad with signature specified by $\epsilon$. The resulting augmented selection monads are used to give the semantics of \cname programs with effect multisets $\epsilon$.

\subsection{Algebraic Effect Handlers}
\label{sec:overview-effect}
\newcommand\ife[3]{ \mathbf{if}~{#1}~\mathbf{then}~{#2}~\mathbf{else}~{#3} }
\newcommand\true{\mathit{True}}
\newcommand\false{\mathit{False}}
\newcommand{\where}{$\mathit{where}$}
\newcommand{\myor}{\mathit{decide}}

Algebraic effect handlers provide a structured approach to managing effects in
computations.
We give a brief introduction here; for a more in depth account see, 
e.g.~\citet{P15,bauer2015programming}.
Consider a non-deterministic choice
operation, $\myor$, which takes a unit and returns a
$\mathsf{Bool}$ as its result. A computation can invoke operations by providing its argument. 
For example,   the following program performs $\myor$ twice, and returns the
conjunction of the results:\footnote{For clarity, we write $x \leftarrow e_1; e_2$ as syntactic sugar for $(\lambda x.~e_2)~e_1$;
and $e_1; e_2$ for when $x \notin \mathsf{fv}(e_2)$.}
\[
f \triangleq x \leftarrow \myor();~y \leftarrow \myor ();~x~\&\&~y
\]
We can define a handler for $\myor$ to specify its semantics. The handler below
handles $\myor$ by invoking the  continuation $k$ with $\true$ and
$\false$ respectively, and collecting the results:
\noindent
\[
\wh{~\{~ \myor \mapsto \lambda x~k. ~(k~\true) +\!+ ~(k~\false),~
\return \mapsto \lambda x.~[x]~\}
}{f}
\]
Within the $\myor$ clause, $x$ is the operation argument, in this case a unit,
and $k$ represents the captured delimited continuation that takes the operation's result
and resumes the computation from the original call site. The handler
explores both branches of the non-deterministic computation by calling $k$
twice and concatenates the result lists. %
The $\return$ clause applies when a value \ensuremath{x} is returned from the
computation. Here, it simply wraps \ensuremath{x} in a singleton list. The
return clause is optional, as a handler that has no special return behavior can have $\return \mapsto \lambda x.~x$. By applying this handler to $f$, we
effectively explore all possible results of the $\myor$ operation, resulting in
$[\true,\false,\false,\false]$.

As can be seen, effect handlers offer modularity by separating syntax and semantics of effect operations.
However, an effect handler's implementation can only depend on the operation argument and the continuation result,
and it cannot use a program's loss information to make decisions.
In practice, a program's loss may actually depend on the operation result,
and  programs don't always return a loss value as their final output.

\subsection{This Work: Handling the Selection Monad}
\label{sec:overview:this}

This paper introduces a novel language design 
that
integrates algebraic effect handlers with  computational ideas derived from the
selection monad. 
Programmers can write handlers that  make use of loss continuations to make selections.
Our design 
allows programmers to define
custom selection functions.
More broadly,
handlers, while maintaining modularity by only handling operations within their scope, can
now additionally leverage future loss information.

To see how our handler design works, consider the following example program where we write $\loss$ to record a loss value:
\[\arraycolsep=1.4pt
\begin{array}{ll}
    pgm \triangleq & b \leftarrow \myor();\enskip
      i \leftarrow \ife{b}{1}{2};\enskip
      \loss(2 * i);\enskip
      \ife{b}{'a'}{'b'}
\end{array}
\]
The $\loss$ operation is a dedicated
writer effect operation that records a loss value. As it is a writer effect,
multiple $\loss$ operations within such programs will be  aggregated. This
allows for a flexible and modular approach to incorporating loss computations. 

We can handle the choice operation $\myor$ using the loss information; for example,
\[\arraycolsep=1.4pt
\begin{array}{llll}
  & \wh{\{~ \myor \mapsto \lambda x~k~l. & y \leftarrow l~\true; z \leftarrow l~\false;\\
  & & \ife{y <= z}{k~\true}{k~\false}
  ~\}}{pgm} \\
\end{array}
\]
Importantly, as we see in this example, losses are made accessible to handlers through  special
\textbf{\emph{choice continuations}}. These are loss continuations which  associate a loss
to each possible result of an operation.
Concretely, handler operation definitions receive the choice 
continuation  as an additional argument.
The example handler given above compares the losses associated with $\true$ and
$\false$, and resumes computation with the choice (boolean selection)  minimizing the  loss.
Using this handler to handle the previous program,
$b$ will be assigned $\true$, resulting a loss
of $2$ and result $'a'$.
In this case, the handler 
implements $\mathsf{argmin}$, corresponding to \citet{AP21,AP23}.
Importantly, however,
with our design the selection is implemented as a separate handler.
With the delimited continuations and choice continuations available, the handler can implement a variety of selections beyond
$\mathsf{argmax}$, as we will see in \Cref{sec:examples}.

Notably, while the continuation $k$ is delimited, the choice continuation $l$ 
has a useful different scope discipline,
which is delimited by a local construct, and otherwise global.
This allows the
handler to make decisions based on fine-grained control over
choice continuations and losses.
We make use of it to, e.g. restrict choice continuation scopes within while loops.
Further, we do not lose generality: restricting the loss value to be the loss accumulated from the continuation can be implemented as a special case by using local.
From the semantical perspective, as shown by the Adequacy Theorem~(\Cref{thm:adequate}),
the design corresponds to a programmable selection monad. 

\begin{wrapfigure}{r}{0.34\textwidth}
\small
\begin{tabular}{rl}
   \hline 
   \multicolumn{2}{l}{
   \textit{semantics}} \\
   loss function   &  $\gamma,k,l$ \\
   \multicolumn{2}{l}{\textit{syntax}} \\
   loss continuation & $\myg$ \\
   choice continuation     & $l,f_l$ \\
   delimited continuation      & $k,f_k$ \\
   \hline
\end{tabular}
\vspace{3pt}
\caption{Terminology}
\label{fig:term}
\vspace{-5pt}
\end{wrapfigure}%
\Cref{fig:term} presents our terminology and corresponding symbols (ambiguities
are resolved by context).
The loss continuation $\myg$ is the programming manifestations of the loss function $\gamma$
of the selection monad, which takes the result of the program and returns a loss.
On the other hand,
when an operation is handled,
the choice continuation $l$, takes
an operation result and returns a loss.

\section{A Model Calculus}%
\label{sec:source}
In this section we present \cname, our higher-order model language incorporating the new kind of handlers with both choice and delimited continuations. 
We give it an operational semantics, show the standard progress and type safety theorems, and prove termination for a  subset of it subject to a mild restriction permitting no ``effect loops'' in operations.

\newcommand{\Dom}{\mathrm{Dom}}
\newcommand{\lcut}[1]{}
\newcommand{\size}[1]{|#1|}

\newcommand{\NF}{\mathrm{NF}}

\newcommand{\kto}[3]{[#1]#2\etype#3}
\newcommand{\lto}[2]{\{#1\}\!\etype\! #2}
\newcommand{\kapp}[2]{[#1]#2}
\newcommand{\lapp}[2]{\{#1\}#2}
\newcommand{\thensymbol}{\ssmathhl{{\raisebox{0.18ex}{$\scriptstyle \blacktriangleright$}}}}
\newcommand{\smallthensymbol}{\ssmathhl{{\raisebox{0.18ex}{$\scriptstyle \blacktriangleright$}}}}
\newcommand{\then}[3]{#2\, {\color{sscolor} \,\raisebox{0.18ex}{$\scriptstyle \blacktriangleright^{}$}\, #3}}
\newcommand{\nat}{\mathbf{nat}}
\newcommand{\mylist}{\mathbf{list}}

\newcommand{\zero}{\mathbf{zero}}
\newcommand{\mysucc}[1]{\mathbf{succ}(#1)}

\newcommand{\nil}{\mathbf{nil}}

\newcommand{\cons}[2]{\mathbf{cons}(#1,#2)}

\subsection{Syntax}
\newcommand{\Eff}{\Sigma}
\newcommand{\inl}[2]{\mathbf{inl}_{#1,#2}}
\newcommand{\inr}[2]{\mathbf{inr}_{#1,#2}}
\newcommand{\mycases}[5]{\mathbf{cases}\, #1\, \mathbf{of}\, #2.\, #3 \talloblong #4.\, #5}
\newcommand{\iter}[3]{\mathbf{iter}(#1,#2,#3)}
\newcommand{\fold}[3]{\mathbf{fold}(#1,#2,#3)}

The syntax of types and effects is given in \Cref{fig:types}.
Types are ranged over by $\sigma$ and $\tau$.
We also write $\parr, \inn, \outt$ for types when talking about parameter or operation types.

We assume
available  a set of basic types $b$ (including $\real$), and
a set of effect labels $\ell$. We take \emph{effects} $\epsilon$ to be multisets  of effect labels; we use juxtaposition $\epsilon\epsilon'$ for multiset union and write $\epsilon \subseteq \epsilon'$ for sub-multiset.
As well as basic types, types include product types $(\sigma_1,\dots,\sigma_n)$,
 sum types $(\sigma + \tau)$, natural numbers $\nat$, and lists $\mylist(\sigma)$
 for iterations and folds (two examples of simple inductive types), 
and function types $(\sigma
\to \tau\etype \epsilon)$, with argument type $\sigma$, 
and result type $\tau$ and effect $\epsilon$.

 We further assume available a
\emph{signature} $\Sigma$ of effect label typings $\ell\type \Op(\ell)$
associating effect labels $\effl$ to finite non-empty sets of
\emph{$\effl$-operations} $\op$ (with disjoint sets associated to different
effect labels). Our language is patterned after the Koka language~\cite{leijen2014koka} with its grouping of operations $\op$ into effects $\ell$.
Each $\op \in \Op(\ell)$ is typed $\op \type \outt \to \inn$;
we often write $\op\type \outt \xxrightarrow{\ell} \inn$ for %
$\op \in \Op(\effl)$
\footnote{One might rather  have expected $\op \type \inn \to \outt$. The idea here is that an operation is an effect: an element of $out$ is sent to start the effect, then the operation returns  an element of $in$ to continue the computation. It may help to think of, e.g., I/O effects, where, with the present convention, output operations  have type $\outt \to ()$ and input operations have type $() \to \inn$.}.

Expressions $e$ and handlers $h$ are given in \Cref{fig:syntax}, where
$x,y,p,\kk,\kl\dots$ range over variables. %
Note that expressions include loss continuation expressions $\myg$.
We make use of 
standard $\lambda$-calculus
abbreviations, for example $\lambda^{\epsilon} (x,y)\type (\sigma,\tau).\, e$
for functions of pairs.

\begin{figure}[t]
\[\begin{array}{llcl}
\text{type} & \sigma,\tau & \BEQ & b  \BOR (\sigma_1,\dots,\sigma_n)\; (n\geq 0) \BOR \sigma + \tau
                            \BOR \nat \BOR \mylist(\sigma) \BOR(\sigma \to \tau\etype \epsilon)
                             \\[5pt]
\text{effect} &  \epsilon & \BEQ & \{\} \mid \epsilon \effl
               \qquad \qquad  \Eff      \enskip \BEQ \enskip \{\overline{\effl_i : \Op(\effl_i)} \} 
               \qquad \qquad \Op(\effl) \enskip \BEQ \enskip \{ \overline{\op_i \type \outt_i \to \inn_i} \}\\[5pt]
\end{array}
\]
    \caption{Syntax of types and effects}
    \label{fig:types}

\[\begin{array}{llcl} 
\text{expr}   &  e        & \BEQ & \;  \con\BOR f(e)   \BOR x  \BOR  \lambda^\epsilon x\type \sigma.\, e \BOR e_1 ~ e_2
      \\%
      &&&\BOR  (e_1,\ldots,e_n) \BOR e.i \\
      &&&\BOR  \inl{\sigma}{\tau}(e) \BOR \inr{\sigma}{\tau}(e)\BOR \mycases{e}{x_1\type \sigma_1}{e_1}{x_2\type \sigma_2}{e_2}\\
      &&&\BOR \zero \BOR \mysucc{e} \BOR\iter{e_1}{e_2}{e_3} \\
      &&&\BOR \nil_\sigma \BOR \cons{e_1}{e_2} \BOR \fold{e_1}{e_2}{e_3}\\
      &&&\BOR  \op(e) \BOR \loss(e) \BOR \wph{h}{e_1}{e_2} \\
      &&&  \BOR \then{\epsilon}{e_1}{\lambda^{\epsilon} x\type \sigma. \, e_2} 
           \BOR   \glocal{\epsilon}{e}{\myg} \BOR \reset{e}
    \\[5pt]

\text{loss cont exp} & \myg & \BEQ  &  \ssmathhl{\lambda^\epsilon x\type \sigma.\, 0} \BOR 
    \ssmathhl{\lambda^{\epsilon} x\type \sigma.\, \then{\epsilon}{e}{\myg}}
    \\[5pt]

\text{handler} & h & \BEQ & \left \{\begin{array}{l}
                     \op_1 \mapsto \lambda^\epsilon z\type (\parr, \outt_1,
                     (\parr,\inn_1) \to \real \etype \epsilon,
                     (\parr,\inn_1) \to \sigma' \etype \epsilon).~ e_1,\dots, \\
                     \op_n \mapsto \lambda^\epsilon z\type (\parr, \outt_n,
                     (\parr,\inn_n) \to \real \etype \epsilon,
                     (\parr,\inn_n) \to \sigma' \etype \epsilon).~ e_n,\\
                   \return \mapsto \lambda^\epsilon z\type (\parr,  \sigma).\, e
                            \end{array}\right \} \\
                            &&& \hspace{140pt} (\op_1,\dots, op_n \mbox{ enumerates some } \Op(\ell)) \\[5pt]

 \end{array}\]
\caption{Syntax of expressions and handlers}
\label{fig:syntax}
\end{figure}

Expressions include constants $c$ and applications of basic functions $f$. 
Abstractions $ \lambda^\epsilon x\type \sigma.\, e$  are explicitly typed, and  annotated
with their result effect $\epsilon$.
We support \textit{parameterized handlers} \cite{plotkin2009handlers},
which generalize effect handlers by keeping a local handler parameter that can
be updated during resumption. Having parameterized handlers is not necessary,
but is convenient when implementing stateful effects. The parameterized handler
expression $\wph{h}{e_1}{e_2}$ handles the computation $e_2$ using
handler $h$, whose parameter has initial value $e_1$. A program can perform an
operation $\op(e)$ by passing the operation $\op$ an argument $e$. The expression $\loss(e)$
invokes the writer effect operation $\loss$, adding a loss $e$. %
Note that, unlike other operations, it is a built-in effect not associated to
any effect label and so cannot be handled;
it can however be used in handlers,
e.g., to define variant loss operations.

The expression $\then{\epsilonk}{e_1}{(\lambda^\epsilon x\type \sigma.\,e_2)}$ is 
 used to build \textit{\textbf{loss continuations}} 
 $\myg$; these form 
a subset of expressions. Loss continuations $\myg$ begin with the
zero loss continuation 
$\ssmathhl{0_{\sigma,\epsilon}} \eqdef \ssmathhl{\lambda^\epsilon x\type\sigma.\, 0}$,
and get extended by
$\ssmathhl{\lambda^{\epsilon} x\type \sigma.\, \then{\epsilon}{e}{\myg}}$ 
(we assume $\thensymbol$ binds more tightly than $\lambda$).
Intuitively, $(\thensymbol)$ (pronounced as "then") accumulates losses: it first evaluates $e_1$,
collects the loss, and passes the evaluation result as $x$ to $e_2$. %
The expression $\glocal{\epsilon}{e}{\myg}$ localises loss continuations to 
expressions $e$ by executing them with loss continuation $\myg$; in contrast, the expression $\reset{e}$ localises losses to $e$, preventing them from passing 
outside $\reset{e}$.
To impose both forms of localisation  the two constructs can be combined, and 
we write $\sglocal{\epsilon}{e}{\myg}$ for $\reset{\glocal{\epsilon}{e}{\myg}}$.
While the language supports the most general local construct  $\glocal{\epsilon}{e}{\myg}$ that takes a $\myg$,
we find that $\glocal{\epsilon}{e}{\ssmathhl{0_{\sigma,\epsilon}}}$,
which localises the loss with respect to the zero continuation, is sufficient for our examples (\Cref{sec:examples}).
Thus, $(\thensymbol)$ and loss continuations do not necessarily need to be part of the user-facing syntax.

A handler $h$ includes a list of operation definitions and a return definition;
it \emph{handles $\effl$} if this list enumerates $\Op(\effl)$ and has \emph{result effect} $\epsilon$ if the abstractions in the definitions have result effect $\epsilon$.
Operations $\op$ takes
a parameter, an operation argument, a choice  continuation $l$ (following our design), and a delimited continuation $k$.
The choice continuation is the key innovation of this calculus. Both
continuations take a potentially updated parameter and the operation result.
Note that $l$ returns $\real$, while $k$ returns $\sigma'$, and the two
continuations are decorated by the effects $\epsilon$ they may cause. Finally, the
 return clause takes the final parameter and the computation
result of type $\sigma$. 

\textit{Remark.}
For space reasons, in the rest of the paper we focus on a subset of the language excluding sum types, natural numbers, and lists.
The full language is detailed in the appendix.

\subsection{Typing Rules} \label{sec:typingrules}

\begin{figure*}
\renewcommand\MathparLineskip{\lineskip=0cm}

\renewcommand\ottaltinferrule[4]{\inferrule*[narrower=1,lab=#1,#2] {#3} {#4}}
\small{\drules[]{$\myG\vdash e\type \sigma \etype  \epsilon $}{Typing Expressions}
{var,const,fun
  ,prd,prj
  ,abs,app
  ,op, loss, handle
}
\vspace{-5pt}
\renewcommand\ottaltinferrule[4]{\inferrule*[narrower=0.4,lab=#1,#2] {#3} {#4}}
\drule{then} \drule{glocal} \drule{reset}

\renewcommand\ottaltinferrule[4]{\inferrule*[narrower=1,lab=#1,#2] {#3} {#4}}
\drules[]{$\myG \vdash h\type \parr,  \sigma \etype \epsilon\effl \Rightarrow \sigma'\etype\epsilon  $}{Typing Handlers}
{handler}}
\caption{Typing rules for \cname}
\label{typing-rules}
\end{figure*}

\Cref{typing-rules} presents the typing rules. As usual, environments $\myG$ are 
finite sets of bindings $x\type \sigma$ of types  to variables, with no variable bound twice
(equivalently, functions from a finite set $\Dom(\myG)$ of variables to types). 
The judgment %
$\myG\vdash e\type \sigma \etype \epsilon$ is that under the context $\myG$,
the expression $e$ has type $\sigma$ and may produce effects in $\epsilon$. %
The type $\sigma$ is determined, due
to the type and effect annotations in the syntax.
When $\myG$ is empty, we may  write $e\type \sigma\etype \epsilon$; we may also write 
$\myG \vdash e\type \sigma$ or $e\type \sigma$ to show the judgments hold for some $\epsilon$. 
The judgment $\myG \vdash h\type \parr, \sigma\etype \epsilon\effl \Rightarrow
\sigma'\etype\epsilon$ is that  $h$ takes a parameter of type
$\parr$ and a computation of $\sigma$ and returns a result of type $\sigma'$,
producing one less effect $\effl$; all of $\parr$, $\sigma$, $\epsilon$,
$\effl$ and $\sigma'$ are determined. True judgments have unique derivations.

The typing rules are mostly standard for  effect handler calculi. For \rref{const,fun}, we assume
available the types of constants $c\type b$, including $r\type \real$, for all $r\in\R$,  and primitive functions $f\type \sigma \to \tau$ (with $\sigma$ and $\tau$ first order), including $+\type (\real,\real) \to \real$ (which we will write infix). Values can have any effect (\rref{const,var,prd,abs}).
In particular, in  \rref{abs}, the function body has the annotated type $\epsilon$ but
the abstraction can have any effect $\epsilon'$.  
In \rref{app}, we check that the argument has the expected type, and that the
effects of the function, the function body, and the argument match.\footnote{Our type-theoretical formalism does not enjoy sub-effecting,
similar to row-based effect style~\cite{HL16,L17}.
Semantically too there is an obstacle.
For a sub-effecting  \rref{app}, the outer $\epsilon$ should be any super-effects of  the inner one. However, a semantics using a monad family $T_\epsilon$, then needs  covariance in the $\epsilon$ and our selection monad family is not. On the other hand, rules \rref{then} and \rref{glocal} employ sub-effecting---needed for the operational semantics, and holding in our denotational semantics. (See the discussions in \Cref{sec:source:op,sec:denotational}.)}
\Rref{prd,prj} are self-explanatory.

When performing an operation (\rref{op}), the global context tells us that
the operation takes a type $\outt$ and has return type $in$. Then we check that
the operation argument has type $\outt$, and  the return type is $in$.
Moreover, we need to make sure that the result effect $\epsilon$ includes the
effect label $\effl$. For the loss operation (\rref{loss}), the operation takes
a loss and has unit return type.

\Rref{handle} takes care of handling. The rule checks that 
$e_1$ has the correct parameter type, and the computation
being handled has type $\sigma$, while the handler $h$ takes a computation in
$\sigma$ and has return type $\sigma'$, and thus the final result has type $\sigma'$.
\Rref{then} checks the loss calculation expression. Note that $\lambda^{\epsilon_2} x\type \sigma. \, e_2$ may produce fewer 
effects than than the expression $e$.
\Rref{glocal} checks loss-continuation-localized computations. Note again that the loss continuation may produce fewer effects than the localized computation; further there is an ``effect conversion" from the effects of that computation to the effects of the whole localized computation. These variations in effects are needed to account for the loss continuations built up by the operational semantics (see the discussion of rule (F) below); they also provide programming flexibility.

Finally, the \rref{handler} type-checks handler definitions. A handler for
$\effl$ handles all operations in $\Op(\effl)$. For each operation $\op_i :
\outt_i \to \inn_i$, the corresponding clause $e_i$ takes a parameter of type $\parr$,
an operation argument of type  $\outt_i$, a choice  continuation of type  $((\parr, in_i) \to
\real\etype \epsilon)$
and a delimited continuation of type  $((\parr, in_i)\to \sigma'\etype \epsilon)$,
and has return type $\sigma'$. The return clause $e$ takes a parameter of type $\parr$ and the
computation result of type $\sigma$, and has return type $\sigma'$. Because of the
return clause, the handler takes a $\sigma$-computation and returns a
$\sigma'$-computation.

\subsection{Operational Semantics}
\label{sec:source:op}
 \newcommand{\myw}{w}

 \newcommand\beff[1]{\lceil #1  \rceil}
\newcommand\bop{\mathsf{bop}}
 \newcommand{\hop}{\mathsf{h}_{\mathsf{op}}}
 \newcommand{\heff}{\mathsf{h}_{\mathsf{eff}}}

\renewcommand\ottaltinferrule[4]{\inferrule*[narrower=1,right=#1,#2] {#3} {#4}}

In this section, we present the small-step operational semantics of \cname.
Our rules axiomatize a judgment $\myg \vdash_\epsilon e \evalto{r} e'$ that 
under the loss continuation $\myg$, $e$ makes a transition to  $e'$, producing loss $\ssmathhl{r}\in \R$.
The decorative $\epsilon$ provides auxiliary information needed to construct loss continuations.
Here $\myg$ produces the loss caused by  the rest of
the program, given the result of executing $e$.

\begin{figure}
\centering
\small{\begin{tabular}{lrcl}

   value & $v $ & $\BEQ $ & $ x \BOR c     \BOR (v_1,\dots,v_n ) \BOR \lambda^\epsilon x\type \sigma.\, e$
      \\

   regular frame & $F$ & $\BEQ$  & $f(\square) \BOR (v_1,\dots, v_k, \square, e_{k+2},\dots, e_n) \BOR \square.i \BOR \square~e \BOR v~\square$\\
         && $\BOR$& $ \op(\square)  \BOR \loss(\square) \BOR \wph{h}{\square}{e}$  \\
    special frame &$S$ & $\BEQ$ & $\wph{h}{v}{\square}
    \BOR \then{\epsilon}{\square}{(\lambda^{\epsilon} x\type \sigma.e)}  
    \BOR \glocal{\epsilon}{\square}{\myg}
    \BOR \reset{\square}$\\
    cont context & $K$ & $\BEQ$ & $\square \BOR F[K]  \BOR S[K]$\\
       stuck expr & $u$  & $\BEQ$  & $ K[\op(v)]\quad (\op \notin \hop(K)) $ \\%
   terminal expr & $w$  & $\BEQ$  & $v \BOR u$ \\
   redex & $R$ & $\BEQ$ &$ f(v) \BOR v.i \BOR v_1~v_2  \BOR \loss(v)  $ \\
                    && $\BOR$ & $\wph{h}{v_1}{K[\op(v_2)]} \quad (\op \notin \hop(K), \op \in h) $ \\
                    && $\BOR$ & $ \wph{h}{v_1}{v_2} $ \\
          && $\BOR$ & $  %
                            \then{\epsilonk}{v}{\lambda^\epsilon x\type \sigma.e_1}
                            \BOR \glocal{\epsilon}{v}{\myg}
                            \BOR \reset{v}$
    \\
\end{tabular}}
\caption{Syntactical classes used for operational semantics}
\label{fig:value}
  
\small{\[
\begin{array}{llcll}
(R1)\; \qquad &
\myg \vdash_\epsilon f(v) & \evalto{0} & v' & (f(v) \to v') \\
(R2) &
\myg \vdash_\epsilon  (v_1,\dots,v_n).i & \evalto{0} & v_i \\
(R3) & 
\myg \vdash_\epsilon (\lambda^\epsilon x\type \sigma.\, e)~v  & \evalto{0} & e[v/x] \\
(R4) & \myg \vdash_\epsilon  \loss(r) &  \evalto{r} &  () \\

\raisebox{10pt}{$(R5)$} & \multicolumn{4}{c}{ \inferrule*[narrower=0.7]{  
\op \notin \hop(K) \\  \op \mapsto v_o \in h\\ v_1:par \\\op\type \outt \xrightarrow{\ell} \inn \\ 
 \mbox{$h$ has effect $\epsilon$} \\
 f_k = \lambda^\epsilon(p,y)\type (par,in).\,\glocal{\epsilon}{\wph{h}{p}{K[y]}}{\myg}\\
 f_l = \lambda^\epsilon (p,y)\type (par,in).\,\then{\epsilon}{(\wph{h}{p}{K[y]})}{\myg}\\
 }
 {\myg \vdash_\epsilon \wph{h}{v_1}{K[\op(v_2)]}  \evalto{0} v_o(v_1,v_2,f_l,f_k)} }  \\[10pt]
 (R6) &
\myg \vdash_\epsilon  \wph{h}{v_1}{v_2} & \evalto{0} & v_r(v_1,v_2) & (\return \mapsto v_r \in h) \\
 (R7) &  \myg \vdash_{\epsilon} \then{\epsilon_1}{v}{\lambda^{\epsilon_1} x\type \sigma. e}
  &\evalto{0} & \glocal{\epsilon_1}{e[v/x]}{\lambda^{\epsilon_1} x: \sigma.\, 0}  \\
  
 (R8) & \myg \vdash_{\epsilon} \glocal{\epsilon_1}{v}{\myg_1} & \evalto{0} & v \\
  
  (R9) & \myg \vdash_{\epsilon} \reset{v} & \evalto{0} & v

\\[10pt]

(F) & \multicolumn{4}{c}{\inferrule{
                     \ssmathhl{\lambda^\epsilon x\type \tau.\, \then{\epsilon}{F[x]}{\myg}}
                     \vdash_\epsilon e \evalto{r} e' }
{\myg \vdash_\epsilon F[e] \evalto{r} F[e']}} \\ [20pt]

\raisebox{10pt}{$(S1)$} &

\multicolumn{4}{c}{\inferrule*[narrower=0.7]{\mbox{$h$ has effect $\epsilon$} \\ 
\mbox{$h$ handles $\ell$}\\ \return \mapsto v_r \in h \\ v_r\type (\parr,\sigma) \to \sigma'\etype \epsilon \\
                     \lambda^\epsilon x\type \sigma.\, (\then{\epsilon}{v_r(v,x)}{\myg}) \vdash_{\epsilon\ell} e \evalto{r} e' }
{\myg \vdash_\epsilon \wph{h}{v}{e} \evalto{r} \wph{h}{v}{e'}}} \\[15pt]

\raisebox{0pt}{$(S2)$} & \multicolumn{4}{c}{\inferrule{\myg[1] \vdash_\epsilon e \evalto{r} e'}
         {\myg \vdash_\epsilon  (\then{\epsilon}{e}{\myg_1}) \evalto{\; 0 \;} r +  
            (\then{\epsilon_1}{e'}{\myg_1})}} \\[15pt]

\raisebox{0pt}{$(S3)$} & \multicolumn{4}{c}{\inferrule{ 
\myg[1] \vdash_{\epsilon_1} e \evalto{r} e'}
{\myg \vdash_{\epsilon} \glocal{\epsilon_1}{e}{\myg[1]}\evalto{r} \glocal{\epsilon_1}{e'}{\myg[1]}}}
\\[15pt]

\raisebox{0pt}{$(S4)$} & \multicolumn{4}{c}{\inferrule{\myg \vdash_\epsilon e \evalto{r} e'}
         {\myg \vdash_\epsilon  \reset{e} \evalto{\; 0 \;} \reset{e'}}} \\
\end{array}
\]}

\caption{Small-step operational semantics rules}
\label{fig:redexrule}
\end{figure}

Before giving the rules we need some syntax to cover values, contexts, stuck expressions, and redexes.
\Cref{fig:value} presents the syntactical classes needed for the operational semantics. %
\emph{Values} $v$ include variables, constants, value tuples, and lambda expressions.
(Note that values can be typed with any effect, and we generally just write $\myG\vdash v\type  \sigma$ or $\vdash v \type \sigma$.) 
\emph{Continuation contexts} $K$ are either a hole $\square$,
or a regular frame $F$ or  special frame $S$ followed by a continuation context. 
(We distinguish between regular frames and special frames,since, as we will see,
they extend  loss continuations differently.)
We write $K[e]$ for the expression obtained by filling the hole in $K$ with $e$.

\emph{Stuck expressions} $u = K[\op(v)]$ are  operation invocations $\op(v)$ that cannot be handled by
handlers in  continuation contexts $K$;
We write $\heff(K)$ for the multiset of effect labels that $K$ handles;
it is defined inductively with main clause $\heff(\wph{h}{v}{K'}) = \heff(K')\ell$, where $h$ handles $\ell$; we further set  $\hop(K) = \{\op \in \Op(\ell)| \ell \in \heff(K)\}$, the set of operations  handled by $K$.
\textit{Terminal expressions} $w$ are values or stuck expressions; they  cannot reduce; in contrast, (closed) \emph{redexes} $R$ are expressions that do.

Expressions can be analysed uniquely:
\begin{lemma}[Expression analysis] \label{eal} Every expression has exactly one of the following five forms:
(1) a value $v$ (for a unique $v$),
(2) a stuck expression $K[\op(v)]$ (for  unique $K$, $\op$, and $v$), 
(3) a redex $R$ (for a unique $R$),
(4) $F[e]$ (for unique $F$ and $e$, with $e$ not a value or stuck), or
(5) $S[e]$ (for  unique $S$ and $e$, with $e$ not a value or stuck).
\end{lemma}

\paragraph{\textbf{Small-step operational semantics}}

\Cref{fig:redexrule} presents our small-step operation semantics rules 
for the  judgment $\myg \vdash_\epsilon e \evalto{r} e'$.
Program execution starts with the zero  loss continuation. Further loss  continuations are progressively built up during execution, in order for subprograms to 
pass their results to their enclosing contexts and so on to the program's loss continuation. (All this is quite analogous to how one computes with ordinary continuations.)

\textit{Redexes.}
Many  redex rules do not use the loss continuation, and produce a zero loss.
For (R1), we assume available deterministic total reductions for primitive functions.
In (R4), $\loss(r)$ produces a loss $r$ returning $()$.
Rules (R8) and (R9) are natural, expressing that the computation  terminates once a value is reached.

Operations get handled by  rule (R5).
The handler operation clause is applied  to  parameter
$v_1$, operation argument $v_2$,  \textbf{delimited continuation} $f_k$,
and \textbf{\textit{choice continuation}} 
$f_l$.
The continuation $f_k$ takes the handler parameter and the operation result to be used
when resuming,  localised to the current loss continuation $\myg$ when called,
since the continuation is captured under the loss continuation $\myg$.
The choice  continuation $f_l$ is built from the standard handler continuation and the current loss continuation using the $\thensymbol$ construct; 
it  therefore has access to
all the losses resulting from executing the operation, and so the handler can make decisions based on
that information. 
In  rule (R6), when a value returns from a handler,
the return clause from the handler applies,
taking as arguments the current handler parameter and the value.
In rule (R7), we evaluate $\then{\epsilon_2}{v}{\lambda^{\epsilon_2} x\type \sigma.\; e}$
by substituting $v$ for $x$  in $e$
and localising the resulting expression to the zero loss continuation, 
as the purpose of $\thensymbol$ is to calculate a loss independently of the current loss continuation.

\textit{Regular frames.}
Rule (F) evaluates expressions $e$ inside regular frames 
$F$.\footnote{The use of frames is a way to present the administrative rules of small-step semantics via a single rule; the idea seems to be folklore. We could as well have used evaluation contexts~\cite{FH92}.}
Importantly,
we adjust the loss continuation $\myg$ to $\ssmathhl{\lambda^\epsilon x\type \tau.\, \then{\epsilon}{F[x]}{\myg}}$
when evaluating $e$. This is because the loss continuation of $e$ is to pass its result $v$ to the context $F[v]$ whose value is then passed to its enclosing loss context, meanwhile accumulating incurred losses. 
Note that, to apply this rule, the decorative $\epsilon$ is used. 
This is the only rule that does so, and it is needed to make the rule deterministic.

\textit{Special frames.}
Lastly, (S1)-(S4) evaluate inside special frames. These adjust the
loss continuations or losses differently from rule (F).
The loss continuation in (S1) uses the return clause from the handler,
since after $e$ is evaluated with the aid of the handler the final result is passed to the return function.
Also,
rule (S1) is where the effect associated with the judgment changes.
It changes from $\epsilon$ to $\epsilon \ell$, when evaluating $e$.
Thus, rule (R5) builds up a loss continuation by combining an expression with a loss continuation with fewer effects.
As a result, sub-effecting is needed in
the typing \rref{then} to ensure type safety of the operational semantics.
Similarly, sub-effecting is used in the typing \rref{glocal},
since rule (R7) further wraps the body of a loss continuation within a local construct.
In rule (S2), the current loss continuation is not imported inside the ``then" construct. Instead $e$ is
evaluated relative to the loss continuation $\ssmathhl{\myg_1}$; moreover, the loss $r$ produced during the evaluation is added to the final result.
In rule (S3) the loss continuation also changes, as with the local construct; note that the loss created by $e$ is exported.
Finally, in rule (S4) the loss continuation does not change, but the loss created by $e$  is not exported.

The standard results hold for our operational semantics:
\begin{theorem} \label{thm:dps}
\myskip
\begin{enumerate}
    \item (\textit{Terminal expressions}) If  $e$ is terminal, then it can make no transition, i.e., $\myg \vdash_{\epsilon} e \xrightarrow{r} e'$ holds for no $\myg, r, \epsilon$, $e'$. 
   \item (\textit{Determinism}) If  $~\myg \vdash_{\epsilon} e \evalto{r} e'$ and  $\myg \vdash_{\epsilon} e \evalto{r'} e''$ then $r = r'$ and $e' = e''$.
   \item (\textit{Progress})  If $e\type \sigma\etype \epsilon_1$ is non-terminal, then
     $\myg \vdash_{\epsilon_1} e \evalto{r} e'$  holds for some $r$ and $e'$ for any $~\myg\type \sigma \to \real\etype \epsilon_2$ with $\epsilon_2 \subseteq \epsilon_1$.
   \item (\textit{Type safety})  If $~\myg\type \sigma \to \real\etype \epsilon_2$,  
        $\myg \vdash_{\epsilon_1} e \evalto{r} e'$, with $\epsilon_2 \subseteq \epsilon_1$,
       and $e\type \sigma\etype \epsilon_1$ then $e'\type \sigma\etype \epsilon_1$.
\end{enumerate}

\end{theorem} \label{thm:dss}

\paragraph{\textbf{Example}}

We consider the example program from \Cref{sec:overview:this} to demonstrate the operational semantics:
\[\arraycolsep=1.4pt
\small
\begin{array}{rll}
    pgm & \triangleq & b \leftarrow \myor();\enskip
      i \leftarrow \ife{b}{1}{2};\enskip
      \loss(2 * i);\enskip
      \ife{b}{'a'}{'b'}\\
    h & \triangleq & 
    \{~ \myor \mapsto \lambda x~k~l. ~y \leftarrow l~\true; z \leftarrow l~\false;\ife{y <= z}{k~\true}{k~\false} ~\}
\end{array}
\]
We evaluate the program under the zero continuation,
and omit handler parameters. 
We write $C$ for the character type,
and $B$ for the boolean type.
First, the operation is handled (rule (R5)):

\vspace{-15pt}

{\small
\begin{equation}
   \ssmathhl{0_{C,\{\}}} \vdash_{\{\}} \wh{h}{pgm} 
 \evalto{0}    (y \leftarrow f_l~\true; z \leftarrow f_l~\false;\ife{y <= z}{f_k~\true}{f_k~\false}) \\[2pt]
\end{equation}
\vspace{-15pt}

\noindent
\begin{tabular}{l}
 where $ f_k = \lambda b: B. ~\glocal{\{\}}{ \wh{h}{ (  i \leftarrow \ife{b}{1}{2};\enskip
      \loss(2 * i);\enskip
      \ife{b}{'a'}{'b'})}}{\ssmathhl{0_{C,\{\}}}}$\\
  \qquad\enskip\enskip   $f_l = \lambda b:B.\,\then{\{\}}{( \wh{h}{ (i \leftarrow \ife{b}{1}{2};\enskip
      \loss(2 * i);\enskip
      \ife{b}{'a'}{'b'}))}}{0_{C,\{\}}}$
\end{tabular}
}

\noindent
We then evaluate $(f_l~\true)$. Rule (F)
changes the loss continuation to
$\myg \triangleq 
\ssmathhl{\lambda^\epsilon y\type \tau.\, {(z \leftarrow f_l~\false;}} $
$\ssmathhl{\ife{y <= z}{f_k~\true}{f_k~\false}) ~ \thensymbol ~ { 0_{C,\{\}} }}$. Rule (R3) reduces the application and produces a 0 loss.
Now we evaluate the following expression under $\myg$:
{\small
\begin{equation}
 \myg \vdash_{\{\}}  \then{\{\}}{( \wh{h}{ (i \leftarrow \ife{\true}{1}{2};\enskip
      \loss(2 * i);\enskip
      \ife{\true}{'a'}{'b'}))}}{0_{C,\{\}}}
\end{equation}}%
Importantly, the $\thensymbol$ operator disregards $\myg$, and evaluates
the expression under $\ssmathhl{0_{C,\{\}}}$ (rule (S2)).
This behavior ensures that continuations are consistently evaluated under the loss continuation they are captured at.
Without it, evaluating continuations would yield different results based on how continuations are used within the handler, which is undesirable. Then, evaluating the program
{\small
\begin{equation}
 \ssmathhl{0_{C,\{\}}} \vdash_{\{\}} ( \wh{h}{ (i \leftarrow \ife{\true}{1}{2};\enskip
      \loss(2 * i);\enskip
      \ife{\true}{'a'}{'b'}))}
\end{equation}
}%
produces a loss $2$ and a value $'a'$. According to rule (S2),
the loss is added to the result of $(\then{}{'a'}{ \ssmathhl{0_{C,\{\}}}})$,
producing the result $2 + 0 = 2$.
Substituting $2$ for $y$ in expression (1), we get
{
\small
\begin{equation}
 (z \leftarrow f_l~\false;\ife{2 <= z}{f_k~\true}{f_k~\false})  \\[2pt]
\end{equation}
}%
Similarly, $(f_l~\false)$ evaluates to $4$, and thus the
computation reduces to $(f_k~\true)$.
Continuing the evaluation will produce the final result $'a'$
and the loss $2$.

\paragraph{\textbf{Big-step operational semantics}}

Finally, we define a big-step operational semantics judgment
$\myg \vdash e \xRightarrow{r}w$, that  under  loss continuation $\myg$,
 expression $e$ evaluates to  terminal expression $w$.

\begin{figure}[H]
\centering
\vspace{-15pt}
\begin{mathpar}
\inferrule{ } {\myg \vdash_\epsilon w \bevalto{0} w }
\and
\inferrule{\myg \vdash_\epsilon e_1 \evalto{r} e_2 \\
\myg \vdash_\epsilon e_2 \bevalto{s} w}
{\myg \vdash_\epsilon e \bevalto{r+s} w}
\end{mathpar}
\caption{Big-step operational semantics rules}
\label{fig:bigeval}
\vspace{-10pt}
\end{figure}

\noindent
It follows immediately from Theorem~\ref{thm:dps} that the big-step semantics is deterministic and 
type safe:
\begin{corollary} 
\label{thm:safe}
Given $e\type \sigma \etype  \epsilon$,
and $\myg \type \sigma\to \bool\etype \epsilon'$ with $\epsilon'\subseteq \epsilon$,
there is at most one $r\in \R$ and terminal expression $w$ such that
$\myg \vdash e \bevalto{r} w$ and then $w\type \sigma\etype \epsilon$.
\end{corollary}

\subsection{Termination}
\label{sec:termination}

\newcommand\mo{{moo}}
\newcommand\cow{{cow}}

We establish termination with a suitable well-foundedness assumption on the effects allowed in the input and output types of operations.
We use the termination result to establish adequacy in \Cref{sec:denotational}.
A result of this type for a standard handler calculus appears in~\citet{FK19}.
However they did not have loss continuations which, as we will see, leads to complex computability definitions.

\paragraph{\textbf{Well-foundness of effects}}

Unfortunately, not all effect handler programs terminate.
Adapting from \citet{BP13}, consider an effect $\cow$
with the corresponding handler $h$:
\[ \cow : \{\; \mo : \unit \to (\unit \to \unit \etype \cow)  \;\}  
\qquad
  h = \{\; \mo \mapsto \lambda (p, x, \effl, k).~k~(\lambda^\cow y.~ \mo(())~()) \;\}
\]
Then the program $e \triangleq  \wph{h}{v}{\mo(())~()} $ diverges:
\[
e \longrightarrow \wph{h}{v}{(\lambda^\cow y.~ \mo(())~())~()}  \longrightarrow  e \longrightarrow \ldots
\]

To rule out such programs where effect labels occur inside the input or output types of their operations, for this subsection and \Cref{sec:denotational} we make use of a well-foundedness assumption on effects.
Specifically, we write $e(\epsilon)$ and $e(\sigma)$ for the set of effect labels appearing in $\epsilon$ or $\sigma$. So, for example 
$e(\sigma \to \tau\etype\epsilon) =  e(\sigma) \cup  e(\tau) \cup \{\ell |\ell \in \epsilon\}$. Our  well-foundedness assumption is that there is an ordering $\ell_1,\dots, \ell_n$ of the labels such that:
\[\op\type \outt \xrightarrow{\ell_j} \inn \;\wedge\; \ell_i \in e(\outt) \cup e(\inn) \implies i < j\]
We then define the effect levels of $\epsilon$ and $\sigma$
by:
$l(\epsilon) = \max_i \{i | \ell_i \in e(\epsilon)\}$ and  $l(\sigma) = \max_i \{i | \ell_i \in e(\sigma)\}$.
The size $|\sigma|$ of  types is defined  standardly  (e.g., 
$|\sigma \to \tau\etype\epsilon| =  1 + |\sigma| +  |\tau| + |\epsilon|$).

Our denotational semantics is defined for programs satisfying the assumption. We  remark that the assumption holds for all our programming examples. In the (also terminating) EFF language~\cite{FK19}, the assumption is baked into the language design: operation types  are only well-defined if they can be shown so using only previously well-typed operations.

\paragraph{\textbf{Computability}}

Our proof uses suitable recursively-defined notions of computability, following Tait~\cite{Tait67}. 
We define the following main notions:
\begin{itemize}
\item[-] \textit{computability} of  closed values $v\type \sigma$,
\item[-]  \textit{loss computability} of  closed  loss continuations  $\myg\type \sigma \to \real\etype \epsilon$, and   
\item[-]  \textit{computability} of closed  expressions $e\type\sigma\etype\epsilon$.
\end{itemize}
 We define these notions by the following mutually-recursive clauses. 
They employ two auxiliary notions.
One is an inductively defined  notion of \textit{$G$-computability} of expressions, where $G$ is a set of loss continuations; the other is  a notion of  \emph{R-computability} of real-valued expressions.

\begin{enumerate}
\item \begin{enumerate}
         \item Every constant $c\type b$ of ground type is computable.
         \item A closed value $(v_1,\dots,v_n)\type (\sigma_1,\dots,\sigma_n)$ is computable if every 
                   $v_i\type \sigma_i$ is computable.
              \item A closed value $\lambda^{\epsilon} x\type \sigma.\, e \type \sigma \to 
                    \tau\etype\epsilon$ is computable if, for every computable value $ v\type \sigma$, the 
                    expression $e[v/x]\type \tau\etype\epsilon$ is computable.
                \end{enumerate}
\item The property  of %
{$G$-computability} of  closed  expressions $ e\type \sigma\etype \epsilon$, for a set $G$ of closed  loss continuations of type $ \myg\type \sigma \to \real\etype \epsilon'$  for some $\epsilon' \subseteq \epsilon$,   is the least such property  $P_{\sigma,\epsilon}$ of these expressions such that 
one of the following three possibilities holds:
\begin{enumerate}
\item  $e$ is a  computable closed  value.
\item  $e$ is an operation value $K[\op(v)]$, with $\op\type \outt \xrightarrow{\ell} \inn$, where $v:\outt$ is  a computable closed  value, and where,  for every  computable closed  value $v_1\type \inn$, 
$P_{\sigma,\epsilon}(K[v_1])$ holds.
\item For every $\myg \in G$, if $\myg \vdash_\epsilon e \xrightarrow{r} e'$ then  $P_{\sigma,\epsilon}(e')$ holds.
\end{enumerate} 
\item \begin{enumerate}
                 \item An expression $e:\real\etype \epsilon$ is {R-computable}
iff it is $\{0_{\real,\epsilon}\}$-computable.

\item A closed  loss continuation $\lambda^\epsilon x\type \sigma.\, e\type \sigma \to \real\etype \epsilon$ is {loss computable} if $e[v/x]$ is R-computable for every computable closed  value $v\type \sigma$.
\end{enumerate}
\item A closed  expression $e\type \sigma\etype\epsilon$ is {computable} 
iff it is $G$-computable, where $G$ is the set of closed  loss-computable 
loss continuations $\myg\type \sigma \to \real\etype \epsilon'$, 
for some $\epsilon'\subseteq \epsilon$. 
\end{enumerate} 

The definitions are proper (i.e. the recursions terminate) as can be be seen by suitable measures $m$ defined on the types and effects of values $v\type \sigma$, 
loss continuations $\myg \type \sigma \to \real\etype \epsilon$, and expressions $e\type \sigma\etype \epsilon$; these are pairs of natural numbers, lexicographically ordered, and are given by:
\[m(v) = (l(\sigma), |\sigma|)
\quad
m(\myg) = (l( \sigma) \max l(\epsilon), | \sigma| )
\quad
m(e) = (l( \sigma) \max l(\epsilon), |\sigma|)\]
These measures do not increase in passing from the definition of one notion to another, so every link in the graph of definitional dependencies is non-increasing. The measures also decrease when passing from the value-computability to itself and from $G$-computability to itself. So every loop in the graph contains a decreasing link, and we see that  the various notions are well-defined.
We extend these notions to open expressions in the usual way, via substitution by computable 
closed  values. We prove the following fundamental lemma:

\begin{lemma}[Fundamental Lemma] \label{lem:fundamental}
\myskip
\begin{enumerate}
    \item Every loss continuation $\myG\vdash  \myg\type \sigma \to \real\etype\epsilon $ is loss computable.
\item Every expression $\myG \vdash e\type \sigma\etype \epsilon$ is computable.
\end{enumerate}
\end{lemma}

We can deduce termination from computability. 

\begin{theorem} [Termination] \label{thm:term}
For $e_1\type \sigma \etype  \epsilon$ and $\myg \type \sigma\to \bool\etype \epsilon'$ with $\epsilon'\subseteq \epsilon$, there are no infinite sequences:
$\myg \vdash e_1 \xrightarrow{r_1} e_2 \xrightarrow{r_2} \dots \xrightarrow{r_{n-2}} e_{n-1 }\xrightarrow{r_{n-1}} e_n \dots$
\end{theorem}
Combining this with Theorem~\ref{thm:safe} we obtain:
\begin{theorem} \label{thm:bterm}
For $e\type \sigma \etype  \epsilon$ and $\myg \type \sigma\to \bool\etype \epsilon'$ with $\epsilon'\subseteq \epsilon$ we have $\myg \vdash e \xRightarrow{r} w$ for a unique $r \in \R$ and terminal expression $w$ (and then $w\type \sigma\etype \epsilon$).
\end{theorem}
\noindent
The corollary covers any effect multiset $\epsilon$; when $\epsilon$ is empty, the terminal is a value by the well-typing.

\makeatletter
\@ifundefined{lhs2tex.lhs2tex.sty.read}%
  {\@namedef{lhs2tex.lhs2tex.sty.read}{}%
   \newcommand\SkipToFmtEnd{}%
   \newcommand\EndFmtInput{}%
   \long\def\SkipToFmtEnd#1\EndFmtInput{}%
  }\SkipToFmtEnd

\newcommand\ReadOnlyOnce[1]{\@ifundefined{#1}{\@namedef{#1}{}}\SkipToFmtEnd}
\usepackage{amstext}
\usepackage{amssymb}
\usepackage{stmaryrd}
\DeclareFontFamily{OT1}{cmtex}{}
\DeclareFontShape{OT1}{cmtex}{m}{n}
  {<5><6><7><8>cmtex8
   <9>cmtex9
   <10><10.95><12><14.4><17.28><20.74><24.88>cmtex10}{}
\DeclareFontShape{OT1}{cmtex}{m}{it}
  {<-> ssub * cmtt/m/it}{}
\newcommand{\texfamily}{\fontfamily{cmtex}\selectfont}
\DeclareFontShape{OT1}{cmtt}{bx}{n}
  {<5><6><7><8>cmtt8
   <9>cmbtt9
   <10><10.95><12><14.4><17.28><20.74><24.88>cmbtt10}{}
\DeclareFontShape{OT1}{cmtex}{bx}{n}
  {<-> ssub * cmtt/bx/n}{}
\newcommand{\tex}[1]{\text{\texfamily#1}}	%

\newcommand{\Sp}{\hskip.33334em\relax}

\newcommand{\Conid}[1]{\mathit{#1}}
\newcommand{\Varid}[1]{\mathit{#1}}
\newcommand{\anonymous}{\kern0.06em \vbox{\hrule\@width.5em}}
\newcommand{\plus}{\mathbin{+\!\!\!+}}
\newcommand{\bind}{\mathbin{>\!\!\!>\mkern-6.7mu=}}
\newcommand{\rbind}{\mathbin{=\mkern-6.7mu<\!\!\!<}}%
\newcommand{\sequ}{\mathbin{>\!\!\!>}}
\renewcommand{\leq}{\leqslant}
\renewcommand{\geq}{\geqslant}
\usepackage{polytable}

\@ifundefined{mathindent}%
  {\newdimen\mathindent\mathindent\leftmargini}%
  {}%

\def\resethooks{%
  \global\let\SaveRestoreHook\empty
  \global\let\ColumnHook\empty}
\newcommand*{\savecolumns}[1][default]%
  {\g@addto@macro\SaveRestoreHook{\savecolumns[#1]}}
\newcommand*{\restorecolumns}[1][default]%
  {\g@addto@macro\SaveRestoreHook{\restorecolumns[#1]}}
\newcommand*{\aligncolumn}[2]%
  {\g@addto@macro\ColumnHook{\column{#1}{#2}}}

\resethooks

\newcommand{\onelinecommentchars}{\quad-{}- }
\newcommand{\commentbeginchars}{\enskip\{-}
\newcommand{\commentendchars}{-\}\enskip}

\newcommand{\visiblecomments}{%
  \let\onelinecomment=\onelinecommentchars
  \let\commentbegin=\commentbeginchars
  \let\commentend=\commentendchars}

\newcommand{\invisiblecomments}{%
  \let\onelinecomment=\empty
  \let\commentbegin=\empty
  \let\commentend=\empty}

\visiblecomments

\newlength{\blanklineskip}
\setlength{\blanklineskip}{0.66084ex}

\newcommand{\hsindent}[1]{\quad}%
\let\hspre\empty
\let\hspost\empty
\newcommand{\NB}{\textbf{NB}}
\newcommand{\Todo}[1]{$\langle$\textbf{To do:}~#1$\rangle$}

\EndFmtInput
\makeatother
\ReadOnlyOnce{polycode.fmt}%
\makeatletter

\newcommand{\hsnewpar}[1]%
  {{\parskip=0pt\parindent=0pt\par\vskip #1\noindent}}

\newcommand{\hscodestyle}{}

\newcommand{\sethscode}[1]%
  {\expandafter\let\expandafter\hscode\csname #1\endcsname
   \expandafter\let\expandafter\endhscode\csname end#1\endcsname}

\newenvironment{compathscode}%
  {\par\noindent
   \advance\leftskip\mathindent
   \hscodestyle
   \let\\=\@normalcr
   \let\hspre\(\let\hspost\)%
   \pboxed}%
  {\endpboxed\)%
   \par\noindent
   \ignorespacesafterend}

\newcommand{\compaths}{\sethscode{compathscode}}

\newenvironment{plainhscode}%
  {\hsnewpar\abovedisplayskip
   \advance\leftskip\mathindent
   \hscodestyle
   \let\hspre\(\let\hspost\)%
   \pboxed}%
  {\endpboxed%
   \hsnewpar\belowdisplayskip
   \ignorespacesafterend}

\newenvironment{oldplainhscode}%
  {\hsnewpar\abovedisplayskip
   \advance\leftskip\mathindent
   \hscodestyle
   \let\\=\@normalcr
   \(\pboxed}%
  {\endpboxed\)%
   \hsnewpar\belowdisplayskip
   \ignorespacesafterend}

\newcommand{\plainhs}{\sethscode{plainhscode}}
\newcommand{\oldplainhs}{\sethscode{oldplainhscode}}
\plainhs

\newenvironment{arrayhscode}%
  {\hsnewpar\abovedisplayskip
   \advance\leftskip\mathindent
   \hscodestyle
   \let\\=\@normalcr
   \(\parray}%
  {\endparray\)%
   \hsnewpar\belowdisplayskip
   \ignorespacesafterend}

\newcommand{\arrayhs}{\sethscode{arrayhscode}}

\newenvironment{mathhscode}%
  {\parray}{\endparray}

\newcommand{\mathhs}{\sethscode{mathhscode}}

\newenvironment{texthscode}%
  {\(\parray}{\endparray\)}

\newcommand{\texths}{\sethscode{texthscode}}

\def\codeframewidth{\arrayrulewidth}
\RequirePackage{calc}

\newenvironment{framedhscode}%
  {\parskip=\abovedisplayskip\par\noindent
   \hscodestyle
   \arrayrulewidth=\codeframewidth
   \tabular{@{}|p{\linewidth-2\arraycolsep-2\arrayrulewidth-2pt}|@{}}%
   \hline\framedhslinecorrect\\{-1.5ex}%
   \let\endoflinesave=\\
   \let\\=\@normalcr
   \(\pboxed}%
  {\endpboxed\)%
   \framedhslinecorrect\endoflinesave{.5ex}\hline
   \endtabular
   \parskip=\belowdisplayskip\par\noindent
   \ignorespacesafterend}

\newcommand{\framedhslinecorrect}[2]%
  {#1[#2]}

\newcommand{\framedhs}{\sethscode{framedhscode}}

\newenvironment{inlinehscode}%
  {\(\def\column##1##2{}%
   \let\>\undefined\let\<\undefined\let\\\undefined
   \newcommand\>[1][]{}\newcommand\<[1][]{}\newcommand\\[1][]{}%
   \def\fromto##1##2##3{##3}%
   \def\nextline{}}{\) }%

\newcommand{\inlinehs}{\sethscode{inlinehscode}}

\newenvironment{joincode}%
  {\let\orighscode=\hscode
   \let\origendhscode=\endhscode
   \def\endhscode{\def\hscode{\endgroup\def\@currenvir{hscode}\\}\begingroup}
   \orighscode\def\hscode{\endgroup\def\@currenvir{hscode}}}%
  {\origendhscode
   \global\let\hscode=\orighscode
   \global\let\endhscode=\origendhscode}%

\makeatother
\EndFmtInput
\ReadOnlyOnce{forall.fmt}%
\makeatletter

\let\HaskellResetHook\empty
\newcommand*{\AtHaskellReset}[1]{%
  \g@addto@macro\HaskellResetHook{#1}}
\newcommand*{\HaskellReset}{\HaskellResetHook}

\global\let\hsforallread\empty
\global\let\hsexistsread\empty

\newcommand\hsforall{\global\let\hsdot=\hsperiodonce}
\newcommand\hsexists{\global\let\hsdot=\hsperiodonce}
\newcommand*\hsperiodonce[2]{#2\global\let\hsdot=\hscompose}
\newcommand*\hscompose[2]{#1}

\AtHaskellReset{\global\let\hsdot=\hscompose}

\HaskellReset

\makeatother
\EndFmtInput

\section{Programming with the Selection Monad}
\label{sec:implementation}

Having established the operational semantics of our design, we implemented it as
an effect handler library in Haskell.
In this section
we first present the programming interface, then briefly explain the
embedding, and, lastly, present programming examples.

\subsection{Effect Handler Interface}
\label{sec:interface}

We begin with a simple example to demonstrate the programming interface. An effect is
declared as a datatype with its fields being operations. For example, the
following datatype:
\noindent
\begin{hscode}\SaveRestoreHook
\column{B}{@{}>{\hspre}l<{\hspost}@{}}%
\column{E}{@{}>{\hspre}l<{\hspost}@{}}%
\>[B]{}[\mykeyword{effect}|\;\;\keyword{data}\;\conid{NDet}\mathrel{=}\conid{NDet}\;\{\mskip1.5mu \varid{decide}\mathbin{::}\conid{Op}\;()\;\conid{Bool}\mskip1.5mu\}\mskip1.5mu]{}\<[E]%
\ColumnHook
\end{hscode}\resethooks
\noindent
declares a \ensuremath{\conid{NDet}} effect (\Cref{sec:overview-effect}) with an operation
\ensuremath{\varid{decide}} from \ensuremath{()} to \ensuremath{\conid{Bool}}.
This embedding uses a
Template Haskell interface (similar to \citet{Kammar13action}) to reduce 
burdensome syntax.

We can perform an operation and handle an effectful program as follows.
A handler 
\begin{wrapfigure}{r}{0.49\textwidth}
\vspace{-15pt}
\begin{hscode}\SaveRestoreHook
\column{B}{@{}>{\hspre}l<{\hspost}@{}}%
\column{3}{@{}>{\hspre}l<{\hspost}@{}}%
\column{5}{@{}>{\hspre}l<{\hspost}@{}}%
\column{15}{@{}>{\hspre}l<{\hspost}@{}}%
\column{18}{@{}>{\hspre}l<{\hspost}@{}}%
\column{E}{@{}>{\hspre}l<{\hspost}@{}}%
\>[B]{}\varid{pgm}\mathrel{=}{}\<[E]%
\\
\>[B]{}\hsindent{3}{}\<[3]%
\>[3]{}\mykeyword{handlerRet}\;{}\<[15]%
\>[15]{}(\lambda \varid{x}\to \mykeyword{return}\;[\mskip1.5mu \varid{x}\mskip1.5mu])\;{}\<[E]%
\\
\>[15]{}(\conid{NDet}\;\{\mskip1.5mu \effect{decide}\mathrel{=}\mykeyword{operation}\;(\lambda \varid{x}\;\varid{l}\;\varid{k}\to {}\<[E]%
\\
\>[15]{}\hsindent{3}{}\<[18]%
\>[18]{}(\plus )\mathop{{\langle}{\$}{\rangle}}\varid{k}\;\conid{True}\mathop{{\langle}{*}{\rangle}}\varid{k}\;\conid{False})\mskip1.5mu\})\;\mathbin{\$}{}\<[E]%
\\
\>[3]{}\hsindent{2}{}\<[5]%
\>[5]{}\keyword{do}\;\varid{y}\leftarrow \mykeyword{perform}\;\effect{decide}\;();\mykeyword{return}\;(\varid{not}\;\varid{y}){}\<[E]%
\ColumnHook
\end{hscode}\resethooks
\vspace{-20pt}
\end{wrapfigure}
(\ensuremath{\conid{NDet}\;\{\mskip1.5mu \varid{decide}\mathrel{=}\mykeyword{operation}\;(\mathbin{...})\mskip1.5mu\}}) is simply an instance of the data
type with field \ensuremath{\varid{decide}}. The function
\ensuremath{\mykeyword{operation}} takes a lambda expression \ensuremath{(\lambda \varid{x}\;\varid{l}\;\varid{k}\to \varid{e})} and returns type \ensuremath{\conid{Op}}, %
whose arguments are, respectively, the operation argument, the choice continuation, and the delimited
continuation.
The function \ensuremath{\mykeyword{handlerRet}} takes a return clause, a handler definition, and the
computation to be handled. We can also use \ensuremath{\mykeyword{handler}} without a
return clause.
The implementation also supports parameterized handlers.

Finally, a computation is written in a \ensuremath{\keyword{do}} block. This can invoke operations
using \ensuremath{\mykeyword{perform}}, by providing an operation and its argument (in this case
\ensuremath{\varid{decide}} and \ensuremath{()}). We can run the program by calling
\ensuremath{\mykeyword{runSel}}. For example, \ensuremath{\mykeyword{runSel}\;\varid{pgm}} returns \ensuremath{[\mskip1.5mu \conid{False},\conid{True}\mskip1.5mu]}.

\subsection{The Selection Monad}

We define the key datatype \ensuremath{\conid{Sel}\;\varid{r}\;\varid{e}\;\varid{a}}  implementing the programming interface:
 the
loss type \ensuremath{\varid{r}} is any
\ensuremath{\conid{Monoid}}
(not just a specific numerical type),
\ensuremath{\varid{e}} is the program's effect, and
\ensuremath{\varid{a}} is it's type:

\begin{hscode}\SaveRestoreHook
\column{B}{@{}>{\hspre}l<{\hspost}@{}}%
\column{E}{@{}>{\hspre}l<{\hspost}@{}}%
\>[B]{}\keyword{newtype}\;\conid{Sel}\;\varid{r}\;\varid{e}\;\varid{a}\mathrel{=}\conid{Sel}\;\{\mskip1.5mu \varid{unSel}\mathbin{::}\conid{Monoid}\;\varid{r}\Rightarrow (\varid{a}\to \conid{Eff}\;\varid{r}\;\varid{e}\;\varid{r})\to \conid{Eff}\;\varid{r}\;\varid{e}\;(\varid{r},\varid{a}))\mskip1.5mu\}{}\<[E]%
\ColumnHook
\end{hscode}\resethooks

\noindent
It takes a loss continuation \ensuremath{(\varid{a}\to \conid{Eff}\;\varid{r}\;\varid{e}\;\varid{r})} and returns a loss-value tuple \ensuremath{(\varid{r},\varid{a})}.
(Note that the definition corresponds to the semantic model
$S_\epsilon(X) = (X \to F_\epsilon(\R)) \to F_\epsilon(W(X))$ (\Cref{sec:selmon}).)
We use the \ensuremath{\conid{Eff}} datatype to represent effectful programs.
Our design is independent of the concrete strategy used for implementing effect handlers.
We
implemented them using \textit{multi-prompt delimited
  continuations}~\cite{dyvbig_jones_sabry_2007, Xie2021multip}. This
implementation closely follows the operational semantics of effect handlers.
For example, we can define \ensuremath{\effect{loss}} as follows:
\begin{hscode}\SaveRestoreHook
\column{B}{@{}>{\hspre}l<{\hspost}@{}}%
\column{E}{@{}>{\hspre}l<{\hspost}@{}}%
\>[B]{}\effect{loss}\;\varid{r}\mathrel{=}\conid{Sel}\mathbin{\$}\lambda \anonymous \to \mykeyword{return}\;(\varid{r},()){}\<[E]%
\ColumnHook
\end{hscode}\resethooks
The definition corresponds to rule $(R4)$ in \Cref{fig:redexrule},
which ignores the loss continuation, produces a loss \ensuremath{r}, and returns a unit value.

\begin{wrapfigure}{r}{0.43\textwidth}
\vspace{-18pt}
\begin{hscode}\SaveRestoreHook
\column{B}{@{}>{\hspre}l<{\hspost}@{}}%
\column{3}{@{}>{\hspre}l<{\hspost}@{}}%
\column{5}{@{}>{\hspre}l<{\hspost}@{}}%
\column{14}{@{}>{\hspre}l<{\hspost}@{}}%
\column{19}{@{}>{\hspre}l<{\hspost}@{}}%
\column{37}{@{}>{\hspre}l<{\hspost}@{}}%
\column{E}{@{}>{\hspre}l<{\hspost}@{}}%
\>[B]{}\keyword{instance}\;\conid{Monad}\;(\conid{Sel}\;\varid{r}\;\varid{e})\;\keyword{where}{}\<[E]%
\\
\>[B]{}\hsindent{3}{}\<[3]%
\>[3]{}\mykeyword{return}\;\varid{x}{}\<[14]%
\>[14]{}\mathrel{=}\conid{Sel}\;(\lambda \varid{g}\to \mykeyword{return}\;(\varid{mempty},\varid{x})){}\<[E]%
\\
\>[B]{}\hsindent{3}{}\<[3]%
\>[3]{}\varid{e}\bind \varid{f}{}\<[19]%
\>[19]{}\mathrel{=}\conid{Sel}\mathbin{\$}\lambda \varid{g}\to \keyword{do}\;{}\<[E]%
\\
\>[3]{}\hsindent{5}{}\<[5]%
\>[5]{}(\varid{r1},\varid{a})\leftarrow \varid{unSel}\;\varid{e}\;(\lambda \varid{a}\to (\varid{f}\;\varid{a})\triangleright\varid{g}){}\<[E]%
\\
\>[5]{}(\varid{r2},\varid{b})\leftarrow \varid{unSel}\;(\varid{f}\;\varid{a})\;\varid{g}{}\<[E]%
\\
\>[5]{}\mykeyword{return}\;(\varid{r1}\text{\textless\textgreater}\varid{r2},\varid{b}){}\<[E]%
\ColumnHook
\end{hscode}\resethooks
\vspace{-25pt}
\end{wrapfigure}
We present the monad instance declaration for \ensuremath{\conid{Sel}} on the right.
The definition requires some explanations. The \ensuremath{\mykeyword{return}} definition is
straightforward: we ignore the loss continuation and the handler
context, and return a pure tuple with a zero loss.
This corresponds to the evaluation of
terminal expressions (\Cref{fig:bigeval}).
The bind definition corresponds to rule $(F)$ in \Cref{fig:redexrule}.
First,
given the loss continuation \ensuremath{\varid{g}}, we would
first like to evaluate \ensuremath{\varid{e}}. However,
we first need to extend the loss continuation.
Here, \ensuremath{\varid{g}} is a loss continuation
for \ensuremath{(\varid{e} \bind \varid{f})}, not for \ensuremath{\varid{e}}.
Therefore, we transform the loss continuation
where $\triangleright$ implements the \textit{then} operator,
and evaluate \ensuremath{\varid{e}} to the extended loss continuation.
The result of evaluating \ensuremath{\varid{e}} is then passed to \ensuremath{\varid{f}},
where \ensuremath{\varid{f}} takes \ensuremath{\varid{a}} and the loss continuation \ensuremath{\varid{g}},
yielding a loss \ensuremath{\varid{r2}}
and a value \ensuremath{\varid{b}}.
Lastly, the two losses are 
combined (\ensuremath{\varid{r1}\text{\textless\textgreater} \varid{r2}}), again according to the big-step operational semantics,
and the final computation result is
\ensuremath{\varid{b}}.

\makeatletter
\@ifundefined{lhs2tex.lhs2tex.sty.read}%
  {\@namedef{lhs2tex.lhs2tex.sty.read}{}%
   \newcommand\SkipToFmtEnd{}%
   \newcommand\EndFmtInput{}%
   \long\def\SkipToFmtEnd#1\EndFmtInput{}%
  }\SkipToFmtEnd

\newcommand\ReadOnlyOnce[1]{\@ifundefined{#1}{\@namedef{#1}{}}\SkipToFmtEnd}
\usepackage{amstext}
\usepackage{amssymb}
\usepackage{stmaryrd}
\DeclareFontFamily{OT1}{cmtex}{}
\DeclareFontShape{OT1}{cmtex}{m}{n}
  {<5><6><7><8>cmtex8
   <9>cmtex9
   <10><10.95><12><14.4><17.28><20.74><24.88>cmtex10}{}
\DeclareFontShape{OT1}{cmtex}{m}{it}
  {<-> ssub * cmtt/m/it}{}
\newcommand{\texfamily}{\fontfamily{cmtex}\selectfont}
\DeclareFontShape{OT1}{cmtt}{bx}{n}
  {<5><6><7><8>cmtt8
   <9>cmbtt9
   <10><10.95><12><14.4><17.28><20.74><24.88>cmbtt10}{}
\DeclareFontShape{OT1}{cmtex}{bx}{n}
  {<-> ssub * cmtt/bx/n}{}
\newcommand{\tex}[1]{\text{\texfamily#1}}	%

\newcommand{\Sp}{\hskip.33334em\relax}

\newcommand{\Conid}[1]{\mathit{#1}}
\newcommand{\Varid}[1]{\mathit{#1}}
\newcommand{\anonymous}{\kern0.06em \vbox{\hrule\@width.5em}}
\newcommand{\plus}{\mathbin{+\!\!\!+}}
\newcommand{\bind}{\mathbin{>\!\!\!>\mkern-6.7mu=}}
\newcommand{\rbind}{\mathbin{=\mkern-6.7mu<\!\!\!<}}%
\newcommand{\sequ}{\mathbin{>\!\!\!>}}
\renewcommand{\leq}{\leqslant}
\renewcommand{\geq}{\geqslant}
\usepackage{polytable}

\@ifundefined{mathindent}%
  {\newdimen\mathindent\mathindent\leftmargini}%
  {}%

\def\resethooks{%
  \global\let\SaveRestoreHook\empty
  \global\let\ColumnHook\empty}
\newcommand*{\savecolumns}[1][default]%
  {\g@addto@macro\SaveRestoreHook{\savecolumns[#1]}}
\newcommand*{\restorecolumns}[1][default]%
  {\g@addto@macro\SaveRestoreHook{\restorecolumns[#1]}}
\newcommand*{\aligncolumn}[2]%
  {\g@addto@macro\ColumnHook{\column{#1}{#2}}}

\resethooks

\newcommand{\onelinecommentchars}{\quad-{}- }
\newcommand{\commentbeginchars}{\enskip\{-}
\newcommand{\commentendchars}{-\}\enskip}

\newcommand{\visiblecomments}{%
  \let\onelinecomment=\onelinecommentchars
  \let\commentbegin=\commentbeginchars
  \let\commentend=\commentendchars}

\newcommand{\invisiblecomments}{%
  \let\onelinecomment=\empty
  \let\commentbegin=\empty
  \let\commentend=\empty}

\visiblecomments

\newlength{\blanklineskip}
\setlength{\blanklineskip}{0.66084ex}

\newcommand{\hsindent}[1]{\quad}%
\let\hspre\empty
\let\hspost\empty
\newcommand{\NB}{\textbf{NB}}
\newcommand{\Todo}[1]{$\langle$\textbf{To do:}~#1$\rangle$}

\EndFmtInput
\makeatother
\ReadOnlyOnce{polycode.fmt}%
\makeatletter

\newcommand{\hsnewpar}[1]%
  {{\parskip=0pt\parindent=0pt\par\vskip #1\noindent}}

\newcommand{\hscodestyle}{}

\newcommand{\sethscode}[1]%
  {\expandafter\let\expandafter\hscode\csname #1\endcsname
   \expandafter\let\expandafter\endhscode\csname end#1\endcsname}

\newenvironment{compathscode}%
  {\par\noindent
   \advance\leftskip\mathindent
   \hscodestyle
   \let\\=\@normalcr
   \let\hspre\(\let\hspost\)%
   \pboxed}%
  {\endpboxed\)%
   \par\noindent
   \ignorespacesafterend}

\newcommand{\compaths}{\sethscode{compathscode}}

\newenvironment{plainhscode}%
  {\hsnewpar\abovedisplayskip
   \advance\leftskip\mathindent
   \hscodestyle
   \let\hspre\(\let\hspost\)%
   \pboxed}%
  {\endpboxed%
   \hsnewpar\belowdisplayskip
   \ignorespacesafterend}

\newenvironment{oldplainhscode}%
  {\hsnewpar\abovedisplayskip
   \advance\leftskip\mathindent
   \hscodestyle
   \let\\=\@normalcr
   \(\pboxed}%
  {\endpboxed\)%
   \hsnewpar\belowdisplayskip
   \ignorespacesafterend}

\newcommand{\plainhs}{\sethscode{plainhscode}}
\newcommand{\oldplainhs}{\sethscode{oldplainhscode}}
\plainhs

\newenvironment{arrayhscode}%
  {\hsnewpar\abovedisplayskip
   \advance\leftskip\mathindent
   \hscodestyle
   \let\\=\@normalcr
   \(\parray}%
  {\endparray\)%
   \hsnewpar\belowdisplayskip
   \ignorespacesafterend}

\newcommand{\arrayhs}{\sethscode{arrayhscode}}

\newenvironment{mathhscode}%
  {\parray}{\endparray}

\newcommand{\mathhs}{\sethscode{mathhscode}}

\newenvironment{texthscode}%
  {\(\parray}{\endparray\)}

\newcommand{\texths}{\sethscode{texthscode}}

\def\codeframewidth{\arrayrulewidth}
\RequirePackage{calc}

\newenvironment{framedhscode}%
  {\parskip=\abovedisplayskip\par\noindent
   \hscodestyle
   \arrayrulewidth=\codeframewidth
   \tabular{@{}|p{\linewidth-2\arraycolsep-2\arrayrulewidth-2pt}|@{}}%
   \hline\framedhslinecorrect\\{-1.5ex}%
   \let\endoflinesave=\\
   \let\\=\@normalcr
   \(\pboxed}%
  {\endpboxed\)%
   \framedhslinecorrect\endoflinesave{.5ex}\hline
   \endtabular
   \parskip=\belowdisplayskip\par\noindent
   \ignorespacesafterend}

\newcommand{\framedhslinecorrect}[2]%
  {#1[#2]}

\newcommand{\framedhs}{\sethscode{framedhscode}}

\newenvironment{inlinehscode}%
  {\(\def\column##1##2{}%
   \let\>\undefined\let\<\undefined\let\\\undefined
   \newcommand\>[1][]{}\newcommand\<[1][]{}\newcommand\\[1][]{}%
   \def\fromto##1##2##3{##3}%
   \def\nextline{}}{\) }%

\newcommand{\inlinehs}{\sethscode{inlinehscode}}

\newenvironment{joincode}%
  {\let\orighscode=\hscode
   \let\origendhscode=\endhscode
   \def\endhscode{\def\hscode{\endgroup\def\@currenvir{hscode}\\}\begingroup}
   \orighscode\def\hscode{\endgroup\def\@currenvir{hscode}}}%
  {\origendhscode
   \global\let\hscode=\orighscode
   \global\let\endhscode=\origendhscode}%

\makeatother
\EndFmtInput
\ReadOnlyOnce{forall.fmt}%
\makeatletter

\let\HaskellResetHook\empty
\newcommand*{\AtHaskellReset}[1]{%
  \g@addto@macro\HaskellResetHook{#1}}
\newcommand*{\HaskellReset}{\HaskellResetHook}

\global\let\hsforallread\empty
\global\let\hsexistsread\empty

\newcommand\hsforall{\global\let\hsdot=\hsperiodonce}
\newcommand\hsexists{\global\let\hsdot=\hsperiodonce}
\newcommand*\hsperiodonce[2]{#2\global\let\hsdot=\hscompose}
\newcommand*\hscompose[2]{#1}

\AtHaskellReset{\global\let\hsdot=\hscompose}

\HaskellReset

\makeatother
\EndFmtInput

\subsection{Examples}
\label{sec:examples}

\paragraph{\textbf{Example: Greedy algorithms}}
A greedy selection strategy always picks the choice that maximizes (or minimizes) losses.
We define the \ensuremath{\conid{Max}} effect with a \ensuremath{\varid{max}} operation; a corresponding handler can selects the element maximizing the loss,
where \ensuremath{\varid{maxWith}} implements $\mathsf{argmax}$:
\begin{hscode}\SaveRestoreHook
\column{B}{@{}>{\hspre}l<{\hspost}@{}}%
\column{77}{@{}>{\hspre}c<{\hspost}@{}}%
\column{77E}{@{}l@{}}%
\column{E}{@{}>{\hspre}l<{\hspost}@{}}%
\>[B]{}[\mykeyword{effect}|\;\;\keyword{data}\;\conid{Max}\mathrel{=}\conid{Max}\;\{\mskip1.5mu \varid{max}\mathbin{::}\conid{Op}\;[\mskip1.5mu \varid{a}\mskip1.5mu]\;\varid{a}\mskip1.5mu\}\mskip1.5mu]\;{}\<[E]%
\\
\>[B]{}\varid{hmax}\mathrel{=}\mykeyword{handler}\;\conid{Max}\;\{\mskip1.5mu \varid{max}\mathrel{=}\mykeyword{operation}\;(\lambda \varid{x}\;\varid{l}\;\varid{k}\to \keyword{do}\;\varid{b}\leftarrow \varid{maxWith}\;\varid{l}\;\varid{x};\varid{k}\;\varid{b}){}\<[77]%
\>[77]{}\mskip1.5mu\}{}\<[77E]%
\ColumnHook
\end{hscode}\resethooks
As an example, we can define criteria for selecting a \ensuremath{\conid{String}} based
on its length and the number of distinct characters, with greater losses (really, rewards) for better strings:
\begin{hscode}\SaveRestoreHook
\column{B}{@{}>{\hspre}l<{\hspost}@{}}%
\column{E}{@{}>{\hspre}l<{\hspost}@{}}%
\>[B]{}\varid{len}\;\varid{x}\mathrel{=}\effect{loss}\;(\varid{fromIntegral}\;(\varid{length}\;\varid{x})){}\<[E]%
\\
\>[B]{}\varid{distinct}\;\varid{x}\mathrel{=}\keyword{let}\;\varid{i}\mathrel{=}\varid{fromIntegral}\;(\varid{length}\;(\varid{group}\;(\varid{sort}\;\varid{x})))\;\keyword{in}\;\effect{loss}\;(\varid{i}\mathbin{*}\varid{i}){}\<[E]%
\ColumnHook
\end{hscode}\resethooks

\noindent
where \ensuremath{\varid{sort}} sorts the string,
and the \ensuremath{\varid{group}} function collects consecutive identical characters into separate lists.  Thus, the number of groups is the number of distinct characters in the string.

\begin{wrapfigure}{r}{0.4\textwidth}
\vspace{-20pt}
\begin{hscode}\SaveRestoreHook
\column{B}{@{}>{\hspre}l<{\hspost}@{}}%
\column{3}{@{}>{\hspre}l<{\hspost}@{}}%
\column{E}{@{}>{\hspre}l<{\hspost}@{}}%
\>[B]{}\varid{password}\mathrel{=}\keyword{do}{}\<[E]%
\\
\>[B]{}\hsindent{3}{}\<[3]%
\>[3]{}\varid{s}\leftarrow \mykeyword{perform}\;\varid{max}\;[\mskip1.5mu \text{\ttfamily \char34 aaa\char34},\text{\ttfamily \char34 aabb\char34},\text{\ttfamily \char34 abc\char34}\mskip1.5mu]{}\<[E]%
\\
\>[B]{}\hsindent{3}{}\<[3]%
\>[3]{}\varid{len}\;\varid{s}{}\<[E]%
\\
\>[B]{}\hsindent{3}{}\<[3]%
\>[3]{}\varid{distinct}\;\varid{s}{}\<[E]%
\\
\>[B]{}\hsindent{3}{}\<[3]%
\>[3]{}\mykeyword{return}\mathbin{\$}\text{\ttfamily \char34 password~is~\char34}\plus \varid{s}{}\<[E]%
\ColumnHook
\end{hscode}\resethooks
\vspace{-25pt}
\end{wrapfigure}

We can then define a program \ensuremath{\varid{password}} that picks a password from a list based on
these criteria,
where \ensuremath{\plus} concatenates two lists.
Using \ensuremath{\varid{hmax}} to handle \ensuremath{\varid{max}}, \ensuremath{\mykeyword{runSel}\mathbin{\$}\varid{hmax}\;\varid{password}}
returns \ensuremath{\text{\ttfamily \char34 password~is~abc\char34}}, since \ensuremath{\text{\ttfamily \char34 abc\char34}} has the greatest reward.

\paragraph{\textbf{Example: Optimizations}}

Greedy algorithms always pick the optimal option. However, it is not always
possible to enumerate all possible choices and identify the best one.

We consider optimization algorithms, specifically \textit{stochastic gradient
  descent} (SGD), a widely used method for iterative optimization. Starting with
initial parameters, SGD minimizes a cost function by repeatedly updating the
parameters in the direction opposite to their gradients, calculated after
processing each randomly selected data point.

We implement gradient descent as a handler that chooses new parameters as
follows, where the 
\begin{wrapfigure}{r}{0.58\textwidth}
\vspace{-15pt}
\begin{hscode}\SaveRestoreHook
\column{B}{@{}>{\hspre}l<{\hspost}@{}}%
\column{33}{@{}>{\hspre}l<{\hspost}@{}}%
\column{34}{@{}>{\hspre}l<{\hspost}@{}}%
\column{37}{@{}>{\hspre}l<{\hspost}@{}}%
\column{55}{@{}>{\hspre}l<{\hspost}@{}}%
\column{E}{@{}>{\hspre}l<{\hspost}@{}}%
\>[B]{}[\mykeyword{effect}|\;\;\keyword{data}\;\conid{Opt}\mathrel{=}\conid{Opt}\;\{\mskip1.5mu \effect{optimize}\mathbin{::}\conid{Op}\;[\mskip1.5mu \conid{Float}\mskip1.5mu]\;[\mskip1.5mu \conid{Float}\mskip1.5mu]\mskip1.5mu\}\mskip1.5mu]\;{}\<[E]%
\\
\>[B]{}\varid{hOpt}\mathrel{=}\mykeyword{handler}\;(\conid{Opt}\;\{\mskip1.5mu \effect{optimize}\mathrel{=}{}\<[34]%
\>[34]{}\mykeyword{operation}\;(\lambda \varid{p}\;\varid{l}\;\varid{k}\to {}\<[E]%
\\
\>[B]{}\hsindent{33}{}\<[33]%
\>[33]{}\keyword{do}\;{}\<[37]%
\>[37]{}\varid{ds}\leftarrow \varid{autodiff}\;\varid{l}\;\varid{p}{}\<[E]%
\\
\>[37]{}\keyword{let}\;\varid{p'}\mathrel{=}\varid{zipWith}\;{}\<[55]%
\>[55]{}(\lambda \varid{w}\;\varid{d}\to (\varid{w}\mathbin{-}\mathrm{0.01}\mathbin{*}\varid{d}))\;\varid{p}\;\varid{ds}{}\<[E]%
\\
\>[37]{}\varid{k}\;\varid{p'})\mskip1.5mu\}){}\<[E]%
\ColumnHook
\end{hscode}\resethooks
\vspace{-15pt}
\end{wrapfigure}
function \ensuremath{\varid{autodiff}\;\varid{f}\;\varid{x}} calculates the gradient of a
differentiable function \ensuremath{\varid{f}} at point \ensuremath{\varid{x}}.
The \ensuremath{\effect{optimize}} clause first calculates the gradient
\ensuremath{\varid{ds}} by differentiating the
choice continuation \ensuremath{\varid{l}} with respect to the parameters \ensuremath{\varid{p}}. 
It then updates the parameters using \ensuremath{\varid{zipWith}},
where \ensuremath{\mathrm{0.01}} is the \textit{learning
  rate}.
Lastly, it
resumes the
continuation with updated parameters \ensuremath{\varid{p'}}. 

As a concrete example, we use the simplest form of \textit{linear
  regression}~\cite{legendre1806nouvelles} with only one variable,
  a standard example when explaining SGD.
Given a dataset of $(x_i, y_i)_{i=1,..,n}$, where $x_i$
and $y_i$ are real numbers, the goal is to find a weight $w$ and a bias $b$ that
minimize the cost function $\sum_{i=1,..n,}(f(x_i) - y_i)^2$ for the linear
model $f(x) = wx + b$.

\begin{wrapfigure}{r}{0.4\textwidth}
\vspace{-18pt}
\begin{hscode}\SaveRestoreHook
\column{B}{@{}>{\hspre}l<{\hspost}@{}}%
\column{3}{@{}>{\hspre}l<{\hspost}@{}}%
\column{7}{@{}>{\hspre}l<{\hspost}@{}}%
\column{E}{@{}>{\hspre}l<{\hspost}@{}}%
\>[B]{}\varid{linearReg}\;[\mskip1.5mu \varid{w},\varid{b}\mskip1.5mu]\;\varid{x}\;\varid{target}\mathrel{=}{}\<[E]%
\\
\>[B]{}\hsindent{3}{}\<[3]%
\>[3]{}\keyword{do}\;{}\<[7]%
\>[7]{}[\mskip1.5mu \varid{w'},\varid{b'}\mskip1.5mu]\leftarrow \mykeyword{perform}\;\effect{optimize}\;[\mskip1.5mu \varid{w},\varid{b}\mskip1.5mu]{}\<[E]%
\\
\>[7]{}\keyword{let}\;\varid{y}\mathrel{=}\varid{w'}\mathbin{*}\varid{x}\mathbin{+}\varid{b'}{}\<[E]%
\\
\>[7]{}\effect{loss}\mathbin{\$}(\varid{target}\mathbin{-}\varid{y})\mathbin{*}(\varid{target}\mathbin{-}\varid{y}){}\<[E]%
\\
\>[7]{}\mykeyword{return}\;[\mskip1.5mu \varid{w'},\varid{b'}\mskip1.5mu]{}\<[E]%
\ColumnHook
\end{hscode}\resethooks
\vspace{-15pt}
\end{wrapfigure}

The program \ensuremath{\varid{linearReg}} on the right defines a linear regression model. It takes the
current parameters \ensuremath{\varid{w}} and \ensuremath{\varid{b}} (represented as a list) and a data
point \ensuremath{\varid{x}} and \ensuremath{\varid{target}}, and returns new parameters \ensuremath{\varid{w'}} and \ensuremath{\varid{b'}}.

The program first calls an operation $\ensuremath{\effect{optimize}}$ with the current parameters
\ensuremath{[\mskip1.5mu \varid{w},\varid{b}\mskip1.5mu]} and receives updated parameters \ensuremath{[\mskip1.5mu \varid{w'},\varid{b'}\mskip1.5mu]}. It then calculates the
predicted value \ensuremath{\varid{y}} using these new parameters and calls \ensuremath{\effect{loss}} with the
corresponding squared error. Finally, the program returns updated parameters.

We combine the \ensuremath{\varid{hOpt}} handler and the \ensuremath{\varid{linearReg}} program as follows.
\begin{hscode}\SaveRestoreHook
\column{B}{@{}>{\hspre}l<{\hspost}@{}}%
\column{E}{@{}>{\hspre}l<{\hspost}@{}}%
\>[B]{}\varid{foldM}\;(\lambda \varid{p}\;(\varid{x},\varid{y})\to \varid{lreset}\mathbin{\$}\varid{hOpt}\mathbin{\$}\varid{linearReg}\;\varid{p}\;\varid{x}\;\varid{y})\;\varid{random\char95 params}\;\varid{training\char95 data}{}\<[E]%
\ColumnHook
\end{hscode}\resethooks
The program traverses the training dataset, applying gradient descent to each
data point. Note that we apply \ensuremath{\varid{lreset}} that combines local and reset within the
loop body, so each iteration makes decisions based on its own loss. Moreover, we can introduce
a random effect to shuffle the training data, introducing stochasticity into the
process.

\paragraph{\textbf{Example: Hyperparameters}}

In the gradient descent handler, we used the learning rate $0.01$.
For training programs, variables such as the learning rate that govern the training
process are called
\begin{wrapfigure}{r}{0.49\textwidth}
\vspace{-15pt}
\begin{hscode}\SaveRestoreHook
\column{B}{@{}>{\hspre}l<{\hspost}@{}}%
\column{31}{@{}>{\hspre}l<{\hspost}@{}}%
\column{32}{@{}>{\hspre}l<{\hspost}@{}}%
\column{35}{@{}>{\hspre}l<{\hspost}@{}}%
\column{E}{@{}>{\hspre}l<{\hspost}@{}}%
\>[B]{}[\mykeyword{effect}|\;\;\keyword{data}\;\conid{LR}\mathrel{=}\conid{LR}\;\{\mskip1.5mu \effect{lrate}\mathbin{::}\conid{Op}\;()\;\conid{Float}\mskip1.5mu\}\mskip1.5mu]\;{}\<[E]%
\\
\>[B]{}\varid{gd}\mathrel{=}\mykeyword{handler}\;(\conid{Opt}\;\{\mskip1.5mu \effect{optimize}\mathrel{=}{}\<[32]%
\>[32]{}\mykeyword{operation}\;(\lambda \varid{p}\;\varid{l}\;\varid{k}\to {}\<[E]%
\\
\>[B]{}\hsindent{31}{}\<[31]%
\>[31]{}\keyword{do}\;{}\<[35]%
\>[35]{}\varid{ds}\leftarrow \varid{autodiff}\;\varid{l}\;\varid{p}{}\<[E]%
\\
\>[35]{}\alpha\leftarrow \mykeyword{perform}\;\effect{lrate}\;(){}\<[E]%
\\
\>[35]{}\keyword{let}\;\varid{p'}\mathrel{=}\varid{zipWith}\;(\lambda \varid{w}\;\varid{d}\to (\varid{w}\mathbin{-}\alpha\mathbin{*}\varid{d}))\;\varid{p}\;\varid{ds}{}\<[E]%
\\
\>[35]{}\varid{k}\;\varid{p'})\mskip1.5mu\}){}\<[E]%
\ColumnHook
\end{hscode}\resethooks
\vspace{-20pt}
\end{wrapfigure}
\textit{hyperparameters}. The process of finding their
optimal configuration is known as \textit{hyperparameter
  optimization}~\cite{feurer2019hyperparameter}.

We can abstract the learning rate as a separate effect operation as shown on the right.
A handler that always returns a pre-defined learning rate can be defined as follows:
\begin{hscode}\SaveRestoreHook
\column{B}{@{}>{\hspre}l<{\hspost}@{}}%
\column{E}{@{}>{\hspre}l<{\hspost}@{}}%
\>[B]{}\varid{readLR}\;\alpha\mathrel{=}\mykeyword{handler}\;(\conid{LR}\;\{\mskip1.5mu \effect{lrate}\mathrel{=}\mykeyword{operation}\;(\lambda \varid{x}\;\anonymous \;\varid{k}\to \varid{k}\;\alpha)\mskip1.5mu\}){}\<[E]%
\ColumnHook
\end{hscode}\resethooks

More interestingly, a handler for hyperparameter tuning can compare losses from
different configurations. As an
example, the handler below implements a simple \textit{grid
  search} that exhaustively explores
a subset of the hyperparameter space, in this case, for simplicity, two
options:
\begin{hscode}\SaveRestoreHook
\column{B}{@{}>{\hspre}l<{\hspost}@{}}%
\column{39}{@{}>{\hspre}l<{\hspost}@{}}%
\column{E}{@{}>{\hspre}l<{\hspost}@{}}%
\>[B]{}\varid{tuneLR}\;(\alpha_1,\alpha_2)\mathrel{=}\mykeyword{handlerRet}\;(\lambda \anonymous \to \mykeyword{return}\;\alpha_1){}\<[E]%
\\
\>[B]{}\conid{LR}\;\{\mskip1.5mu \effect{lrate}\mathrel{=}\mykeyword{operation}\;(\lambda \anonymous \;\varid{l}\;\anonymous \to \keyword{do}\;{}\<[39]%
\>[39]{}\varid{err1}\leftarrow \varid{l}\;\alpha_1;\varid{err2}\leftarrow \varid{l}\;\alpha_2{}\<[E]%
\\
\>[39]{}\keyword{if}\;\varid{err1}\mathbin{<}\varid{err2}\;\keyword{then}\;\mykeyword{return}\;\alpha_1\;\keyword{else}\;\mykeyword{return}\;\alpha_2)\mskip1.5mu\}{}\<[E]%
\ColumnHook
\end{hscode}\resethooks
The handler calculates the losses for the two learning rates respectively, and
returns the one with a lower loss, without resuming the computation.

\paragraph{\textbf{Example: Two-player games}}

\newcommand\pdata{{\displaystyle p_{\mathsf{data}}}}
\newcommand\pnoise{{\displaystyle p_{\mathsf{noise}}}}

\begin{wrapfigure}{r}{0.22\textwidth}
  \vspace{-14pt}
  \small
  \noindent
  \begin{tabular}{l|cc}
    \backslashbox{A}{B} & Left & Right \\
    \hline
    Left    & 5    &  3 \\
    Right  & 2    & 9 \\
  \end{tabular}
  \vspace{-20pt}
\end{wrapfigure}
A minimax game corresponds to the philosophy of minimizing potential loss in a
worst-case scenario. It involves two players: maximizer, who seeks to
maximize the loss, and minimizer, who aims to minimize it (\Cref{sec:selmon}).
As an example, consider a game with four final
states.
The losses associated with both player's decisions are shown in the table on the right.

A corresponding minimizer handler can be defined as:
\begin{hscode}\SaveRestoreHook
\column{B}{@{}>{\hspre}l<{\hspost}@{}}%
\column{3}{@{}>{\hspre}l<{\hspost}@{}}%
\column{31}{@{}>{\hspre}c<{\hspost}@{}}%
\column{31E}{@{}l@{}}%
\column{E}{@{}>{\hspre}l<{\hspost}@{}}%
\>[B]{}[\mykeyword{effect}|\;\;\keyword{data}\;\conid{Min}\;\varid{a}\mathrel{=}\conid{Min}\;\{\mskip1.5mu \varid{min}\mathbin{::}\conid{Op}\;[\mskip1.5mu \varid{a}\mskip1.5mu]\;\varid{a}\mskip1.5mu\}\mskip1.5mu]\;{}\<[E]%
\\
\>[B]{}\varid{hmin}\mathrel{=}\mykeyword{handler}\;\conid{Min}\;\{\mskip1.5mu \varid{min}\mathrel{=}\mykeyword{operation}\;(\lambda \varid{x}\;\varid{l}\;\varid{k}\to {}\keyword{do}\;\varid{b}\leftarrow \varid{minWith}\;\varid{l}\;\varid{x};\varid{k}\;\varid{b}){}\mskip1.5mu\}\<[E]%
\\
\ColumnHook
\end{hscode}\resethooks

\begin{wrapfigure}{r}{0.52\textwidth}
\vspace{-15pt}
\begin{hscode}\SaveRestoreHook
\column{B}{@{}>{\hspre}l<{\hspost}@{}}%
\column{3}{@{}>{\hspre}l<{\hspost}@{}}%
\column{E}{@{}>{\hspre}l<{\hspost}@{}}%
\>[B]{}\keyword{data}\;\conid{Strategy}\mathrel{=}\conid{Left}\mid \conid{Right}\;\keyword{deriving}\;(\conid{Enum}){}\<[E]%
\\
\>[B]{}\varid{minimax}\mathrel{=}\keyword{do}{}\<[E]%
\\
\>[B]{}\hsindent{3}{}\<[3]%
\>[3]{}\varid{a}\leftarrow \mykeyword{perform}\;\varid{max}\;[\mskip1.5mu \conid{Left},\conid{Right}\mskip1.5mu]{}\<[E]%
\\
\>[B]{}\hsindent{3}{}\<[3]%
\>[3]{}\varid{b}\leftarrow \mykeyword{perform}\;\varid{min}\;[\mskip1.5mu \conid{Left},\conid{Right}\mskip1.5mu]{}\<[E]%
\\
\>[B]{}\hsindent{3}{}\<[3]%
\>[3]{}\effect{loss}\mathbin{\$}[\mskip1.5mu [\mskip1.5mu \mathrm{5},\mathrm{3}\mskip1.5mu],[\mskip1.5mu \mathrm{2},\mathrm{9}\mskip1.5mu]\mskip1.5mu]\mathbin{!!}(\varid{fromEnum}\;\varid{a})\mathbin{!!}(\varid{fromEnum}\;\varid{b}){}\<[E]%
\\
\>[B]{}\hsindent{3}{}\<[3]%
\>[3]{}\mykeyword{return}\;(\varid{a},\varid{b}){}\<[E]%
\ColumnHook
\end{hscode}\resethooks
\vspace{-15pt}
\end{wrapfigure}
We can then play this simple game, where the maximizer A chooses over the minimizer B, as
given on the right,
where we encode the loss table as a nested list, and \ensuremath{\mathbin{(!!)}} is list index operator. 
The program \ensuremath{(\mykeyword{runSel}\mathbin{\$}\varid{hmax}\mathbin{\$}\varid{hmin}\;\varid{minimax})} returns \ensuremath{(\conid{Left},\conid{Right})} with loss \ensuremath{\mathrm{3}}. This is
because A maximizes over the choices of B.
Specifically,
B, the minimizer, chooses 3 (between 3 and 5),
and 2 (between 2 and 9).
Then A, the maximizer, chooses 3 (between 3 and 2).
Note how the loss is shared by two handlers.

In the training domain, \textit{generative
  adversarial networks}~\cite{Gan2014} is also a two-player game.
The algorithm simultaneously trains two models that contest with each other: the
generative model learns to generate samples, while the discriminative model learns
to distinguish between real and generated samples.
More explicitly, it corresponds to
{\small$\min_{G} \max_{D} ~ ( \mathbb{E}_{x \sim \pdata}[\log D(x)] + \mathbb{E}_{z \sim
  \pnoise}[\log (1- D(G(z)))])$},
\noindent
where the discriminator is a minimizer and the generator is a maximizer.

\paragraph{\textbf{Example: Nash equilibrium}}
In game theory, a \textit{Nash equilibrium} describes a situation where no
player can improve their outcome by unilaterally changing their strategy,
assuming all other players maintain their current strategies.

A classic example is the \textit{prisoner's dilemma}, illustrated in the table
below. Here, the loss is 
represented
as a pair, indicating the respective
prison sentences for prisoner A and prisoner B.
If both
\begin{wrapfigure}{r}{0.33\textwidth}
  \small
  \noindent
  \begin{tabular}{l|cc}
    \backslashbox[17mm]{A}{B} & defects & cooperates \\
    \hline
    defects   & (3, 3)    & (0, 5) \\
    cooperates  & (5, 0)   & (1, 1)
  \end{tabular}
  \vspace{-20pt}
\end{wrapfigure}%
prisoners cooperate (by
staying silent), they each serve one year (loss of 1). However, if one prisoner
cooperates while the other defects, the cooperating prisoner serves 5 years,
while the defecting prisoner goes free. If both defect, they each serve 3
years. In this scenario, defection always yields a better individual outcome
regardless of the other prisoner's choice.

We define game steps as follows, using \ensuremath{\conid{Strategy}} for defection (\ensuremath{\conid{Left}}) and cooperation (\ensuremath{\conid{Right}}).
\begin{hscode}\SaveRestoreHook
\column{B}{@{}>{\hspre}l<{\hspost}@{}}%
\column{33}{@{}>{\hspre}l<{\hspost}@{}}%
\column{E}{@{}>{\hspre}l<{\hspost}@{}}%
\>[B]{}\keyword{data}\;\conid{Step}\mathrel{=}\conid{Move}\;\conid{Strategy}\mid \conid{Stay}\;\conid{Strategy}\;\keyword{deriving}\;(\conid{Eq}){}\<[E]%
\\
\>[B]{}[\mykeyword{effect}|\;\;\keyword{data}\;\conid{Play}\mathrel{=}\conid{Play}\;\{\mskip1.5mu \varid{play}{}\<[33]%
\>[33]{}\mathbin{::}\conid{Op}\;(\conid{Step},\conid{Step})\;(\conid{Step},\conid{Step})\mskip1.5mu\}\mskip1.5mu]{}\<[E]%
\ColumnHook
\end{hscode}\resethooks

\begin{wrapfigure}{r}{0.55\textwidth}
\vspace{-17pt}
\begin{hscode}\SaveRestoreHook
\column{B}{@{}>{\hspre}l<{\hspost}@{}}%
\column{3}{@{}>{\hspre}l<{\hspost}@{}}%
\column{E}{@{}>{\hspre}l<{\hspost}@{}}%
\>[B]{}\varid{hNash}\mathrel{=}\mykeyword{handler}\;\conid{Play}\;\{\mskip1.5mu \varid{play}\mathrel{=}\mykeyword{operation}\;(\lambda (\varid{a},\varid{b})\;\varid{l}\;\varid{k}\to \keyword{do}{}\<[E]%
\\
\>[B]{}\hsindent{3}{}\<[3]%
\>[3]{}\keyword{let}\;(\varid{a1},\varid{b1})\mathrel{=}(\varid{getStrtgy}\;\varid{a},\varid{getStrtgy}\;\varid{b}){}\<[E]%
\\
\>[B]{}\hsindent{3}{}\<[3]%
\>[3]{}\keyword{let}\;(\varid{a2},\varid{b2})\mathrel{=}(\varid{move}\;\varid{a1},\varid{move}\;\varid{b1}){}\<[E]%
\\
\>[B]{}\hsindent{3}{}\<[3]%
\>[3]{}\varid{l1}\leftarrow \varid{l}\;(\conid{Stay}\;\varid{a1},\conid{Stay}\;\varid{b1});\varid{l2}\leftarrow \varid{l}\;(\conid{Stay}\;\varid{a2},\conid{Stay}\;\varid{b1}){}\<[E]%
\\
\>[B]{}\hsindent{3}{}\<[3]%
\>[3]{}\varid{l3}\leftarrow \varid{l}\;(\conid{Stay}\;\varid{a1},\conid{Stay}\;\varid{b2}){}\<[E]%
\\
\>[B]{}\hsindent{3}{}\<[3]%
\>[3]{}\keyword{if}\;(\varid{fst}\;\varid{l2}\mathbin{<}\varid{fst}\;\varid{l1})\;\keyword{then}\;\varid{k}\;(\conid{Move}\;\varid{a2},\conid{Stay}\;\varid{b1}){}\<[E]%
\\
\>[B]{}\hsindent{3}{}\<[3]%
\>[3]{}\keyword{else}\;\keyword{if}\;(\varid{snd}\;\varid{l3}\mathbin{<}\varid{snd}\;\varid{l1})\;\keyword{then}\;\varid{k}\;(\conid{Stay}\;\varid{a1},\conid{Move}\;\varid{b2}){}\<[E]%
\\
\>[B]{}\hsindent{3}{}\<[3]%
\>[3]{}\keyword{else}\;\varid{k}\;(\conid{Stay}\;\varid{a1},\conid{Stay}\;\varid{b1}))\mskip1.5mu\}{}\<[E]%
\ColumnHook
\end{hscode}\resethooks
\vspace{-20pt}
\end{wrapfigure}
Given both players' strategies, the handler on the right tries to reduce one
player's loss by adjusting their strategy while holding the other player's
strategy unchanged. The function \ensuremath{\varid{getStrtgy}} extracts the strategy from the
current step, and \ensuremath{\varid{move}} modifies the strategy of the specified player.
Each player compares their own loss and decides
whether to adjust their strategy or stay unchanged.

The \ensuremath{\varid{game}} program below iteratively adjusts the players' strategies until both
choose to \ensuremath{\conid{Stay}}. This signifies that a Nash equilibrium has been reached, and
the program 
terminates.
The program
\begin{wrapfigure}{r}{0.43\textwidth}
\vspace{-17pt}
\begin{hscode}\SaveRestoreHook
\column{B}{@{}>{\hspre}l<{\hspost}@{}}%
\column{3}{@{}>{\hspre}l<{\hspost}@{}}%
\column{9}{@{}>{\hspre}l<{\hspost}@{}}%
\column{11}{@{}>{\hspre}l<{\hspost}@{}}%
\column{E}{@{}>{\hspre}l<{\hspost}@{}}%
\>[B]{}\varid{game}\;\varid{a}\;\varid{b}\mathrel{=}\keyword{do}{}\<[E]%
\\
\>[B]{}\hsindent{3}{}\<[3]%
\>[3]{}(\varid{a'},\varid{b'})\leftarrow \varid{lreset}\mathbin{\$}\varid{hNash}\mathbin{\$}\keyword{do}{}\<[E]%
\\
\>[3]{}\hsindent{6}{}\<[9]%
\>[9]{}(\varid{a1},\varid{b1})\leftarrow \mykeyword{perform}\;\varid{play}\;(\varid{a},\varid{b}){}\<[E]%
\\
\>[3]{}\hsindent{6}{}\<[9]%
\>[9]{}\keyword{let}\;(\varid{a2},\varid{b2})\mathrel{=}(\varid{getStrtgy}\;\varid{a1},\varid{getStrtgy}\;\varid{b1}){}\<[E]%
\\
\>[3]{}\hsindent{6}{}\<[9]%
\>[9]{}\effect{loss}\mathbin{\$}[\mskip1.5mu [\mskip1.5mu (\mathrm{3},\mathrm{3}),(\mathrm{0},\mathrm{5})\mskip1.5mu],[\mskip1.5mu (\mathrm{5},\mathrm{0}),(\mathrm{1},\mathrm{1})\mskip1.5mu]\mskip1.5mu]{}\<[E]%
\\
\>[9]{}\hsindent{2}{}\<[11]%
\>[11]{}\mathbin{!!}(\varid{fromEnum}\;\varid{a2})\mathbin{!!}(\varid{fromEnum}\;\varid{b2}))\;{}\<[E]%
\\
\>[3]{}\hsindent{6}{}\<[9]%
\>[9]{}\mykeyword{return}\;(\varid{a1},\varid{b1}){}\<[E]%
\\
\>[B]{}\hsindent{3}{}\<[3]%
\>[3]{}\keyword{if}\;\varid{isStay}\;\varid{a'}\mathop{\&\&}\varid{isStay}\;\varid{b'}\;\keyword{then}\;\mykeyword{return}\;(\varid{a},\varid{b}){}\<[E]%
\\
\>[B]{}\hsindent{3}{}\<[3]%
\>[3]{}\keyword{else}\;\varid{lreset}\mathbin{\$}\varid{game}\;\varid{a'}\;\varid{b'}{}\<[E]%
\ColumnHook
\end{hscode}\resethooks
\vspace{-30pt}
\end{wrapfigure}
\noindent
\ensuremath{\mykeyword{runSel}\mathbin{\$}\varid{game}\;(\conid{Move}\;\conid{Right})\;(\conid{Move}\;\conid{Right})} 
returns the strategies
\ensuremath{(\conid{Stay}\;\conid{Left},\conid{Stay}\;\conid{Left})} through 2 steps, indicating that both prisoners defect. This outcome
represents a Nash equilibrium, as neither prisoner can improve their individual
outcome by unilaterally changing their strategy.
It is easy to imagine an alternative handler that minimizes the total loss
for both players. In that case, the game would return \ensuremath{(\conid{Stay}\;\conid{Right},\conid{Stay}\;\conid{Right})},
which minimizes the combined loss.

\section{Denotational semantics}
\label{sec:denotational}
\newcommand{\lsem}[1]{\mathcal{L}\sem{#1}}
\newcommand{\Onebb}{\mathbbm{1}}
\newcommand{\opsi}{\overline{\psi}}
\newcommand{\V}{\mathrm{V}}
\newcommand{\EV}{\mathrm{EV}}
\newcommand{\mysim}{\preceq}
\newcommand{\gcut}[1]{}
We next give \cname a denotational semantics  using a suitably augmented selection monad.
We  give  soundness and adequacy theorems, thereby both showing that our computational ideas inspired by the selection monad are in fact in accord with it, and also providing a theoretical foundation for our combination of algebraic effect handlers and the selection monad.

\subsection{Semantics of Types} \label{sec:semtypes}

As discussed in Section~\ref{sec:selmon}, our semantics employs a family $S_\epsilon(X) = (X \to R_\epsilon) \to W_\epsilon(X)$ of augmented selection  monads where $W_\epsilon(X) = F_\epsilon(\R \times X)$ and $R_\epsilon = F_\epsilon(\R)$, with $R$ the reals. The $F_\epsilon$ are used to interpret unhandled effect operations, and $R_\epsilon$ is the free $W_\epsilon$-algebra on the one-point set.
The  $W_\epsilon$ are the commutative  combination \cite{HPP06} of the $F_\epsilon$  and the writer monad $\R \times -$. Algebraically this choice of monad combination corresponds to the  loss operation commuting with  the other operations. Semantically it results in loss effects commuting with 
operation calls; via the Soundness Theorem~\ref{thm:sound}, this is congruent with the operational semantics.
Also, recalling the discussion on subtyping in \cref{sec:typingrules}, note that the $S_\epsilon$ are not effect-covariant, as the $R_\epsilon$ appear contravariantly.

\newcommand{\ssem}[1]{{\color{blue}\mathcal{S}}\sem[blue]{#1}}
\newcommand{\N}{\mathbb{N}}

We define  $F_\epsilon(X)$ 
to be the least set $Y$ such that:
\[Y\; = \; \left ( \sum_{\ell \in \epsilon, \op: \outt \xrightarrow{\ell} \inn, 0 < i \leq \epsilon(\ell)} 
 \ssem{\outt} \times Y^{\ssem{\inn}}  \right ) + X \]
There is an inclusion $F_{\epsilon_1}(X) \subseteq F_{\epsilon}(X)$ if $\epsilon_1 \subseteq \epsilon$; which we use   without specific comment.  
The elements of $F_\epsilon(X)$ can be thought of as \emph{effect values} or \emph{interaction trees}, much as in~\cite{plotkin2001adequacy,xia2020interaction,FK19}. They are trees whose internal nodes are decorated with four things: an effect $\ell \in \epsilon$, an $\ell$-operation $\op: \outt \to \inn$, a handler execution depth index, and  an element of $\ssem{\outt}$. Nodes have successor nodes for each element of $\ssem{\inn}$; and the leaves of the tree are decorated with elements of $X$.  The idea is that such trees indicate possible computations in which various operations occur before finally yielding a value in $X$. Trees of this  kind were used to give a monadic denotational semantics to an algebraic effect language in~\cite{FK19}. 

Note the circularity  in these definitions: the $F_\epsilon$ are defined from the $S_\epsilon$,and vice versa. However the effect levels  strictly decrease in the first case
(because of the well-foundedness assumption)  and do not 
increase in the second (recall \Cref{sec:termination}), justifying the definitions.

Given the $S_\epsilon$ we can define $\ssem{\sigma}$ the semantics of types, as in  Figure~\ref{fig:tsem},
where we assume available a given semantics $\sem{b}$ of basic types, including 
$\ssem{\real} = R$.
\begin{figure}[h]
 \small{
\[\begin{array}{lcl}
\ssem{b} & = & \sem{b}\\
 \ssem{(\sigma_1,\dots,\sigma_n)} & = &  \ssem{\sigma_1} \times \dots \times  \ssem{\sigma_n} \\
  \ssem{\sigma \to \tau\etype \epsilon} & = &   \ssem{\sigma} \to   S_\epsilon(\ssem{\tau})\\
\end{array}\] }
\caption{Semantics of types}
\label{fig:tsem}
\end{figure}

\subsection{Monads} \label{sec:monads}

\newcommand{\mlet}[4]{\mathrm{let}_{#1}\; #2\; \mathrm{be}\; #3\; \mathrm{in}\; #4}
\newcommand{\myfold}[3]{\mathrm{fold}(#1,#2,#3)}
\newcommand{\comp}[2]{#1 \circ #2}

\gcut{
}

For our denotational semantics  we need the monadic structure of the $S_\epsilon$, which is available via the free algebra structures of the $W_\epsilon$ and the $F_\epsilon$.
Beginning with the $F_\epsilon$, say that an 
\textit{$\epsilon$-algebra} is a set $X$ equipped with %
 functions
\[\varphi_{\ell,\op,i} :  \ssem{\outt}  \times X^{\ssem{\inn}}  \to X\]
for $\ell \in \epsilon, \op\type \outt \xrightarrow{\ell} \inn$, and $0 < i \leq \epsilon(\ell)$. 

Then $F_{\epsilon}(X)$ is the free such algebra taking the functions to be: 
\[\varphi^X_{\ell,\op,i}(o,k) \eqdef ((\ell,\op,i),(o,k))\]
with the unit  at $X$ being given by $\eta_{F_\epsilon}(x) = x$, 
ignoring injections into sums. If $(Y,\psi_{\ell,\op,i})$ is another such algebra,
the unique homomorphic extension $f^{\dagger_{F_\epsilon}}$ of a function $f: X \to Y$ is given by setting 
\[f^{\dagger_{W_\epsilon}}(x) = f(x)
\qquad \quad  
\qquad \quad 
f^{\dagger_{F_\epsilon}} ((\ell,\op,i),(o,k)) = \psi_{\ell,\op,i}(o, \comp{f^{\dagger_{F_\epsilon}}}{k})\]

Turning to $W_\epsilon(X) = F_\epsilon(\R \times X)$, we can again see this as a free algebra monad. Say that an \emph{action $\epsilon$-algebra} is an $\epsilon$-algebra  $(Y,\psi_{\ell,\op,i})$ together with an additive action $\cdot: \R \times Y \to Y$ commuting with the $\psi_{\ell,\gamma,\op,i}$  (by an additive action we mean one such that $0\cdot y = y$ and $r \cdot (s\cdot y) = (r + s)\cdot y$). Then $F_\epsilon(\R \times X)$ is the free such algebra with operations $\varphi^{\R \times X}_{\ell,\op,i}$ and action 
given by:
\[r\cdot u \eqdef \mlet{F_\epsilon}{s\in \R, x\in X}{u}{(r+s,x)}\] 
The unit is $\eta_{W_\epsilon}(x) = (0,x)$, and if $(Y,\psi_{\ell,\op,i}, \cdot)$ is another such algebra, the unique homomorphic extension $f^{\dagger_{W_\epsilon}}$ of a function $f: X \to Y$ is given by $f^{\dagger_{W_\epsilon}}(r,x) = r\cdot f(x)$ and:
\[f^{\dagger_{W_\epsilon}} ((\ell,\op,i),(o,k)) = \psi_{\ell,\op,i}(o,\comp{f^{\dagger_{W_\epsilon}}}{k})\]

 Turning finally to the augmented selection monad $S_\epsilon(X) = (X \to R_\epsilon) \to W_\epsilon(X)$. The unit  $\eta_{S_\epsilon}$ at $X$ is given by $\eta_{S_\epsilon}(x) = \lambda \gamma 
\in X  \to R_\epsilon.\, \eta_{W_\epsilon}(x)$ (and recall that $\eta_{W_\epsilon}(x) = (0,x)$). For the Kleisli extension,
rather than  follow the definitions in, e.g.,~\citet{AP21} via a $W_\epsilon$-algebra on $R_\epsilon$ we give definitions that are a little easier to read. 

First, $R_\epsilon$ is an action $\epsilon$-algebra, with 
$\psi_{\ell,\op,i} :  \ssem{\outt}  \times R_\epsilon^{\ssem{\inn}}   \to R_\epsilon$
given by 
$\psi_{\ell,\op,i}(o,k) = 
\varphi^\R_{\ell,\op,i}(o,k)$
and action $\R \times R_\epsilon \to R_\epsilon$ given by:
$r\cdot u \eqdef \mlet{F_{\Onebb,\epsilon}}{s\in \R}{u}{r+s}$. 
Next (using  that $R_\epsilon$ 
is an action $\epsilon$-algebra) the loss $\ER{\epsilon}{F}{\gamma} \in R_\epsilon$ associated to %
$F \in S_\epsilon(Y)$ and loss function $\gamma\type Y \to R_\epsilon$  is
\[\ER{\epsilon}{F}{\gamma} \eqdef \gamma^{\dagger_{W_\epsilon}}(F(\gamma))\]
Then the Kleisli extension $f^{\dagger_{S_\epsilon}}\type S_\epsilon(X) \to S_\epsilon(Y)$  of a function $f\type X \to S_\epsilon(Y)$ is defined~by:
\begin{equation} \label{eqn:Kleisli}
    f^{\dagger_{S_\epsilon}}(F) = \lambda \gamma \in Y \to R_\epsilon.\,\mlet{W_\epsilon}{x \in X}{F(\lambda x\in X.\ER{\epsilon}{fx}{\gamma})}{fx\gamma}
\end{equation}
Finally $S_\epsilon(X)$ is an $\epsilon$-algebra, with functions
$\overline{\varphi}_{\ell,\op,i} :  \ssem{\outt}  \times S_\epsilon(X)^{\ssem{\inn}}   \to S_\epsilon(X)_\epsilon$ given by:
\[\overline{\varphi}^X_{\ell,\op,i}(o,f)(\gamma) = \varphi^{R \times X}_{\ell,\op,i}(o,\lambda x \in \inn.\, f(x)(\gamma))\]
\newcommand{\CL}[1]{\mathbf{L}[#1]}
\vspace{-20pt}
\subsection{Semantics of Expressions and Handlers} \label{sec:semexp}

Given an environment $\myG = x_1\type \sigma_1,\dots,x_n:\sigma_n$ we take
$\ssem{\myG}$ to be the functions (called \emph{environments}) $\rho$ on $\Dom(\myG)$ such that  
$\rho(x_i)  \in \ssem{\sigma_i}$, for $i = 1,\dots,n$. We give  semantics to typed expressions, and handlers according to the following schemes:
{\small\begin{mathpar}
\inferrule {\myG \vdash e\type\sigma\etype \epsilon}
  {\ssem{e} :  \ssem{\myG} \to S_\epsilon(\ssem{\sigma})}
\and
\inferrule{\myG \vdash h\type \parr,  \sigma \etype \epsilon\effl \Rightarrow \sigma'\etype\epsilon }
 {\ssem{h} :  \ssem{\myG} \to (\ssem{\parr} \times S_{\epsilon\ell}(\ssem{\sigma}))\to S_{\epsilon}(\ssem{\sigma'})}
\end{mathpar}}
Our denotational semantics of typed expressions is given in Figure~\ref{fig:esem}.
We  make use of an abbreviation, available for any monad $M$:
 {\small \[\mlet{M}{x\in X}{exp_1}{exp_2} \eqdef (\lambda x \in X.\, exp_2)^{\dagger_M}(exp_1)\]}
 for mathematical expressions $exp_1$ and $exp_2$. This abbreviation makes monadic binding available at the meta-level, and that makes for more transparent formulas. 
 We assume given semantics $\sem{c} \in \sem{b}$, for constants $c\type b$, and $\sem{f}\type \sem{\sigma}  \to \sem{\tau}$ for basic function symbols $f\type \sigma \to \tau$ 
(with $\sem{r} = r$ and $+$  denoting the  addition of $\R$).
We also use an auxiliary ``loss function" semantics. For $\myG, x\type \sigma \vdash e\type \real\etype \epsilon$ we define $\lsem{\lambda^\epsilon x\type \sigma.\,e}\type 
 \ssem{\myG} \to \ssem{\sigma} \to  R_\epsilon$ by:
{\small \[\lsem{\lambda^\epsilon x\type \sigma.\,e}(\rho) = 
\lambda a \in \ssem{\sigma}.\, 
   \mlet{F_\epsilon}{r_1,r_2 \in \R}{\ssem{e}(\rho[a/x]) (\lambda r \in \R.\,0)}{r_2}\]}\\[-2em]

\begin{figure}
\boxed{
\inferrule {\myG \vdash e\type\sigma\etype \epsilon}
  {\ssem{e} :  \ssem{\myG} \to S_\epsilon(\ssem{\sigma})}
  }
 \hfill \emph{(Expressions)}\\

\small{\[\begin{array}{lcll}
\ssem{x}(\rho)  & = & \eta_{S_\epsilon}(\rho(x))\\[0.1em]
\ssem{\con}(\rho)  & = & \eta_{S_\epsilon}(\sem{\con}) \\[0.1em]
\ssem{f(e)}(\rho) & = & 
  \mlet{{S_\epsilon}}{a \in \ssem{\sigma_1}}{\ssem{e}(\rho)}
     {\eta_{S_\epsilon}(\sem{f}(a))} \quad (f\type \sigma_1 \to \sigma)\\
\ssem{(e_1,\ldots,e_n)}(\rho) & = &
 \mlet{{S_\epsilon}}{a_1 \in \ssem{\sigma_1}}{\ssem{e_1}(\rho)}
                                                       {\\&& \dots \\&& \mlet{{S_\epsilon}}{a_n \in \ssem{\sigma_n}}{\ssem{e_n}(\rho)}{\\&&\eta_{S_\epsilon}((a_1,\dots,a_n))}} \qquad \qquad\qquad\qquad (\sigma = (\sigma_1,\dots, \sigma_n))
                                                       \\
\ssem{e.i}(\rho)  & = & S_\epsilon(\pi_i)(\ssem{e}(\rho))\\ 
\ssem{\lambda^{\epsilon_1} x\type \sigma.\, e}(\rho) & = & 
\eta_{S_\epsilon}(\lambda a \in \ssem{\sigma}.\, \ssem{e}(\rho[a/x]))\\

\ssem{e_1 ~ e_2}(\rho) & = &
 \mlet{{S_\epsilon}}
 {\varphi \in \ssem{\sigma_1} \to S_\epsilon(\ssem{\sigma})}{\ssem{e_1}(\rho)}
     {\\&&  \mlet{{S_\epsilon}}{a \in \ssem{\sigma_1}}{\ssem{e_2}(\rho)}{\varphi(a)}}\\
    &&  \hspace{50pt} (\Gamma \vdash e_1 \type \sigma_1 \to \sigma\etype \epsilon)\\[0.1em] 
\ssem{\op(e)}(\rho)  & = &
\mlet{S_\epsilon}
         {a \in \ssem{\outt}}
        {\ssem{e}(\rho)}
{\overline{\varphi}^X_{\ell,\op,\epsilon(\ell)}(a,(\eta_{S_\epsilon})_{\ssem{\inn}})}\\
&&\hspace{50pt} (\op\type \outt \xrightarrow{\effl} \inn, \sigma = \inn)\\
\ssem{\wph{h}{e_1}{e_2}}(\rho) & = & 
\mlet{{S_\epsilon}}{a \in \ssem{par}}{\ssem{e_1}(\rho)}{\ssem{h}(\rho)(a,\ssem{e_2}(\rho))}\\
 &&\hspace{50pt}  (\myG \vdash e_1\type par)  \\[0.1em]
\ssem{\loss(e)}(\rho) & = & \lambda \gamma \in  R_\epsilon^{\ssem{\sigma}}. \,
\mlet{{F_\epsilon}}{r\in R, a \in \R}{\ssem{e}(\rho)(\gamma)}{ (a + r,())}\\
\ssem{\then{\epsilon_1}{e_1}
  {\lambda^{\epsilon_1} x\type \sigma_1.\, e_2}}(\rho) 
&=& \lambda \gamma \in  R_\epsilon^{\ssem{\sigma}}.\,\\
&&\quad\;\mlet{F_\epsilon}{r_1 \in \R, a \in \ssem{\sigma_1}}{\ssem{e_1}(\rho)(\lsem{\lambda^\epsilon x\type \sigma_1.\, e_2}(\rho))}
{\\&&\;\quad \mlet{F_{\epsilon_1}}
                {r_2,r_3\in \R}
                {\ssem{e_2}(\rho[a/x])(\lambda r \in \R.\,0)}
                {(r_2,r_1 + r_3})}\\

\ssem{\glocal{\epsilon_1}{e}{\myg}}(\rho) & = & 
\lambda \gamma \in   R_\epsilon^{\ssem{\sigma}}.\,\ssem{e}(\rho)\lsem{\myg}(\rho)\\\
\ssem{\reset{e}}(\rho) & = & 
\lambda \gamma \in  R_\epsilon^{\ssem{\sigma}}.\,%
\mlet{F_\epsilon}
            {r_1 \in \R, a \in \ssem{\sigma}}
            {\ssem{e}(\rho)(\gamma)}
            {\eta_{W_\epsilon}(a)}%
\end{array}\]}

\caption{Semantics of expressions}
\label{fig:esem}
\end{figure}

\vspace{-10pt}
The denotational semantics of expressions, up to application, is generic to any monadic semantics of  call-by-value effectful languages. The semantics of operation calls uses the $\epsilon$-algebraic structure of the $S_\epsilon$. Other semantics read as  denotational versions of  operational computations. For example, the semantics of 
$\then{\epsilon_1}{e_1}{\lambda^\epsilon x\type \sigma_1.\, e_2}$ ignores the current loss 
continuation, passes the value $a$ of $e_1$ (with loss continuation the loss denotation of 
$\lambda^\epsilon x\type \sigma_1.\, e_2$) to $e_2$, evaluates that (with zero loss continuation),
and 
adds the loss $r_1$ of $e_1$ to the result $r_3$, keeping the resulting loss $r_2$. 
The sub-effecting  in the typing \rref{then} is here reflected in the semantic inclusion 
$F_{\epsilon_1}(X)\subseteq F_{\epsilon}(X)$.

\subsubsection*{{Semantics of Handlers}}
We build up the semantics of handlers in stages. Here is the high-level idea. Ignoring environments and parameters,  we seek a semantics:
$\ssem{h} : S_{\epsilon\ell}(\ssem{\sigma})\to S_{\epsilon}(\ssem{\sigma'})$,  
equivalently $\ssem{h}\gamma G \in   W_{\epsilon}(\ssem{\sigma'})$
for $\gamma: \ssem{\sigma'} \to R_\epsilon$, and $G \in S_{\epsilon\ell}(\ssem{\sigma})$. 
Following the standard approach to the semantics of handlers~\cite{pretnar2013handling} we exploit a free algebra property, here that of of $F_{\epsilon\ell}(R \times \ssem{\sigma})$, constructing an $\epsilon\ell$-algebra on $F_{\epsilon}(R \times \ssem{\sigma'})$ using $h$'s operation definitions (it may not be an action one),  then   obtaining a homomorphism to it from $F_{\epsilon\ell}(R \times \ssem{\sigma})$, and finally applying that to $G \gamma'$, with $\gamma'$  obtained from $\gamma$ using 
$h$'s return function. 

So, consider a handler $h$:
 {\small \[\left \{\begin{array}{l}
                     \op_1 \mapsto \lambda^\epsilon z\type (\parr, \outt_1,
                     (\parr,\inn_1) \to \real \etype \epsilon,
                     (\parr,\inn_1) \to \sigma' \etype \epsilon).~ e_1,\dots, \\
                     \op_n \mapsto \lambda^\epsilon z\type (\parr, \outt_n,
                     (\parr,\inn_n) \to \real \etype \epsilon,
                     (\parr,\inn_n) \to \sigma' \etype \epsilon).~ e_n,\\
                   \return \mapsto \lambda^\epsilon z\type (\parr,  \sigma).\, e_r
                            \end{array}\right \} \]}
where $\myG \vdash h\type \parr, \sigma\etype \epsilon\effl \Rightarrow \sigma'\etype\epsilon$, and fix  $\rho \in \ssem{\myG}$ and $\gamma \in \R_\epsilon^\ssem{\sigma'}$.
We first construct  an $\epsilon\ell$-algebra on $A = W_\epsilon(\ssem{\sigma'})^{\ssem{\parr}}$. 
So for $\ell_1 \in \epsilon\ell$, $\op\type \outt \xrightarrow{\ell_1} \inn$, and $0 < i \leq (\epsilon\ell)\ell_1$ we need functions 
$\psi_{\ell,\op,i}\type \ssem{\outt} \times  A^{\ssem{\inn}}  \to A$. 
For $\ell_1 \in \epsilon$, $\op\type \outt \xrightarrow{\ell_1} \inn$, and $0 < i \leq \epsilon(\ell_1)$, we set
{\small \[\psi_{\ell_1,\op,i}(o,k) = 
 \lambda p \in \ssem{\parr}.\, 
   ((\ell_1,\op,i), (o, \lambda a \in \ssem{\inn}.\, k\hspace{0.5pt}a\hspace{0.5pt}p))\]}
and, for $\op_j$
and  $i = \epsilon(\ell) +1$ we set
 {\small \[\psi_{\ell,\op_j,i}(o,k) = 
    \lambda p\in \ssem{\parr}.\, \ssem{e_j}(\rho[(p,o,l_1,k_1)/z])\gamma
 \]}
 where $k_1(p,a) = k\hspace{0.5pt}a\hspace{0.5pt}p$
 and $l_1(p,a) = \lambda \gamma_1 \in  \R_\epsilon^{\ssem{\sigma'}}.\, \delta_\epsilon(\gamma^{\dagger_{W_\epsilon}}(k\hspace{0.5pt}a\hspace{0.5pt}p))$.
 (in the definition of $l_1$ we use the fact that $R_\epsilon$ 
is an action $\epsilon$-algebra, and 
$\delta_\epsilon\type F_\epsilon(\R) \to F_\epsilon(\R \times \R)$
is %
$F_\epsilon(\lambda r\in \R.\, (0,r))$).\\[-1.15em]

  We  use this algebra to extend the map  $s\type \R \times \ssem{\sigma} \to A$ defined by
  {\small \[s(r,a) = \lambda p\in \ssem{\parr}.\, r\cdot(\ssem{e_r}(\rho[(p,a)/z])\gamma)\] }
  (Recall that $ \return \mapsto \lambda^\epsilon z\type (\parr,  \sigma).\, e_r$ is in $h$.)
  The semantics of the handler $h$ is then given by:
{\small  \[\ssem{h}(\rho)(p,G)(\gamma) = s^{\dagger_{F_{\epsilon\ell}}}(G(\lambda a \in \ssem{\sigma}.\, \ER{\epsilon}{\ssem{e}(\rho[(p,a)/z])}{\gamma}))(p)\]}\\[-3em]

\subsection{Soundness and Adequacy of Operational Semantics}
\newcommand{\vsem}[1]{{\color{red}\mathcal{V}}\sem[red]{#1}}

Below we may omit $\rho$ in $\ssem{e}(\rho)$ (or $\vsem{v}(\rho)$) when $e$ (respectively $v$) is closed. In \Cref{fig:vsem} we define a ``value semantics" $\vsem{v}\type \ssem{\myG} \to \ssem{\sigma}$ for values $\myG\vdash v\type \sigma$. It helps us to state our soundness and adequacy results. \\[-1.5em]
\begin{figure}[h]
    \centering
{\small \[\begin{array}{lcllcl}
\vsem{x}(\rho) &=& \rho(x) \qquad \qquad \qquad &
\vsem{(v_1,\dots,v_n)}(\rho) &=& (\vsem{v_1}(\rho),\dots,\vsem{v_n}(\rho)) \\
\vsem{c}(\rho) &=& \sem{c} &
\vsem{\lambda^{\epsilon_1} x\type \sigma_1.e}(\rho) &=& \lambda a \in \ssem{\sigma}.\, \ssem{e}(\rho[a/x])\\

\end{array}\]}
\caption{Semantics of values}
\label{fig:vsem}
\end{figure}
\vspace{-10pt} \begin{lemma} \label{lem:valsem} For any value $\myG \vdash v\type \sigma \etype \epsilon$ we have:
 $\ssem{v}(\rho) = \eta_{S_\epsilon}(\vsem{v}(\rho))$.
 \end{lemma}
As terminal expressions can  be stuck  we also need a lemma for their semantics:
 \begin{lemma} \label{lem:Kop} For terminal $\myG \vdash K[\op(v)]\type \sigma\etype\epsilon$ with $\op\type \outt \xrightarrow{\ell} \inn$
 we have:
 \[\ssem{K[\op(v)]}(\rho) = 
 \lambda \gamma.\,\varphi^{\R \times \ssem{\sigma}}_{\ell,\op,\epsilon(\ell)}
 (\vsem{v}(\rho),\lambda a \in \ssem{\inn}.\,
     \ssem{K[x]}(\rho[a/x])(\gamma))\]
 \end{lemma}

For soundness, we assume that the semantics of basic functions is sound w.r.t.\, the operational semantics, i.e.
$f(v) \rightarrow v' \implies \sem{f}(\sem{v}) = \sem{v'}$.
We have small-step soundness:
\begin{theorem} [Small-step Soundness] 
\label{thm:smallsound} 
Suppose  $e\type \sigma\etype \epsilon$ and  $\myg \type \sigma \to \real\etype \epsilon_1$ with $\epsilon_1 \subseteq \epsilon$. Then:
{\small \[\myg \vdash_\epsilon e \xrightarrow{r} e' \;\; \implies \;\; 
    \ssem{e}\lsem{\myg} = r\cdot(\ssem{e'}\lsem{\myg})\]}
\end{theorem}

and that implies evaluation (big-step) soundness:
\begin{theorem} [Evaluation soundness] 
\label{thm:sound} For all %
$e\type \sigma\etype \epsilon$ and  %
$\myg \type \sigma \to \real\etype \epsilon'$ with $\epsilon' \subseteq \epsilon$ we have:
{\small \begin{enumerate}
\item If $\myg \vdash_\epsilon e \xRightarrow{r} v$ then $ 
    \ssem{e}\lsem{\myg} = (r,\vsem{v})$.
\item If $
\myg \vdash_\epsilon e \xRightarrow{r} K[\op(v)]$ then
$ \ssem{e}\lsem{\myg} = r\cdot \ssem{K[\op(v)]}\lsem{\myg}$

     \end{enumerate}}
\end{theorem}
Combining soundness and termination (Theorem~\ref{thm:term}) 
we obtain adequacy:
\begin{theorem} [Adequacy] 
\label{thm:adequate} For all %
$e\type \sigma\etype \epsilon$ and %
$\myg \type \sigma \to \real\etype \epsilon'$ with $\epsilon' \subseteq \epsilon$ we have:
{\small \begin{enumerate}

\item
If $ \ssem{e}\lsem{\myg} = (r,a)$ then, for some $v$,
    $\myg \vdash_\epsilon e \xRightarrow{r} v$ and $\vsem{v} = a$.
\item
If $
\ssem{e}\lsem{\myg} = \varphi^{\R \times \ssem{\sigma}}_{\ell,\op,\epsilon(\ell)} (a,f)$
then, for some $K[\op(v)]$,  $\myg \vdash_\epsilon e \xRightarrow{r} K[\op(v)]$, 
$a = \vsem{v}$,\\ and $f = \lambda b \in \ssem{\inn}.\,
    r\cdot \ssem{K[x]}(x \mapsto b)\lsem{\myg}
 $.
\end{enumerate}}
\end{theorem}

\gcut{\begin{proof}
Suppose that $\ssem{e}\lsem{\myg} = (r,a)$. By the termination theorem, Theorem~\ref{thm:term}, either there are $r'$ and $v$ such that
$\myg \vdash_\epsilon e \xRightarrow{r'} v$ or else there are $r'$ and $K[\op(v)]$ such that
$\myg \vdash_\epsilon e \xRightarrow{r'} K[\op(v)]$. 
Using the second part of Theorem~\ref{thm:sound}, we see that the latter cannot happen, as then $\ssem{e}\lsem{\myg}$ would not have the form $(r,a)$.
So there are $r'$ and $v$ such that
$\myg \vdash_\epsilon e \xRightarrow{r'} v$.
Then, by evaluation soundness, we have $\ssem{e}\lsem{\myg} = (r',\vsem{v})$. So $(r,\vsem{v}) = (r',a)$ and the conclusion follows.

The second part is proved similarly.
\end{proof}}
As usual, if $\sigma$ is first-order and the 
denotation map $\sem{c}$ of constants is 1-1, 
we have the corollary:\\[-1em]
\[\myg \vdash_\epsilon e \xRightarrow{r} v \;\; \iff \;\; 
\ssem{e}\lsem{\myg} = (r,\vsem{v})\]

Fixing $\sigma$ and $\epsilon$, set $E = \{e:\sigma\etype\epsilon\}$, and let 
 $\EV$, the set of \textit{effect values},
be the least set such that:
{\small \[\EV\;\; = \;\; 
\sum_{\ell \in \epsilon, \op: \outt \xrightarrow{\ell} \inn} 
 \hspace{-15pt} V_\outt \times \EV^{V_\inn}  %
  \;+ \;(R \times V_\sigma) \]}%
where, for any $\tau$,  $V_\tau \eqdef \{v|v\type\tau\}$. Following~\cite{plotkin2001adequacy,plotkin2009adequacy}, 
we  evaluate  expressions as far as effect values.
Fixing $\myg\type \sigma \to \loss\etype \epsilon'$ (with $\epsilon'\subseteq\epsilon$)
define (using the evident $R$-action on $\EV$) a \textit{giant step} evaluation function
$\mathrm{Eval}: \mathrm{E} \rightharpoonup \EV$ (shown total via computability)  by:
{\small \[\mathrm{Eval}(e) = 
   \left \{
   \begin{array}{ll}
      (r,v)  & (\myg \vdash_\epsilon e \xRightarrow{r} v) \\
       ((\ell,\op),(v, \lambda w \in \V_{\inn}.\, r\cdot\mathrm{Eval}(K[w])) & 
        (\myg \vdash_\epsilon e \xRightarrow{r} K[\op(v)], \op: \outt \xrightarrow{\ell} \inn)
   \end{array}
   \right .\]}
Next let $\mysim$ be the least relation between $\EV$ and $W_\epsilon(\ssem{\sigma})$
such that
(1) $(r,v) \mysim (r,\vsem{v})$
and
(2)$((\ell,\op),(v,f)) \mysim ((\ell,\op,\epsilon(\ell)),(\vsem{v},g))$ if
$\forall w \in V_\inn. f(w) \mysim g(\vsem{w})$.

\begin{theorem}[Giant Step Adequacy] \label{thm:giant} For all $e\type \sigma\etype \epsilon$ we have: $\mathrm{Eval}(e) \mysim \ssem{e}\lsem{\myg}$.
\end{theorem}

\section{Related Work and Conclusion}
\label{sec:conclusion}

Our work may be the first  advocating a %
language design based on effect handlers and the selection monad.
It is  closest to \citet{AP21,AP23} who  
used an $\mathsf{argmax}$ selection function to make their choices.
As they themselves 
said, this is unreasonable when there is %
no 
access to optimal strategies.
Further, neither handlers nor choice continuations were provided
(though they did suggest trying monadic reflection and reification~\cite{filinski1994representing}).
\citet{LGG23} proposed  using  effect handlers 
for \textit{reinforcement learning} (RL)~\cite{sutton2018reinforcement},  but did not support choice continuations. 
Basic RL %
(e.g., \textit{multi-armed bandits} %
as in \citet{LGG23}),
does not need choice continuations as action losses  are directly given, and can be transmitted to a user-defined loss effect (and
 ordinary state effects can be used to represent learner's states).
More sophisticated RL algorithms benefit from choice continuations, e.g., \textit{deep reinforcement learning}~\cite{riedmiller2005neural} where  policies are  approximated by neural networks, and so need 
gradient descent.

There are several directions for future work.
First, we are interested in improving the performance of the selection monad.
Specifically, the choice continuation shares expressions with the delimited continuation,
(though this need not  lead to recomputations).
For instance, in the gradient descent handler \ensuremath{hOpt} in \Cref{sec:examples},
\ensuremath{l} is differentiated with respect to the current parameter,
and \ensuremath{k} is resumed with the updated parameter.
In broader scenarios, we expect that further program transformations and advanced compiler optimizations (e.g., memoization) will mitigate recomputations.
Moreover, a more efficient approach to jointly make nested choices is described
in \citet{hartmann2024towards},
using a %
generalized form of selection monad.
We  would also like to integrate our design into existing languages and frameworks.
e.g., JAX~\cite{jax2018github}, a functional programming DSL popular for large-scale
ML
tasks %
(see ~\cite{jaxhandler}).
Interesting too are  frameworks providing choices for users, such as \citet{smartChoices2018}. There are several interesting possibilities for advancing our framework: adding recursive functions 
or iteration; adding effect polymorphism as in~\citet{L17}; obtaining subeffecting using effect inclusions $\epsilon'\subseteq  \epsilon$ to type expressions  yielding $\epsilon$ results from  $\epsilon'$-continuations; and allowing users to locally vary the reward monoid (e.g., to a product with  independent localising constructs, facilitating multi-objective optimization). %

\begin{acks}
We thank Mart\'{i}n Abadi, Adam Paszke, Dimitrios Vytiniotis, and Dan Zheng
for helpful discussion, and  the reviewers for their helpful comments.
\end{acks}

 \bibliographystyle{ACM-Reference-Format}
 \bibliography{Languagesforlearning.bib}

 \newpage
 \appendix
\section{The core language}
\label{appendix:sec:source}
\newcommand{\acut}[1]{#1}
\newcommand{\Lsc}{L_{\sigma,\epsilon}}
In this Appendix we consider \cname in full, presenting proofs of: determinism, progress, and safety of the operational semantics (Theorem~\ref{Athm:dps}), the fundamental lemma (Lemma~\ref{Alem:comp}), termination (Theorem~\ref{Athm:term}), and  soundness and adequacy of the operational semantics 
(Theorems~\ref{Athm:sound},~\ref{Athm:adequate}, and~\ref{Athm:giant}).
\subsection{\cname syntax}

\[\begin{array}{rclrcl}
    \sigma,\tau & \BEQ & b  \BOR (\sigma_1,\dots,\sigma_n) \BOR \sigma + \tau
                         \BOR \nat \BOR \mylist(\sigma)
                         \BOR(\sigma \to \tau\etype \epsilon) \\
 \Eff & \BEQ & \{\overline{\effl_i : \Op(\effl_i)} \} \\
    \epsilon & \BEQ & \{\} \mid \epsilon \effl  \\
    \Op(\effl) & \BEQ & \{ \overline{\op_i \type \outt_i \to \inn_i} \}\\
  \end{array} \]

\[\begin{array}{lcl}
    e & \BEQ & \;  \con\BOR f(e)   \BOR x \\%
       &&\BOR  (e_1,\ldots,e_n) \BOR e.i \\
       &&\BOR  \inl{\sigma}{\tau}(e) \BOR \inr{\sigma}{\tau}(e)\BOR \mycases{e}{x_1\type \sigma_1}{e_1}{x_2\type \sigma_2}{e_2}\\
       &&\BOR \zero \BOR \mysucc{e} \BOR\iter{e_1}{e_2}{e_3} \\
       && \BOR \nil_\sigma \BOR \cons{e_1}{e_2} \BOR \fold{e_1}{e_2}{e_3}\\
      &&                \BOR  \lambda^\epsilon x\type \sigma.\, e \BOR e_1 ~ e_2 \\
      &&\BOR \op(e)  \BOR \loss(e) \BOR \wph{h}{e_1}{e_2} \\
      && \BOR \then{\epsilon}{e_1}{\lambda^\epsilon x\type \sigma.\, e_2} \BOR \glocal{\epsilon}{e}{\myg} \BOR\reset{e} %
    \\[1em]

 \myg & \BEQ  &  \lambda^\epsilon x\type \sigma.\, 0 \BOR 
    \lambda^{\epsilon} x\type \sigma.\, \then{\epsilon}{e}{\myg}
    \\[5pt]

    h & \BEQ & \left \{\begin{array}{l}
                     \op_1 \mapsto \lambda^\epsilon z\type (\parr, \outt_1,
                     (\parr,\inn_1) \to \real \etype \epsilon,
                     (\parr,\inn_1) \to \sigma' \etype \epsilon).~ e_1,\dots, \\
                     \op_n \mapsto \lambda^\epsilon z\type (\parr, \outt_n,
                     (\parr,\inn_n) \to \real \etype \epsilon,
                     (\parr,\inn_n) \to \sigma' \etype \epsilon).~ e_n,\\
                   \return \mapsto \lambda^\epsilon z\type (\parr,  \sigma).\, e
                            \end{array}\right \} \\
                            && \hspace{100pt} (\op_1,\dots, op_n \mbox{ enumerates some } \Op(\ell)) \\[2em]

 \end{array}\]

Handlers $h$ include a list of operation definitions and a return definition;
$h$  \emph{handles $\effl$} if this list enumerates $\Op(\effl)$.

\subsection{Typing Rules}

\small{
\drules[]{$\myG\vdash e\type \sigma \etype  \epsilon $}{Typing Expressions}
{const,fun,var
  ,prd,prj,inl,inr,cases,zero,succ,iter,nil,cons,fold
  ,abs,app,op,loss,handle, then, glocal, reset
}
\drules[]{$\myG \vdash h\type \parr,  \sigma \etype \epsilon\effl \Rightarrow \sigma'\etype\epsilon  $}{Typing Handlers}
{handler}

}
\subsection{Operational Semantics}

The syntactical classes used for operational semantics:

  \begin{center}
  \begin{tabular}{lrcl}
    value & $v $ & $\BEQ $ & $ x \BOR c$ \\
    && $ \BOR$ & $(v_1,\dots,v_n ) \BOR \inl{\sigma}{\tau}(v) \BOR \inr{\sigma}{\tau}(v)$\\
    && $    \BOR$ & $ \zero \BOR \mysucc{v}\BOR \nil_\sigma \BOR \cons{v_1}{v_2} $\\
    && $\BOR$ & $\lambda^\epsilon x\type \sigma.\, e$\\

    regular frame & $F$ & $\BEQ$ & $f(\square) $\\
         && $\BOR$& $ (v_1,\dots, v_k, \square, e_{k+2},\dots, e_n) \BOR \square.i$ \\
          &&$ \BOR $& $ \inl{\sigma}{\tau}(\square) \BOR \inr{\sigma}{\tau}(\square)$\\
         && $   \BOR $& $\mycases{\square}{x_1\type \sigma_1}{e_1}{x_2\type \sigma_2}{e_2}$\\
         &&$ \BOR $& $\mysucc{\square}$\\
         && $\BOR$ & $\iter{\square}{e_2}{e_3}\BOR \iter{v_1}{\square}{e_3} \BOR \iter{v_1}{v_2}{\square}$\\
          &&$ \BOR $& $\cons{\square}{e_2} \BOR \cons{v_1}{\square}$\\
          && $\BOR$ &$\fold{\square}{e_2}{e_3} \BOR  \fold{v_1}{\square}{e_3} \BOR \fold{v_1}{v_2} {\square}$\\
          && $\BOR $&$\square~e \BOR v~\square $\\
          &&$ \BOR$& $ \loss(\square) 
          \BOR \op(\square) \BOR \wph{h}{\square}{e}$ \\
  special frame &$S$ & $\BEQ$ & $\wph{h}{v}{\square} 
             \BOR \then{\epsilon}{\square}{(\lambda^{\epsilon} x\type \sigma.e)}
             \BOR \glocal{\epsilon}{\square}{\myg} %
             \BOR \reset{\square}
  $\\
  cont context & $K$ & $\BEQ$ & $\square \BOR F[K]  \BOR S[K]$\\
        stuck expr & $u$  & $\BEQ$  & $ K[\op(v)]\quad (\op \notin \hop(K)) $ \\

  terminal expr & $w$  & $\BEQ$  & $v \BOR u$ \\

 redex & $R$ & $\BEQ$ &$ f(c)  \BOR v.i \BOR \mycases{v}{x_1\type \sigma_1}{e_1}{x_2\type \sigma_2}{e_2}$\\%
               && $\BOR $ & $\iter{v_1}{v_2}{v_3} \BOR \fold{v_1}{v_2}{v_3} \BOR v_1~v_2$ \\
                    && $\BOR$ & $\loss(v) \BOR \wph{h}{v_1}{v_2} $\\
                    && $\BOR$ & $ \wph{h}{v_1}{K[\op(v_2)]} \quad (\op \notin \hop(K), \op \in h) $ \\
          && $\BOR$ & $   %
                      \then{}{v}{\lambda^\epsilon x\type \sigma.e_1}
                      \BOR \glocal{\epsilon}{v}{\myg}
                      \BOR \reset{v}$
    \\
  \end{tabular}
  \end{center}

We write $\heff(K)$ for the multiset of effect labels that $K$ handles;
it is defined inductively with main clause $\heff(\wph{h}{v}{K'}) = \heff(K')\ell$, where $h$ handles $\ell$; we further define  $\hop(K)$ to be $\{\op \in \Op(\ell)| \ell \in \heff(K)\}$, the set of operations  handled by $K$.

A well-typed value $\myG \vdash v\type\sigma$,  can be typed with any effect: $\myG \vdash  v\type\sigma\etype\epsilon$. This fact is used to prove the following
 ``cut" lemma. We omit the (straightforward) proofs of some lemmas in this appendix.
\begin{lemma} \label{Alem:cut} We have:
\[\frac{\myG \vdash v\type \sigma \qquad \myG, x\type \sigma \vdash e\type \tau \etype \epsilon}
{\myG \vdash e[v/x] \type \tau \etype \epsilon}\]
\end{lemma}

Continuation contexts behave well with respect to typing:
\begin{lemma} \label{Alem:ctx} Suppose that $\vdash K[e]\type \tau\etype \epsilon$, with $K$ handling $\epsilon'$.
Then, for a unique $\sigma$ we have  $\vdash e\type \sigma\etype \epsilon,\epsilon'$ and then for any $\myG \vdash e_1\type \sigma\etype \epsilon,\epsilon'$ we have $\myG \vdash K[e_1]\type \tau\etype\epsilon$.
\end{lemma}

\begin{lemma}[Analysis] \label{Alem:eal} Every expression has exactly one of the following five forms:
\begin{enumerate}
\item a value $v$,
\item a stuck expression $K[\op(v)]$, for a unique $K$, $\op$, and $v$, 
\item a redex $R$, 
\item $F[e]$, for a unique $F$ and $e$, with $e$ not a value or stuck, or
\item $S[e]$, for a unique $S$ and $e$, with $e$ not a value or stuck.
\end{enumerate}
\end{lemma}
\begin{proof}
The proof is by cases on the form of expressions. 
We consider one case as an example, an application $e_1~e_2$.

First we show that the application can always be analysed into one of the 5 forms.
If $e_1$ is not a value or a stuck expression, then $e_1~e_2$ is of the form 
$F[e_1]$, with $e_1$ not a value or stuck, where $F = \square~e_2$.
If $e_1$ is stuck then so is $e_1~e_2$. Otherwise 
$e_1$ is a value, say $v_1$.
Then If $e_2$ is not a value or a stuck expression, then $e_1~e_2$ is of the form 
$F[e_2]$, with $e_1$ not a value or stuck, where $F = v_1~\square$.
If $e_2$ is stuck then so is $v_1~e_2$. 
If $e_2$ is a value, say $v_2$ then $e_1~e_2$ is a redex $v_1~v_2$.

Next we show that the analysis is unique. As it is an application, the only possible forms it can have are
$v_1~v_2$ for values $v_1 = e_1$ and $v_2 = e_2$; $F_1[e_1]$, where $F_1 = \square~e_2$, or $F_2[e_2]$, where 
$F_2 = v_1~\square$ for a value $v_1 = e_1$. 

Now, if $e_1$ is not a value, the only possible form $e_1~e_2$ can have is the second of 
these, viz $F_1[e_1]$. If $e_1$ is stuck, then $e_1~e_2$ cannot have the form $F_1[e_1]$ with $e_1$ not a value or stuck. So the only possible analysis is being stuck. If, on the other hand, $e_1$ is not stuck, then neither is $e_1~e_2$, so the only possible analysis is as $F_1[e_1]$ with $e_1$ not a value or stuck.

Suppose instead that $e_1$ is a value $v_1$. Then $e_1~e_2$ can't have the analysis $F_1[e_1]$ with $e_1$ not a value or stuck, 
So we can assume $e_1~e_2$ has one of the  forms  $v_1~v_2$ or  $F_2[e_2]$.
If $e_2$ is a value then $e_1~e_2$ cannot have the form $F_2[e_2]$ with $e_2$ not a value or stuck. 
So the only possible analysis is the remaining form $v_1~v_2$. 
Suppose instead  that $e_2$ is not a value.If it is not stuck, then neither is $v_1~e_2$, so the only analysis it can have is 
$F_2[e_2]$ with $e_2$ not a value and not stuck. On the other hand if it is stuck then it cannot have that form, and the only possible analysis  it can have is being stuck.

\end{proof}
\newcommand{\ket}[1]{\lvert #1\rangle}

\begin{figure}[H]
    \centering

\small{
\[\begin{array}{llcll}
 \qquad \qquad &

\myg \vdash_\epsilon f(v) & \evalto{0} & v' & (f(v) \to v') \\
 
& \myg \vdash_\epsilon  (v_1,\dots,v_n).i & \evalto{0} & v_i \\

& \myg \vdash_\epsilon   \mycases{\inl{\sigma_1}{\sigma_2}(v)}{x_1\type \sigma_1}{e_1}{x_2\type \sigma_2}{e_2} & \evalto{0} & e_1[v/x_1] \\

 & \myg \vdash_\epsilon   \mycases{\inl{\sigma_2}{\sigma_2}(v)}{x_1\type \sigma_1}{e_1}{x_2\type \sigma_2}{e_2} & \evalto{0} & e_2[v/x_2] \\

 &  \myg \vdash_\epsilon   \iter{0}{v_2}{v_3} & \evalto{0} & v_2 \\

 & \myg \vdash_\epsilon   \iter{\mysucc{v_1}}{v_2}{v_3} & \evalto{0} &
v_3 \iter{v_1}{v_2}{v_3} \\

 & \myg \vdash_\epsilon  \fold{\nil_\sigma}{v_2}{v_4}  &  \evalto{0}  &  v_2 \\

 & \myg \vdash_\epsilon  \fold{\cons{v_1}{v_2}}{v_3}{v_4} & \evalto{0} &  
    v_4(v_1,\fold{v_2}{v_3}{v_4})  \\

& \myg \vdash_\epsilon (\lambda^\epsilon x\type \sigma.\, e)~v  & \evalto{0} & e[v/x] \\

 & \myg \vdash_\epsilon  \loss(r) &  \evalto{r} &  () \\[10pt]

 & \multicolumn{3}{c}{ \inferrule*[narrower=0.7]{ v_1:par \\ \op \notin \hop(K) \\ \op\type \outt \xrightarrow{\ell} \inn \\ \op \mapsto v_o \in h\\
 \mbox{$h$ has effect $\epsilon$} \\ 
 f_k = \lambda^\epsilon(p,y)\type (par,in).\,\glocal{\epsilon}{\wph{h}{p}{K[y]}}{\myg}\\
 f_l = \lambda^\epsilon (p,y)\type (par,in).\,\then{\epsilon}{(\wph{h}{p}{K[y]})}{\myg}}
 {\myg \vdash_\epsilon\wph{h}{v_1}{K[\op(v_2)]}  \evalto{0} v_o(v_1,v_2,f_l,f_k)} } \\[10pt]
 
   & \myg \vdash_\epsilon  \wph{h}{v_1}{v_2} & \evalto{0} & v_r(v_1,v_2) & (\return \mapsto v_r \in h) \\

 &  \myg \vdash_{\epsilon} \then{\epsilon_1}{v}{\lambda^{\epsilon_1} x\type \sigma. e}
  &\evalto{0} & \glocal{\epsilon_1}{e[v/x]}{\lambda^{\epsilon_1} x: \sigma.\, 0} \\
  
  & \myg \vdash_{\epsilon} \glocal{\epsilon_1}{v}{\myg_1} & \evalto{0} & v & 
  \\
  
  & \myg \vdash_{\epsilon} \reset{v} & \evalto{0} & v

\\[15pt]

 & \multicolumn{3}{c}{\inferrule{
                     \lambda^\epsilon x\type \tau.\, (\then{\epsilon}{F[x]}{\myg}) \vdash_\epsilon e \evalto{r} e' }
{\myg \vdash_\epsilon F[e] \evalto{r} F[e']}} \\ [20pt]

 & \multicolumn{3}{c}{\inferrule*[narrower=0.7]{\mbox{$h$ has effect $\epsilon$}\\\return \mapsto v_r \in h \quad v_r\type (\parr,\sigma) \to \sigma'\etype \epsilon \\
 \mbox{$h$ handles $\ell$}\\
                     \lambda^\epsilon x\type \sigma.\, (\then{\epsilon}{v_r(v,x)}{\myg}) \vdash_{\epsilon\ell} e \evalto{r} e' }
{\myg \vdash_\epsilon \wph{h}{v}{e} \evalto{r} \wph{h}{v}{e'}}} \\ [20pt]

& \multicolumn{3}{c}{\inferrule{\myg_1 \vdash_\epsilon e \evalto{r} e'}
         {\myg \vdash_\epsilon  (\then{\epsilon}{e}{\myg_1}) \evalto{\; 0 \;} r +  
            (\then{\epsilon_1}{e'}{\myg_1})}} \\[25pt]

& \multicolumn{3}{c}{\inferrule{ \vdash e \type \sigma\etype \epsilon_1 
\\ \myg_1 \vdash_{\epsilon_1} e \evalto{r} e'}
{\myg \vdash_{\epsilon} \glocal{\epsilon_1}{e}{\myg_1}\evalto{r} \glocal{\epsilon_1}{e'}{\myg}}}
\\[25pt]

& \multicolumn{3}{c}{\inferrule{\myg \vdash_\epsilon e \evalto{r} e'}
         {\myg \vdash_\epsilon  \reset{e} \evalto{\; 0 \;} \reset{e'}}} \\
\end{array}
\]}
    \caption{Small step op sems}
    \label{ssopsems}
\end{figure}
\newcommand{\thisfigurecut}[1]{}

\thisfigurecut{\begin{figure*}
 \begin{mathpar}

 \Infer{handle 1}{
   {\lambda^\epsilon x\type \parr.~\then{\epsilon}{(\wph{h}{x}{e_2})} {\myg } \vdash e_1 \xRightarrow{r_1} v_1}\\\\
   {\lambda^\epsilon y\type \sigma.~\then{\epsilon}{(\wph{h}{v_1}{y})} {\myg } \vdash e_2 \xRightarrow{r_2} v_2}\\\\
   \myg\vdash e[v_1/p, v_2/x]  \xRightarrow{r_3} w}{
   \myg\vdash \wph{h}{e_1}{e_2} \xRightarrow{r_1 + r_2 + r_3}  w
 }\\
 \hspace{100pt} (\return(p\type \parr, x\type \sigma) \mapsto e \in h)\\\\

\and
 \Infer{handle 2}{
   {\lambda^\epsilon x \type \parr.~\then{\epsilon}{(\wph{h}{x}{e_2})} {\myg } \vdash e_1 \xRightarrow{r_1} v_1}\\\\
    {\lambda^\epsilon y \type \sigma.~\then{\epsilon}{(\wph{h}{v_1}{y})} {\myg } \vdash e_2 \xRightarrow{r_2} E[\op(v)]}\\\\
   \myg\vdash e [v_1/p, v_2/x, f_l/\kl,  f_k/\kk] \xRightarrow{r_3} w}{
   \myg\vdash \wph{h}{e_1}{e_2} \xRightarrow{r_1 + r_2 + r_3}  w
 }\\
 \hspace{50pt} (\op(p\type \parr, x\type \outt,\kl \type (\parr,\inn)\to \real\etype \epsilon,\kk \type (\parr, \inn)\to
\sigma'\etype \epsilon) \mapsto  e \in h,\\
   f_k = (\lambda^\epsilon (p, y)\type (\parr, \inn).~\wph{h}{p}{E[y]}), \\
   f_l = \lambda^\epsilon (p, y)\type (\parr, \inn).~\then{\epsilon}{f_k(p,y) } {\myg })\\\\

 \and

\Infer{gen op}{\lambda^\epsilon x\type \sigma.\, \then{\epsilon}{F[x]} {\myg } \vdash e \xRightarrow{r} E[\op(v)]}
                      {\myg\vdash F[e] \xRightarrow{r} F[E[\op(v)]]} \quad (F \mbox{ does not handle } \op, \vdash e\type \sigma)
\and
  \Infer{reset op}{
    \myg\vdash  e \xRightarrow{r} E[\op(v)]
  }{
    \myg\vdash \reset{e} \xRightarrow{0} \reset{E[\op(v)]}
  }

  \and
  \Infer{local op}{
    0_\sigma \vdash  e \xRightarrow{r} E[\op(v)]
  }{
    \myg\vdash \local{\epsilon}{e} \xRightarrow{r} \local{\epsilon}{E[\op(v)]}
  } \quad (\vdash e\type \sigma)\\\\

\end{mathpar}
 \caption{Operational semantics rules for \cname (cntnd)}
\label{Copsem3}
\end{figure*}}

\newpage
\begin{theorem} \label{Athm:dps}
\myskip
\begin{enumerate}
    \item (\textit{Termination}) If  $e$ is terminal, then it can make no transition, i.e., $\myg \vdash_{\epsilon} e \xrightarrow{r} e'$ holds for no $\myg, r, \epsilon$, and $e'$. 
   \item (\textit{Determinism}) If  $~\myg \vdash_{\epsilon} e \evalto{r} e'$ and  $\myg \vdash_{\epsilon} e \evalto{r'} e''$ then $r = r'$ and $e' = e''$.
   \item (\textit{Progress})  If $e\type \sigma\etype \epsilon_1$ is non-terminal, then
     $\myg \vdash_{\epsilon_1} e \evalto{r} e'$  holds for some $r$ and $e'$ for any $~\myg\type \sigma \to \real\etype \epsilon_2$ with $\epsilon_2 \subseteq \epsilon_1$.
   \item (\textit{Type safety})  If $~\myg\type \sigma \to \real\etype \epsilon_2$,  
        $\myg \vdash_{\epsilon_1} e \evalto{r} e'$, with $\epsilon_2 \subseteq \epsilon_1$,
       and $e\type \sigma\etype \epsilon_1$ then $e'\type \sigma\etype \epsilon_1$.
\end{enumerate}

\end{theorem}
\begin{proof}
 We first prove termination. The proof is by induction on $e$. By the analysis lemma $e$ is not a redex, so no redex rule can apply. 
It can have the form $F[e_1]$ or $S[e_1]$, but by the unique analysis lemma $e_1$ must be terminal, so by the induction 
hypothesis $e_1$ cannot make a move, so no context rule can apply to $e$ either.

 We next prove determinism by induction on expressions $e$. We split into cases using the analysis lemma.
If $e$ is terminal, then,as we have just seen, $e$ cannot make any transition.
If $e$ is a redex, then, by the analysis lemma it is so uniquely. 
There is then a unique rule that can apply and it is deterministic (all rules are). Finally if $e$ has 
one of the forms $F[e_1]$ or $S[e_1]$, then there is at most one context rule that can apply. 
By the induction. hypothesis,
$e_1$ can make at most one transition, and so then $e$ can only make at most one transition as context rules 
are deterministic in the transitions the sub-expressions can make.

 Suppose next that, for a non-terminal $e$, we have $~\myg\type \sigma \to \real\etype \epsilon_2$ and $e\type \sigma\etype \epsilon_1$ with $\epsilon_2 \subseteq \epsilon_1$. We  show that 
 $\myg \vdash_{\epsilon_1} e \evalto{r} e'$  for some $r$ and $e'$ with $e'\type \sigma\etype\epsilon_1$. This  proves progress and type safety.
 
 We first consider the cases where $e$ is a redex:
 
 \begin{enumerate} 
 \item Suppose $e$ the form $f(v)$. 
 For some $\sigma_1$ and $\tau$ we have  $f\type \sigma_1 \to \tau$. As $e\type \sigma\etype \epsilon_1$ we have 
 $v\type \sigma_1$ and $\sigma = \tau$. So, for some value $v' \type \tau$,  
 we have $f(v) \to v'$. So $\myg \vdash_\epsilon e \xrightarrow{0} v'$ 
 and $v'\type\sigma\etype \epsilon_1$.  
 
\item Suppose $e$ has the form $v.i$. As $e\type \sigma$, we have $\sigma = (\sigma_1,\dots, \sigma_n)$
and $v = (v_1,\dots,v_n)$, for some $v_1\type\sigma_1,\dots, v_n\type\sigma_n$. Then
$\myg\vdash_{\epsilon_1}v.i \xrightarrow{0} v_i$ and 
$v_i\type \sigma$.
\item Suppose $e$ has the form $\mycases{\inl{\sigma_1}{\sigma_2}(v)}{x_1\type \sigma_1}{e_1}{x_2\type \sigma_2}{e_2}$.
 Then $\myg \vdash_{\epsilon_1} e \xrightarrow{0} e_1[v/x_1]$. As $e\type\sigma\etype\epsilon_1$ we have $v\type \sigma_1$ and $x_1\type \sigma_1\vdash e_1\type \sigma\etype \epsilon_1$. So by Lemma~\ref{Alem:cut} we have $e_1[v/x_1]\type\sigma\etype\epsilon_1$.

\item Suppose $e$ has the form $ \iter{0}{v_2}{v_3}$. Then we have  
$\myg \vdash_{\epsilon_1}  e  \evalto{0}  v_2$. As $e\type \sigma\etype \epsilon_1$ we have $v_2\type \sigma$.

\item Suppose $e$ has the form    $\fold{\cons{v_1}{v_2}}{v_3}{v_4}$.
Then $\myg \vdash_{\epsilon_1} e \evalto{0} v_4(v_1,\fold{v_2}{v_3}{v_4})$. 
As $e\type \sigma \etype \epsilon_1$, for some $\tau$ we have  
$v_1\type \tau$,
$v_2\type \mylist(\tau)$, 
$v_3\type \sigma$, 
and $v_4\type (\tau,\sigma) \to \sigma\etype \epsilon_1$.
So we see that $\fold{v_2}{v_3}{v_4})\type \sigma \etype \epsilon_1$, and so that
$v_4(v_1,\fold{v_2}{v_3}{v_4})\type \sigma \etype \epsilon_1$.

\item Suppose that $e$ has the form $v_1~v_2$. As $e\type\sigma\etype \epsilon_1$ we have $v_1\type \tau \to \sigma\etype\epsilon_1$ and $v_2\type\tau$ for some $\tau$. So $v_1$ has the form $\lambda^{\epsilon_1} x\type \tau.\, e_3$. So $\myg \vdash_{\epsilon_1} e \xrightarrow{0} e_3[v_2/x]$. As 
$\lambda^{\epsilon_1} x\type \tau.\, e_3\type \tau \to \sigma\etype\epsilon_1$ and $v_2\type \tau$, we have $e_3[v_2/x]\type \sigma\etype\epsilon_1$.

\item Suppose $e$ has the form $\loss(v)$. As $e\type \sigma$, we have $v = r$, for some $r \in \R$, and
$\sigma = ()$. We have $\myg\vdash_{\epsilon_1} e \xrightarrow{r} ()$, and $()\type ()$.

\item Suppose $e$ has the form $\wph{h}{v_1}{K[\op(v_2)]}$ with $\op \notin \hop(K)$ and $ \op \in h$.

 As $e\type\sigma\etype\epsilon_1$ we have $h$ has effect $\epsilon_1$ and also 
$h \mbox{ handles } \ell$, $ h\type \parr, \sigma_1\etype\epsilon_1\ell \Rightarrow \sigma\etype \epsilon_1$, 
$v_1\type \parr$,
and $K[\op(v_2)] \type \sigma_1\etype \epsilon_1\ell$,
for some $\ell$, $\parr$, and $\sigma_1$. 
As $\op \in h$, from the first of these  we have
 $\op\type \outt \xrightarrow{\ell} \inn$ for some $\outt$ and $\inn$. 
 
 We then have 
 \[\myg\vdash_{\epsilon_1} e \xrightarrow{0} v_o(v_1,v_2,f_l,f_k)\]
 where 
\[f_k = \lambda^{\epsilon_1}(p,y)\type (par,in).\,\glocal{\epsilon_1}{\wph{h}{p}{K[y]}}{\myg}\]
and 
\[f_l = \lambda^{\epsilon_1} (p,y)\type (par,in).\,\then{\epsilon_1}{\wph{h}{p}{K[y]}}{\myg}\]
 
 From $ h\type \parr, \sigma_1\etype\epsilon_1\ell \Rightarrow \sigma\etype \epsilon_1$ we see that 
$v_o \type (\parr, \outt, (\parr, in) \to \real\etype \epsilon,(\parr,  in)\to \sigma \etype \epsilon_1)\to  \sigma\etype \epsilon_1$, where  $\op \mapsto v_o \in h$.
From $K[\op(v_2)] \type \sigma_1\etype \epsilon_1\ell$ and $\op\type \outt \xrightarrow{\ell} \inn$ we see,
via Lemma~\ref{Alem:ctx}, that $v_2\type \outt$ and $y\type \inn \vdash K[y]\type \sigma_1\etype\epsilon_1\ell$.
From the latter and the typing of $h$ we see that 
\[p\type par, y\type in \vdash \wph{h}{p}{K[y]}\type \sigma \etype \epsilon_1\]
From that and $\myg\type \sigma \to \real\etype \epsilon_2$ with $ \epsilon_2 \subseteq \epsilon_1$  we see that 
\[p\type par, y\type in \vdash \glocal{\epsilon_1}{\wph{h}{p}{K[y]}}{\myg}\type \sigma \etype \epsilon_1\]
and
\[p\type par, y\type in \vdash \then{\epsilon_1}{\wph{h}{p}{K[y]}}{\myg}\type \sigma \etype \epsilon_1\]
So $f_k\type (\parr,  in_i)\to \sigma \etype \epsilon_1$ and
$f_l\type (\parr, in_i) \to \real\etype \epsilon_1$.
Then, putting the types of $v_o$, $v_1$, $v_2$, $f_k$, and $f_l$ together, we see that 
$v_o(v_1,v_2,f_l,f_k)\type \sigma\etype \epsilon_1$, as required.

\item Suppose $e$ has the form  $\wph{h}{v_1}{v_2}$. As $e\type\sigma\etype\epsilon_1$ we have
$h \mbox{ handles } \ell$, $ h\type \parr, \sigma_1\etype\epsilon_1\ell \Rightarrow \sigma\etype \epsilon$, 
$v_1\type \parr$,
and $v_2 \type \sigma_1$,
for some $\ell$, $\parr$, and $\sigma_1$. From the second of these we see that $v_r\type \sigma_1 \to \sigma \etype\epsilon_1$, where  $\return \mapsto v_r \in h$. So we have  
$\myg \vdash \wph{h}{v_1}{v_2}\xrightarrow{0} v_r(v_1,v_2)$ and $v_r(v_1,v_2)\type\sigma\etype\epsilon $

\item  Suppose $e$ has the form $\then{}{v}{\lambda^\epsilon x\type \tau.e_1}$. 
 As $e\type \sigma\etype\epsilon_1$, we have $v\type \tau $, $\epsilon \subseteq \epsilon_1$, and $x\type \tau \vdash e_1\type \real \etype \epsilon$.
Then 
$\myg\vdash_{\epsilon_1}\then{}{v}{\lambda^\epsilon x\type \sigma.e_1} \xrightarrow{0} \glocal{\epsilon}{e_1[v/x]}{\lambda^{\epsilon}x : \real.\,0}$.
and $\glocal{\epsilon}{e_1[v/x]}{\lambda^{\epsilon}x : \real.\,0}\type \sigma\etype\epsilon_1$.

\item   Suppose $e$ has the form $\glocal{\epsilon}{v}{\og}$. Then $\myg\vdash_{\epsilon_1}e \xrightarrow{0} v$, and we have $v\type \sigma$ as  $\glocal{\epsilon}{v}{\og}\type\sigma\etype \epsilon_1$.

\item   Suppose $e$ has the form $\reset{v}$. Then 
$\myg\vdash_{\epsilon_1}e \xrightarrow{0} v$, and we have $v\type \sigma$ as  $\reset{v}\type\sigma\etype \epsilon_1$.

\end{enumerate}

We next consider the various possible context cases.

\begin{enumerate}

\item Suppose $e$ has the form $F[e_1]$ where $e_1$ is not terminal. As $e\type \sigma\etype \epsilon_1$, 
by Lemma~\ref{Alem:ctx}, we have $e_1\type \tau\etype\epsilon_1$ for some type $\tau$ 
(regular contexts do not handle operations), and $x\type \tau \vdash F[x]\type \sigma\etype \epsilon_1$.
We therefore have  $\lambda^{\epsilon_1} x\type \tau.\, (\then{\epsilon_1}{F[x]}{\myg})\type \tau \to \real!\epsilon_1$, as
$\myg\type \sigma \to \real\etype\epsilon_2$, with $\epsilon_2 \subseteq \epsilon_1$.

So, by the induction hypothesis, for some $r$ and $e_1'\type \epsilon_1$  we have $\lambda^{\epsilon_1} x\type \tau.\, (\then{\epsilon_1}{F[x]}{\myg}) \vdash_{\epsilon_1} e_1 \xrightarrow{r} e_1'$, and $e_1'\type \tau\etype\epsilon_1$.
So by the rule for regular contexts, we have $\myg \vdash_{\epsilon_1} F[e_1]\xrightarrow{r} F[e_1]$, 
and, by Lemma~\ref{Alem:ctx} we also have $F[e_1']\type \sigma\etype\epsilon_1$.

\item Suppose $e$ has the form $S[e_1]$, with $e_1$ non-terminal where $S = \wph{h}{v}{\square}$.  
As $e\type \sigma\etype \epsilon_1$, we have $h$ has effect $\epsilon_1$ and also 
$h\type \parr,\sigma_1\etype\epsilon_1\ell \Rightarrow \sigma\etype \epsilon_1$,
$v\type\parr$, and 
$e_1\type \sigma_1\etype \epsilon_1\ell$, 
for some $\parr$, $\sigma_1$, and $\ell$.
So for $\return \mapsto e_r \in h$, we have 
$e_r\type (\parr,\sigma_1) \to \sigma\etype \epsilon_1$.
So we have 
$\lambda^{\epsilon_1} x\type \sigma.\, (\then{\epsilon}{v_r(v,x)}{\myg}\type \sigma \to \real\etype \epsilon_1$.

Then, by the induction hypothesis we have 
$\lambda^{\epsilon_1} x\type \sigma.\, (\then{\epsilon_1}{v_r(v,x)}{\myg}) \vdash_{\epsilon_1\ell} e_1 \evalto{r} e_1'$ for some $r$ and 
$e_1'\type \sigma_1\etype \epsilon_1\ell$.
So by the context rule for $S$ we have $\myg \vdash_\epsilon \wph{h}{v}{e} \evalto{r} \wph{h}{v}{e'}$, and we see that $\wph{h}{v}{e'}\type \sigma\etype\epsilon_1$, as required.

\item Suppose $e$ has the form $S[e_1]$, with $e_1$ non-terminal where $S = \then{\epsilon}{\square}{(\lambda^{\epsilon} x\type \tau.e_2)}$.
Since $e\type \sigma\etype \epsilon_1$, 
we have $\sigma = \real$, 
$e_1\type \tau\etype\epsilon_1$,
$\epsilon \subseteq \epsilon_1$, and
$\lambda^{\epsilon} x\type \tau.e_2\type \tau \to \real\etype\epsilon$.
So, by the induction hypothesis, we have $\lambda^{\epsilon} x\type \tau.e_2 \vdash_{\epsilon_1} e_1 \xrightarrow{r} e_1'$ for some $r$ and $e_1'\type\tau\etype\epsilon_1$.
Then, using the context rule for $S$, we see that
$\myg \vdash_{\epsilon_1}  (\then{\epsilon}{e_1}{\lambda^{\epsilon} x\type \tau.e_2}) \evalto{\; 0 \;} r +  
            (\then{\epsilon_1}{e_1'}{\lambda^{\epsilon} x\type \tau.e_2})$. 
Finally, as $e_1'\type\tau\etype\epsilon_1$ and $\lambda^{\epsilon} x\type \tau.e_2\type \tau \to \real\etype\epsilon$
we see that $(r + \then{\epsilon_1}{e_1'}{\lambda^{\epsilon} x\type \tau.e_2})\type\real\etype\epsilon_1$ as required.

\item  Suppose $e$ has the form $S[e_1]$, with $e_1$ non-terminal where $S = \glocal{\epsilon}{\square}{\og}$.

As $e\type \sigma\etype \epsilon_1$, we have $e_1\type \sigma \etype\epsilon$, 
with $\epsilon \subseteq \epsilon_1$, and 
$\og\type \sigma \to \real\etype \overline{\epsilon}$, with $\overline{\epsilon}\subseteq \epsilon$.

By the induction hypothesis we then have $\og \vdash_{\epsilon} e_1 \evalto{r} e_1'$ 
for some $r$ and $e'$, with $e_1'\type \sigma\etype\epsilon$.

So by the context rule for this $S$ we have $\myg \vdash_{\epsilon_1} \glocal{\epsilon}{e_1}{\og} \evalto{r} \glocal{\epsilon}{e_1'}{\og}$.

Finally, as $e_1'\type \sigma \etype\epsilon$, we have $\glocal{\epsilon}{e_1'}{\og}\type \sigma \etype\epsilon_1$, as required.

\item  Suppose $e$ has the form $S[e_1]$, with $e_1$ non-terminal where $S = \reset{\square}$.

As $e\type \sigma\etype \epsilon_1$, we have $e_1\type \sigma \etype\epsilon_1$.
By the induction hypothesis we then have $\myg \vdash_{\epsilon_1} e_1 \evalto{r} e_1'$ 
for some $r$ and $e_1'$, with $e_1'\type \sigma\etype\epsilon$.

So by the context rule for this $S$ we have $\myg \vdash_{\epsilon_1} \reset{e_1} \evalto{0} \reset{e'_1}$.

Finally, as $e_1'\type \sigma \etype\epsilon_1$, we have $\reset{e_1'}\type \sigma \etype\epsilon_1$, as required.

\end{enumerate}

\end{proof}

\begin{corollary} \label{Acor:evaltd} For $e\type \sigma \etype  \epsilon$ and $\myg \type \sigma\to \bool\etype \epsilon'$ with $\epsilon'\subseteq \epsilon$ there is at most one $r \in \R$ and terminal expression $w$ such that $\myg \vdash e \xRightarrow{r} w$ and then $w\type \sigma\etype \epsilon$.
\end{corollary}
\begin{proof} This follows immediately from Theorem~\ref{Athm:dps}.
\end{proof}

\subsection{Termination}
\newcommand{\ov}{\overline}
\paragraph{\textbf{Well-foundness of effects}}

We write $e(\epsilon)$ and $e(\sigma)$ for the effect labels appearing in $\epsilon$ or $\sigma$. So, for example 
$e(\sigma \to \tau\etype\epsilon) =  e(\sigma) \cup  e(\tau) \cup (\epsilon)$. Our  well-foundedness assumption is that there is an ordering $\ell_1,\dots, \ell_n$ of the effect labels such that:
\[\op\type \outt \xrightarrow{\ell_j} \inn \;\wedge\; \ell_i \in e(\outt) \cup e(\inn) \implies i < j\]
We make this assumption for this subsection, and the next section (\Cref{sec:denotational}).

With it, we define the effect levels of multisets of effect labels and types by setting:
$l(\epsilon) = \max_i \{i | \ell_i \in e(\epsilon)\}$ and  $l(\sigma) = \max_i \{i | \ell_i \in e(\sigma)\}$.
We will also make use of the size $|\sigma|$ of an type defined in a standard way (e.g.
$|\sigma \to \tau\etype\epsilon| =  1 + |\sigma| +  |\tau| + |\epsilon|$).

\paragraph{\textbf{Computability}}
Our proof uses suitable recursively-defined notions of computability, following Tait~\cite{Tait67}. 
We define the following main notions:
\begin{itemize}
\item[-] \textit{computability} of  values $v\type \sigma$,
\item[-]  \textit{loss computability} of  loss continuations  $\myg\type \sigma \to \real\etype \epsilon$, and   
\item[-]  \textit{computability} of  expressions $e\type\sigma\etype\epsilon$.
\end{itemize}
 
The definitions are proper (i.e. the recursions terminate) as can be be seen by suitable measures $m$ defined on closed values $v\type \sigma$, closed expressions $e\type \real\etype \epsilon$, closed loss continuations $\myg \type \sigma \to \real\etype \epsilon$, and closed expressions $e\type \sigma\etype \epsilon$; these are pairs of natural numbers, with the lexicographic ordering and are given by:
\[m(v) = (l(\sigma), |\sigma|)
\quad
m(e) = (l(\epsilon), 1)
\quad
m(\myg) = (l( \sigma) \max l(\epsilon), | \sigma| )
\quad
m(e) = (l( \sigma) \max l(\epsilon), |\sigma|)\]
 
We define these notions by the following clauses. 
They employ two auxiliary notions.
One is an inductively defined  notion of \textit{$G$-computability} of closed expressions, where $G$ is a set of closed loss continuations; the other is  a notion of  \emph{R-computability} of closed real-valued expressions.

\begin{enumerate}
\item \begin{enumerate}
         \item Every constant $c\type b$ of ground type is computable.as is $\zero$ and every $\nil_\sigma$.
         \item A  closed value $(v_1,\dots,v_n)\type (\sigma_1,\dots,\sigma_n)$ is computable if every 
                   $v_i\type \sigma_i$ is computable.
                   \item A closed value of one of the forms $\inl{\sigma}{\tau}(v)$,  $\inr{\sigma}{\tau}(v)$, or $\mysucc{v}$ is computable if $v$ is. 
            \item A closed value of  the form  $\cons{v_1}{v_2}$ is computable if $v_1$ and $v_2$ are.  
         
              \item A closed value $\lambda^{\epsilon} x\type \sigma.\, e \type \sigma \to 
                    \tau\etype\epsilon$ is computable if, for every computable value $ v\type \sigma$, the 
                    expression $e[v/x]\type \tau\etype\epsilon$ is computable.
                \end{enumerate}
\item The property  of %
{$G$-computability} of  closed expressions $ e\type \sigma\etype \epsilon$, for a set $G$ of closed loss continuations of type $ \myg\type \sigma \to \real\etype \epsilon'$, where $\epsilon' \subseteq \epsilon$,   is the least such property  $P_{\sigma,\epsilon}$ of these expressions such that 
one of the following three mutually exclusive possibilities holds:
\begin{enumerate}
\item  $e$ is a  computable value.
\item  $e$ is an  stuck expression  $K[\op(v)]$, with $\op\type \outt \xrightarrow{\ell} \inn$, where $v:\outt$ is  a computable  value, and  for every  computable closed value $v_1\type \inn$, 
$P_{\sigma,\epsilon}(K[v_1])$ holds.
\item $e$ is not stuck and for every $\myg \in G$, if $\myg \vdash_\epsilon e \xrightarrow{r} e'$ then  $P_{\sigma,\epsilon}(e')$ holds.
\end{enumerate} 
\item \begin{enumerate}
                 \item A closed expression $e:\real\etype \epsilon$ is {R-computable}
iff it is $\{0_{\real,\epsilon}\}$-computable.

\item A closed loss continuation $\lambda^\epsilon x\type \sigma.\, e\type \sigma \to \real\etype \epsilon$ is {loss-computable} if $e[v/x]$ is R-computable for every closed computable value $v\type \sigma$.
\end{enumerate}
\item A closed expression $e\type \sigma\etype\epsilon$ is {computable} 
iff it is $\Lsc$-computable, where $\Lsc$ is the set of loss-computable 
loss continuations $\myg\type \sigma \to \real\etype \epsilon'$, 
for some $\epsilon'\subseteq \epsilon$. 

\end{enumerate}

 Note that a closed value is computable as a value iff it is computable as an expression. 
Also a terminal expression $  K[\op(v)] \type \sigma\etype \epsilon$ where $\op\type \outt \xrightarrow{\ell} \inn$ is computable iff $  v \type \outt$ is and 
 $K[v_1]$ is for every computable $  v_1\type \inn$.
We may write $L$ instead of $\Lsc$ when $\sigma$ and $\epsilon$ can be understood from the context.

As $G$-computability is defined by a least-fixed point, for expressions $e\type \sigma\etype \epsilon$ and loss continuations $\myg\type \sigma \to \real\etype \epsilon$ with $\epsilon' \subseteq \epsilon$ in $G$, 
the following principle of \emph{$G$-induction} holds:
\begin{itemize}
\item[] Let $P_\epsilon(e)$ be a predicate of closed expressions of type $e\type \sigma \etype \epsilon$ such that the following three clauses hold:
\begin{enumerate}
\item   For every computable value $v\type \sigma$, $P_\epsilon(v)$ holds.
\item For every stuck  $K[\op(v)]$, with $\op\type \outt \xrightarrow{\ell} \inn$ and $v\type \outt$ a computable value,
$P_\epsilon(K[\op(v)])$ holds provided that $P_\epsilon(K[v_1])$ holds for every  computable value $v_1\type \inn$.  
\item For every non-stuck $e\type\sigma\etype\epsilon$, $P_\epsilon(e)$ holds provided that 
$\myg \vdash_\epsilon e \xrightarrow{r} e'$  implies $P_\epsilon(e')$ holds, for every $\myg \in G$.
\end{enumerate}
Then $P_\epsilon(e)$ holds for every    $G$-computable expression $e\type \real\etype \epsilon$.
\end{itemize}
As an example, one can use this induction principle to show that if $G'\subseteq G$ and $e$ is $G$-computable, then it is also $G'$-computable. So we see that every computable loss continuation is loss-computable. We term $\{(\lambda x\type \real.\, 0)\}$-induction \emph{R-induction}.

We extend computability to open expressions in the standard way: an expression 
$x_1\type \sigma_1,\dots,x_n\type \sigma_n \vdash e\type \sigma\etype \epsilon$ is said to be computable if the expression $ e[v_1/x_1,\dots,v_n/x_n]\type \sigma\etype \epsilon$ is computable for all closed computable values $ v_1\type \sigma_1,\dots,  v_n\type \sigma_n$,. When the $v_i$ are known from the context, we generally write $\ov{e}$ for $e[v_1/x_1,\dots,v_n/x_n]$.

We say that a handler $\myG \vdash h\type \parr,\sigma\etype \epsilon\ell \Rightarrow \sigma'\etype \epsilon$ is computable if,  for every handler operation  definition clause  
$\op \mapsto e_{\op}$, $\myG\vdash e_{\op}\type ((\parr, \outt_i, (\parr, in_i) \to \real\etype \epsilon,(\parr,  in_i)\to \sigma'\etype \epsilon)\to  \sigma'\etype \epsilon)\etype \epsilon$ is   and so is $\myG\vdash e_r\type (\parr, \sigma) \to \sigma'\etype \epsilon$ (where  the handler's return clause is $\return \mapsto e_r$).

\begin{lemma} \label{Alem:glocal}
 If $e\type \sigma\etype \epsilon_1$ is $\{\myg\}$-computable for a loss-computable $\myg\type \sigma \to \real\etype \epsilon_2$ (with 
 $\epsilon_2 \subseteq \epsilon_1$), then $\glocal{\epsilon_1}{e}{\myg}\type \sigma\etype\epsilon$ (with $\epsilon_1 \subseteq \epsilon$) is computable.

\end{lemma}
\begin{proof} 
We proceed by $\{\myg\}$-induction. There are three cases:
\begin{enumerate}
\item If $e$ is a computable value $v$, then 
  for any $\myg_1 \in \Lsc$   we have 
$\myg_1 \vdash_\epsilon \glocal{\epsilon_1}{v}{\myg} 
\xrightarrow{0}v$.

    \item If $e$ is a stuck expression of the form $K[\op(v)]$ 
    with $\op : \outt \to \inn$, $v$ computable and $\glocal{\epsilon_1}{K[w]}{\myg}$ computable 
    for every computable $w\type \inn$, then we simply note that
    $\glocal{\epsilon_1}{K[op(v)]}{\myg}$ is stuck.
    \item Suppose that $e$ is not stuck. Then $\glocal{\epsilon_1}{e}{\myg}$ is not stuck.
    Further, if $\myg_1 \vdash_\epsilon \glocal{\epsilon_1}{e}{\myg} \xrightarrow{r} e''$ for $\myg_1 \in \Lsc$, then, as $e$ is not stuck, for some $e'$,
    $\myg \vdash_{\epsilon_1} e \xrightarrow{r} e'$
    and $e'' = \glocal{\epsilon_1}{e'}{\myg}$ and, by the induction hypothesis, $e''$ is computable.

\end{enumerate}
\end{proof}

\begin{lemma} \label{Alem:technical}
 \myskip
 \begin{enumerate}
    \item Suppose $e\type \real\etype\epsilon_1$ is $R$-computable. Then so is any
    $r + \glocal{\epsilon_1}{e}{0_{\real,\epsilon_1}}\type \real\etype\epsilon$ (where $\epsilon_1 \subseteq \epsilon$).
    \item 
 Suppose $\myg\type \real \to \real\etype\epsilon_1$  is loss-computable. Then, for any $r \in \R$, so is
$\lambda^\epsilon x\type\real.\, \then{\epsilon}{(r + x)}{\myg}$  (where $\epsilon_1 \subseteq \epsilon$).
\end{enumerate}
\end{lemma}
\begin{proof}
\myskip
\begin{enumerate}
\item We proceed by R-induction. There are three clauses:
\begin{enumerate}
\item  For a value $v\type\real$ we have 
$0_{\real,\epsilon} \vdash_\epsilon  r + \glocal{\epsilon_1}{v}{0_{\real,\epsilon_1}} 
\xrightarrow{0} r + v \xrightarrow{0} s$, 
where $s = r +v$.
\item For a stuck expression  $K[\op(v)]$, with $\op\type \outt \to \inn$ and computable value $v\type \outt$, suppose that $r + \glocal{\epsilon_1}{K[v_1]}{0_{\real,\epsilon_1}}$  is $0_{\real,\epsilon}$-computable for every  computable value $v_1\type \inn$.  This case is immediate, as then $r + \glocal{\epsilon_1}{K[v]}{0_{\real,\epsilon_1}}$ is a stuck expression, and every  $r + \glocal{\epsilon_1}{K[v_1]}{0_{\real,\epsilon_1}}$ is $0_{\real,\epsilon}$-computable.
\item Suppose that $e$ is not stuck. Then neither is
$r + \glocal{\epsilon_1}{e}{0_{\real,\epsilon_1}}$.
We may suppose that $0_{\real,\epsilon_1} \vdash_\epsilon e \xrightarrow{s} e'$, for some $e'$ and $s$, with $r + \glocal{\epsilon_1}{e'}{0_{\real,\epsilon_1}}$   $0_{\real,\epsilon}$-computable. Then we have
  \[0_{\real,\epsilon} \vdash_\epsilon  \glocal{\epsilon_1}{e}{0_{\real,\epsilon_1}} \xrightarrow{s}
  \glocal{\epsilon_1}{e'}{0_{\real,\epsilon_1}}\]
  and so
  \[0_{\real,\epsilon} \vdash_\epsilon r + \glocal{\epsilon_1}{e}{0_{\real,\epsilon_1}} \xrightarrow{s}
  r + \glocal{\epsilon_1}{e'}{0_{\real,\epsilon_1}}\]
  and we conclude.
  \end{enumerate}
  
\item 
Suppose $\myg$ has the form $\lambda^{\epsilon_1} x\type \real.\, e$. Then, for any $v \in \R$ we have:
\[0\vdash_\epsilon \then{\epsilon}{(r + v)}{\myg} \;\xrightarrow{0}\;
0 + (\then{\epsilon}{s}{\myg})\; \xrightarrow{0}\; 0 + \glocal{\epsilon_1}{e[s/x]}{0_{\real,\epsilon_1}}\]
where $s = r + v$, and we apply part 1.
\end{enumerate}
\end{proof}

\begin{lemma} \label{Alem:ginc}
\myskip
\begin{enumerate}
\item Let $e\type \real\etype \epsilon$ be computable,
for $\myg_1:\real \to \real\etype\epsilon_1$ with $\epsilon_1 \subseteq \epsilon$. Then $r+ e$ is too, for any $r \in \R$.
\item
Let $\myg\type \sigma \to \real \etype \epsilon'$ be  loss-computable   and let $e\type \sigma \etype \epsilon$ be a $\{\myg\}$-computable expression (with $\epsilon'\subseteq \epsilon$). Then 
$\then{\epsilon}{e}{\myg}$  is computable
 for $\myg_1:\real \to \real\etype\epsilon_1$ with $\epsilon_1 \subseteq \epsilon$.
\item 
Let $\myg\type \sigma \to \real \etype \epsilon'$ be loss-computable
 and let 
$x\type \tau \vdash e\type \sigma \etype \epsilon$ be a computable expression (with $\epsilon'\subseteq \epsilon$). Then $\lambda^\epsilon x\type \tau.\, \then{\epsilon}{e}{\myg}\type \tau \to \real\etype \epsilon$ is computable.
\end{enumerate}
\end{lemma}
\begin{proof}
\myskip
\begin{enumerate}
\item  We proceed by $L_{\real,\epsilon}$-induction on $e$ to show that $r + e$ is computable. We consider the three clauses in turn:
\begin{enumerate}
\item  For a value $v\type\real$, we have $\myg \vdash_\epsilon r + v  \xrightarrow{0} s$ where $s = r +v$, for $\myg\in \L_{\real,\epsilon}$.
\item For a stuck expression  $K[\op(v)]$, with $\op\type \outt \to \inn$ and computable value $v\type \outt$, suppose that $r + K[v_1]$ is R-computable for every  computable value $v_1\type \inn$.  This case is immediate, as then $r + K[\op(v)]$ is a stuck expression, and every  $r + K[v_1]$ is R-computable.
\item Suppose that for all $\myg' \in L_{\real,\epsilon}$, if $\myg' \vdash_\epsilon e \xrightarrow{s'} e'$ then $r + e'$ is computable. 
Now, choose a $\myg\in \L_{\real,\epsilon}$ and suppose $\myg\vdash_\epsilon r+ e \xrightarrow{s} e''$ to show that $e''$ is computable. 
Then $\lambda^\epsilon x\type\real.\, \then{}{(r +x)}{\myg}\vdash _\epsilon e \xrightarrow{s} e'$, for some $e'$, and $e'' = r+e'$.
By Lemma~\ref{Alem:technical}, as $\myg \in \L_{\real,\epsilon}$, so is $\lambda^\epsilon x\type\real.\, \then{}{(r +x)}{\myg}$. Then, by our initial assumption we have $e'' = r+e'$ computable, as required.

\end{enumerate}

\item  We proceed by $\{\myg\}$-induction on $e$, to show that $\then{\epsilon}{e}{\myg}$ is computable. Suppose $\myg$ 
has the form $\lambda^{\epsilon} x\type \sigma.\, e_1$. We consider the three clauses in turn:
\begin{enumerate}
\item  For a value $v\type\real$, we have $\myg_1 \vdash_\epsilon \then{\epsilon}{v}{\myg}  \xrightarrow{0} \glocal{\epsilon'}{e_1[v/x]}{\lambda^{\epsilon'} x:\real.\,0}$ for
$\myg_1 \in \Lsc$. As $\myg$ is loss-computable, this shows that $\then{\epsilon}{v}{\myg}$ is computable, using Lemma~\ref{Alem:glocal}.
\item For a stuck expression  $K[\op(v)]$, with $\op\type \outt \to \inn$ and computable value $v\type \outt$, suppose that $\then{\epsilon}{K[v_1]}{\myg} $ is computable for every  computable value $v_1\type \inn$.  This case is immediate, as then $\then{\epsilon}{K[\op(v)]}{\myg} $ is a stuck expression, and every  $\then{\epsilon}{K[v_1]}{\myg} $ is computable.
\item Suppose that $\myg \vdash_\epsilon e \xrightarrow{r} e'$ and $\then{\epsilon}{e'}{\myg}$ is computable. Then
  $\myg_1 \vdash_\epsilon \then{\epsilon}{e}{\myg} \xrightarrow{0} r + (\then{\epsilon}{e'}{\myg})$, for any $\myg_1 \in \Lsc$. So, using part 1 of this lemma, we see that $\then{\epsilon}{e}{\myg}$ is computable.
\end{enumerate}

\item Part 3 follows immediately from part 2.
\end{enumerate}
\end{proof}

\begin{lemma} \label{{Alem:easyframe}} 
Let $F$ be a regular frame such  that $x\type \sigma \vdash F[x]\type \tau \etype \epsilon$ is computable.  Then 
$F[e]\type \tau\etype\epsilon$ is computable for any computable expression $e\type \sigma \etype \epsilon$.
\end{lemma} 
\begin{proof} 
Fix a loss-computable $\myg: \tau \to \real\etype\epsilon'$ 
with $\epsilon'\subseteq\epsilon$, to show $F[e]$ $\{\myg\}$-computable. 
By Lemma~\ref{Alem:ginc}, the loss continuation 
$\myg' \eqdef \lambda^\epsilon x\type \sigma. (\then{\epsilon}{F[x]}{\myg}) \type \tau \to \real\etype \epsilon$ is loss-computable
and so $e$ is $\{\lambda^\epsilon x\type \sigma. (\then{\epsilon}{F[x]}{\myg})\}$-computable.

We show by $\{\myg'\}$-induction on $\{\myg'\}$-computable expressions $e$  that $F[e]\type \tau\etype \epsilon$ is $\{\myg\}$-computable. 

\begin{enumerate}
\item  $F[v]$ is computable for every computable value $v\type \sigma$ as $F$ is assumed computable.
\item 
Let $K[\op(v)]$ be a stuck expression, with $\op\type \outt \to \inn$, computable value $v\type \outt$, and $F[K[v_1]]$ $\{\myg'\}$-computable, for every computable value $v_1\type \inn$. Then $F[K[\op(v)]]$ is a stuck expression as $F$ is not $\wph{h}{v}{\square}$, and it is then seen to be $\{\myg'\}$-computable as every $F[K[v_1]]$ is.
\item We may suppose that 
$\lambda^\epsilon x\type \sigma.\, (\then{\epsilon}{F[x]}{\myg}) \vdash_\epsilon e \xrightarrow{r} e'$ and $F[e']$ is $\{\myg\}$-computable. Then $\myg \vdash_\epsilon F[e] \xrightarrow{r} F[e']$ and so $F[e]$ is too.
\end{enumerate}

\end{proof}

\begin{lemma} \label{Alem:handframe}
 Suppose the $\ell$-handler  $h\type \parr, \sigma \etype \epsilon\ell \Rightarrow \sigma'\etype \epsilon$ and 
the expression $e\type \sigma \etype \epsilon\ell $ are both computable. Then, for any  computable  value, $v\type \parr$ so is the expression 
$\wph{h}{v}{e}\type \sigma'\etype \epsilon$.
\end{lemma}
\begin{proof}
First, for any computable value $v\type\parr$, set $F_v \eqdef \wph{h}{v}{\square}$, and for any $\myg \in L_{\sigma',\epsilon}$ set 
$\myg_v \eqdef  \lambda^\epsilon x\type \sigma.\,\then{\epsilon}{e_r(v_p,x)}{\myg}$,
where the return clause in $h$ is $\return \mapsto e_r$.
By Lemmas~\ref{Alem:ginc} and~\ref{{Alem:easyframe}}, each $\myg_v$ is 
computable, since $e_r$ is computable as $h$ is.
We use $\Lsc$-induction on $e$ to show that for any computable $v\type\parr$, $F_v[e]$  is computable. 
The three cases are:
\begin{enumerate}
\item Suppose $e$ is a computable value $w\type\sigma$. Choose  computable values $v\type\parr$  to show $F_v[w]$ is computable. For any
$\myg \in L_{\sigma',\epsilon}$ we have 
\[\myg \vdash_\epsilon F_v[w] \xrightarrow{0} e_r(v_p,v) \qquad (\ast)\]
 Applying Lemma~\ref{{Alem:easyframe}}, we see that $e_r(v_p,v)$ is computable, as $e_r$ is computable since $h$ is. So, using $(\ast)$, we see that $F_v[w]$ is $\{\myg\}$-computable.
\item 
Let $K[\op(w)]$ be stuck, with $\op\type \outt \to \inn$, computable value $w\type \outt$, such that for every computable value $w_1\type \inn$,$F_{v_1}[K[w_1]]$ is computable  for every computable value $v_1\type\parr$. 
We have to show that $F_v[K[\op(w)]]$ is computable for all computable values $v\type\parr$.

In case $\op$ is not an $h$-operation, we see that each  $F_{v}[K[\op(w)]]$ is a stuck expression and $F_{v}[K[w_1]]$ is computable for all values  $w_1\type \inn$. So, in this case, 
each $F_v[K[\op(w)]]$ is $\{\myg\}$-computable, as required.

In case $\op$ is an  $h$-operation, for any $\myg \in L_{\sigma',\epsilon}$ we have
\[\myg \vdash_\epsilon F_v[K[\op(w)]] \xrightarrow{0} e_\op(v,w,f_l,f_k)\qquad (\ast\ast)\]
where 
\[\op \mapsto e_\op \in h\]
and
\[f_k = \lambda^\epsilon(p,y)\type (par,in).\,\glocal{\epsilon}{\wph{h}{p}{K[y]}}{\myg}\]
and
 \[f_l = \lambda^\epsilon (p,y)\type (par,in).\,\then{\epsilon}{(\wph{h}{p}{K[y]})}{\myg}\]

and it suffices to prove that $e_\op(v,w,f_l,f_k)$ is computable.
As $e_\op$ is computable by the assumption that $h$ is and as $v$ and $w$ are computable by assumption, we need only prove that $f_k$ and $f_l$ are.

To see that $f_k$ is computable we note that,  for all $v_1\type \parr$ and $w_1 \type \inn$, we are given that
$\wph{h}{v_1}{w_1}$ is computable, and then, using Lemma~\ref{Alem:glocal}, we see that  $\glocal{\epsilon}{\wph{h}{v_1}{K[w_1]}}{\myg}$ is computable .

The proof for $f_l$ is similar, instead using Lemma~\ref{Alem:ginc}.
 \item Suppose $e$ is not stuck. Here, for any $\myg_1 \in L_{\sigma,\epsilon\ell}$, if 
 $\myg_1 \vdash_{\epsilon\ell} e \xrightarrow{r} e'$ then $F_v[e']$
is computable for any computable value $v\type\parr$. 
As $e$ is not stuck then neither is any $F_v[e]$ and for any
$\myg \in L_{\sigma',\epsilon}$, if $\myg\vdash_\epsilon F_v[e] \xrightarrow{r} e''$ then 
 $e'' = F_v[e']$, for some $r$ and $e'$ such that
 $\myg_v \vdash_{\epsilon\ell} e \xrightarrow{r} e'$ (recall that $\myg_v$ is computable and so in $L_{\sigma,\epsilon\ell}$) and so $e'' = F_v[e']$ is computable, as required.

\end{enumerate}

\end{proof}

\begin{lemma}[Fundamental Lemma] \label{Alem:comp}Every expression $\myG \vdash e\type \sigma\etype \epsilon$ is computable.
\end{lemma} 
\begin{proof}

  The proof is by induction on the size of $e$ (defined in a standard way). Consider a context $\myG = x_1\type \sigma_1,\dots,x_n\type \sigma_n$, and choose computable  values $v_1\type\sigma_1,\dots,v_n\type \sigma_n$. We have to show that 
  $\ov{e} = e[v_1/x_1,\dots,v_n/x_n]$ is computable. 

 The proof splits into cases according to the form of $e$, following Lemma~\ref{eal}.
\begin{enumerate}
\item %
 \begin{enumerate}
      \item If $e = x_i$ then $\ov{e} = v_i$ is a computable  value by assumption. So $e$ is a computable expression. 
\item If $e$ is a basic constant, $\zero$, or some $\nil_\sigma$ then it is a computable value by definition.
 \item If $e = (v'_1,\dots, v'_m)$ then $\ov{e} = (\ov{v'_1},\dots, \ov{v'_n})$ is a computable value as, by the induction hypothesis all the $\ov{v'_j}$ are computable expressions.
 \item If $e$ has one of the forms $\inl{\sigma}{\tau}(v)$, $\inr{\sigma}{\tau}(v)$, 
   $\mysucc{v}$, or $\cons{v_1}{v_2}$ the proof is similar to the previous case.
   \item If $e = \lambda^{\epsilon} x\type \sigma_1.\, e_1$ then we  show that $\ov{e} =  \lambda^{\epsilon} x\type \sigma_1.\, \ov{e_1}$ is a computable value, that is that for every computable value $ v\type \sigma_1$, $\ov{e_1}[v/x]$ is a computable expression, and that is immediate from the induction hypothesis.
 \end{enumerate}
\item %
The next case is when $e$ is stuck, with the form $K[\op(v)]$ with $\op\type \outt \xrightarrow{\ell} \inn$ and $\myG \vdash v\type \outt$. It suffices to show that $\ov{K}[w]$ is computable for every computable $w\type \inn$. However, as $K[x]$ (choosing $x \notin \Dom(\myG)$)  is smaller than $e$, we have $K[x]$ computable, and so $\ov{K}[w] = K[x][v_1/x_1,\dots,v_n/x_n,w/x]$ is computable.
\item The next case is when $e$ is a redex. As before we split into subcases according to the form of the expression.
  \begin{enumerate}
  \item %
 
  \item  %

  \item  %
  Suppose that $e$ has the form  $v.i$. Then, using the typing information, the fact that $v$ is a value,  and the induction hypothesis we see that  $\ov{e}$ has the form $(v_1,\dots,v_m).i$ with the $v_j$ computable. So as   $\myg \vdash_\epsilon \ov{e} \xrightarrow{0} v_i$ (for every loss-computable loss continuation $\myg\type b \to \real\etype \epsilon'$ with  $\epsilon' \subseteq  \epsilon$)  and as $v_i$ is computable, we see that $\ov{e}$ is computable, as required.

  \item %
  $\inl{\sigma}{\tau}(v')$ or $\inr{\sigma}{\tau}(v')$ with $v'$ computable. In the first case we see that $\ov{e_1}[\ov{v'}/x]$ is computable (as $\ov{e_1}$ is computable by the induction hypothesis) and that $\myg\vdash_\epsilon \ov{e} \xrightarrow{0}  \ov{e_1}[\ov{v_1}/x]$ (for the relevant $\myg$)  and so that  $\ov{e}$ is computable, as required. The second case is similar to the first case.
   
    \item %
    $(\lambda^\epsilon x\type \sigma_1.\, e_1)\ov{v_2}$ with $e_1$ and $\ov{v_2}$ computable, and we  have   $\myg\vdash_\epsilon \ov{e} \xrightarrow{0}  \ov{e_1}[\ov{v_2}/x]$ for the relevant $\myg$.

  \item %
 By the induction hypothesis, $v_1$, $v_2$, and $v_3$ are computable. 
  We proceed by (natural numbers) induction on $\ov{v_1}$ to show that 
  $\ov{e} = \iter{\ov{v_1}}{\ov{v_2}}{\ov{v_3}}$ is computable.  
  
  If $\ov{v_1}$ is $\zero$ then, for any relevant $\myg$ we have  
  $\myg \vdash_\epsilon \iter{\ov{v_1}}{\ov{v_2}}{\ov{v_3}} \xrightarrow{0} \ov{v_2}$ and we are done as $\ov{v_2}$ is computable.
    If $\ov{v_1}$ is $\mysucc{w}$ then, for any relevant $\myg$ we have  
    $\myg \vdash_\epsilon   \iter{\mysucc{w}}{\ov{v_2}}{\ov{v_3}} \xrightarrow{0} 
v_3 \iter{w}{\ov{v_2}}{\ov{v_3}}$. As $v_3$ is computable, the frame $\ov{v_3}\square$ is computable (see previous case), further, by the induction hypothesis, 
$\iter{w}{\ov{v_2}}{\ov{v_3}}$ is computable. So,  by Lemma~\ref{{Alem:easyframe}}, $\ov{v_3} \iter{w}{\ov{v_2}}{\ov{v_3}}$ is computable, and so $\iter{\mysucc{w}}{\ov{v_2}}{\ov{v_3}}$ is too, as required.

  \item %
  By the induction hypothesis, $v_1$, $v_2$, and $v_3$ are computable. 
  We proceed by (structural) induction on $\ov{v_1}$ to show that 
  $\ov{e} = \fold{\ov{v_1}}{\ov{v_2}}{\ov{v_3}}$ is computable.

  If $\ov{v_1}$ is $\nil$ then, for any relevant $\myg$ we have  
  $\myg \vdash_\epsilon \fold{\ov{v_1}}{\ov{v_2}}{\ov{v_3}} \xrightarrow{0} \ov{v_2}$ and we are done as $\ov{v_2}$ is computable.
  
    If $\ov{v_1}$ is $\cons{w_1}{w_2}$ then we have  
    $\myg \vdash_\epsilon   \fold{\cons{w_1}{w_2}}{\ov{v_2}}{\ov{v_3}} \xrightarrow{0} 
\ov{v_3} (w_1,\fold{w_2}{\ov{v_2}}{\ov{v_3}})$,  for any relevant $\myg$. 
The frame $(w_1,\square)$ is computable and so $(w_1,\fold{w_2}{\ov{v_2}}{\ov{v_3}})$ is computable using Lemma~\ref{{Alem:easyframe}}, as, by the induction hypothesis, $\fold{w_2}{\ov{v_2}}{\ov{v_3}}$ is computable.
Next, as $v_3$ is computable, the frame $\ov{v_3}\square$ is computable, 
and so, again using   Lemma~\ref{{Alem:easyframe}}, we see that $\ov{v_3} (w_1,\fold{w_2}{\ov{v_2}}{\ov{v_3}})$
is computable, as  $(w_1,\fold{w_2}{\ov{v_2}}{\ov{v_3}})$ is computable.
  So, as  \[\myg \vdash_\epsilon   \fold{\cons{w_1}{w_2}}{\ov{v_2}}{\ov{v_3}} \xrightarrow{0} 
\ov{v_3} (w_1,\fold{w_2}{\ov{v_2}}{\ov{v_3}})\] $ \fold{\cons{w_1}{w_2}}{\ov{v_2}}{\ov{v_3}} $ is computable, as required.

  \item %
Using the induction hypothesis, we see that  $h$, and so $\ov{h}$ is computable as are $v_1$ and $v_2$, or $v'_1$, $v$, and $\ov{K}[\op(v)]$. We may then apply Lemma~\ref{Alem:handframe} to show that $\ov{e}$ is computable.

  \item 

  Suppose that $e$ has the form  $ \loss(v)$. Here we note that 
  $\myg \vdash \loss(\ov{v}) \xrightarrow{\ov{v}} ()$ for the relevant $\myg$.

  \item %
  Suppose that $e$ has the form  $\then{}{v}{\lambda^{\epsilon_1} x\type \sigma.e_1}$.
 This is immediate from Lemma~\ref{Alem:ginc}, using the induction hypothesis. 

  \item %
  
  Suppose that $e$ has the form $\glocal{\epsilon}{v}{\myg}$, for a value $v$. This case is immediate.

 \item %
  
  Suppose that $e$ has the form $\reset{v}$, for a value $v$. This case is immediate.
   
  \end{enumerate}
\item %
 Suppose that $\myG \vdash e\type \sigma \etype \epsilon$ has the form $F[e_1]$, with $e_1$ not a value or stuck. In this case, by the induction hypothesis, both $F[x]$ (with $x \notin \Dom(\myG)$) and $e_1$ are computable. So both
$\ov{F}[x]$ and $\ov{e_1}$ are computable. We may then apply Lemma~\ref{{Alem:easyframe}}.

\item The last case is where $\myG \vdash e\type \sigma \etype \epsilon$ has the form $S[e_1]$, with $e_1$ not a value or stuck.
There are four possibilities:
\begin{enumerate}
\item %
$S$ is $ \wph{h}{v}{\square}$. In this case, $h$ and $e_1$ are computable, by the induction hypothesis, and we may apply Lemma~\ref{Alem:handframe}.

\item %
In this case, $e_1$ are $\lambda^{\epsilon} x\type \sigma. \, e_2$ are computable, by the induction hypothesis, and we may apply  Lemma~\ref{Alem:ginc}.

\item %
$S$ is $\glocal{\epsilon_1}{\square}{\myg}$. In this case $e_1$
is computable by the induction hypothesis, and $\myg$ is computable (and so loss-computable) by the induction hypthesis. We can apply Lemma~\ref{Alem:glocal}.

\item %
$S$ is $\reset{\square}$. We show by $\Lsc$-induction on computable $e$ that $\reset{e}$ is computable. There are three cases:
\begin{enumerate}
\item If $e$ is a computable value $v$, then 
  for any $\myg_1 \in \Lsc$   we have 
$\myg_1 \vdash_\reset{v}
\xrightarrow{0}v$.

    \item If $e$ is a stuck expression of the form $K[\op(v)]$ 
    with $\op : \outt \to \inn$, $v$ computable and $\reset{K[w]}$ computable 
    for every computable $w\type \inn$, then we simply note that
    $\reset{K[op(v)]}$ is stuck.
    \item Suppose that $e$ is not stuck. Then $\reset{e}$ is not stuck.
    Further, if $\myg_1 \vdash_\epsilon \reset{e} \xrightarrow{r} e''$ for $\myg_1 \in \Lsc$, then, as $e$ is not stuck, for some $r' \in \R$ and $e'$,
    $\myg \vdash_{\epsilon_1} e \xrightarrow{r'} e'$
    and $e'' = \reset{e'}$ and, by the induction hypothesis, $e''$ is computable (and note that $r=0$ in this case).

\end{enumerate}

\end{enumerate}

\end{enumerate}
\end{proof}
It follows at once from this fundamental lemma that all loss continuations are loss-computable.
As usual we can deduce termination from computability. 
\begin{theorem}[Termination] \label{Athm:term} 
For $e_1\type \sigma \etype  \epsilon$ and $\myg \type \sigma\to \bool\etype \epsilon'$ with $\epsilon'\subseteq \epsilon$, there are no infinite sequences:
\[\myg \vdash e_1 \xrightarrow{r_1} e_2 \xrightarrow{r_2} \dots \xrightarrow{r_{n-2}} e_{n-1 }\xrightarrow{r_{n-1}} e_n \dots\]
\end{theorem}
\begin{proof}
 By the Fundamental Lemma $\myg$ is computable and so loss-computable and so then $e$ is $\{\myg\}$-computable. An evident L-induction then shows that there is no such infinite sequence.
\end{proof}

We then have:
\begin{theorem} \label{Athm:bterm}
For $e\type \sigma \etype  \epsilon$ and $\myg \type \sigma\to \bool\etype \epsilon'$ with $\epsilon'\subseteq \epsilon$ we have $\myg \vdash e \xRightarrow{r} w$ for a unique $r \in \R$ and terminal expression $w$ (and then $w\type \sigma\etype \epsilon$).
\end{theorem}

\section{Denotational semantics}

\label{sec:Adenotational}

We first repeat the material on the denotational semantics of \cname, viz the semantics of types (Section~\ref{sec:Asemtypes}), the semantics of expressions (Section~\ref{sec:Asemexp}), and the semantics of handlers (Section~\ref{sec:Ahandsem}). We then  prove  correctness and adequacy in  Section~\ref{sec:Aproofs}.

\subsection{Semantics of Types} \label{sec:Asemtypes}

As discussed in Section~\ref{sec:semtypes}, our semantics employs a family $S_\epsilon(X) = (X \to R_\epsilon) \to W_\epsilon(X)$ of augmented selection  monads,
where 
$W_\epsilon(X) = F_\epsilon(\R \times X)$ and 
$R_\epsilon = F_\epsilon(\R)$. The $W_\epsilon$ are used to interpret  the loss operation and unhandled effect operations, viz, those of the effects in  $\epsilon$ (taking account of their multiplicity).
We remark that the  $W_\epsilon$ are the commutative  combination \cite{HPP06} of the 
$F_\epsilon$  and the writer monad $\R \times -$; algebraically this corresponds to having the  loss operation commute with  the operations. Finally, $R_\epsilon$ is the free $W_\epsilon$-algebra on the one-point set.

Given the $S_\epsilon$, we can define $\ssem{\sigma}$ the semantics of types, as in  Figure~\ref{fig:Atsem},
where we assume available a given semantics $\sem{b}$ of basic types.
\begin{figure}[H]
\small{
\[\begin{array}{lcl}
\ssem{b} & = & \sem{b}\\
 \ssem{(\sigma_1,\dots,\sigma_n)} & = &  \ssem{\sigma_1} \times \dots \times  \ssem{\sigma_n} \\
 \ssem{\sigma + \tau} & = & \ssem{\sigma} +  \ssem{\tau}\\
 \ssem{\nat} & = & \N\\
  \ssem{\mylist(\sigma)} & = &  \ssem{\sigma}^*\\
  \ssem{\sigma \to \tau\etype \epsilon} & = &   \ssem{\sigma} \to   S_\epsilon(\ssem{\tau})\\
\end{array}\] }

\caption{Semantics of types}
\label{fig:Atsem}
\end{figure}

We define  $W_\epsilon(X)$. 
to be the least set $Y$ such that:
\[Y\; = \; \left ( \sum_{\ell \in \epsilon, \op: \outt \xrightarrow{\ell} \inn, 0 < i \leq \epsilon(\ell)} 
 \ssem{\outt} \times Y^{\ssem{\inn}}  \right ) + X \]
There is an inclusion $F_{\epsilon_1}(X) \subseteq F_{\epsilon_2}(X)$ if $\epsilon_1 \subseteq \epsilon_2$; we will use it  without specific comment.  
Note the recursion in these definitions: $\ssem{-}$ is defined using the $S_\epsilon$, the $S_\epsilon$ using the $W_\epsilon$,  the  $W_\epsilon$ using the $F_\epsilon$, and the $F_\epsilon$ using $\ssem{-}$. However,  the well-foundedness assumption (\Cref{sec:termination}) justifies these  definitions.

\subsection{{Semantics of Expressions}} \label{sec:Asemexp}

For the denotational semantics of expressions and handlers we need the monadic structure of the $S_\epsilon$, and so that of the $W_\epsilon$ and the $F_\epsilon$. We  make use of an abbreviation, available for any monad $M$:
 \[\mlet{M}{x\in X}{exp_1}{exp_2} \eqdef (\lambda x \in X.\, exp_2)^{\dagger_M}(exp_1)\]
 for mathematical expressions $exp_1$ and $exp_2$. This abbreviation makes monadic binding available at the meta-level, and that makes for more transparent formulas. 
 
 Say that an 
\textit{$\epsilon$-algebra} is a set $X$ equipped with %
 functions
\[\varphi_{\ell,\op,i} :  \ssem{\outt}  \times X^{\ssem{\inn}}  \to X\]
for $\ell \in \epsilon, \op\type \outt \xrightarrow{\ell} \inn$, and $0 < i \leq \epsilon(\ell)$. 

Then $F_{\epsilon}(X)$ is the free such algebra taking the functions to be: 
\[\varphi^X_{\ell,\op,i}(o,k) \eqdef ((\ell,\op,i),(o,k))\]
with unit $\eta_{F_\epsilon}(x) = x$ (we ignore injections into sums). If we have another such algebra $(Y,\psi_{\ell,\op,i})$,
the unique homomorphic extension $f^{\dagger_{F_\epsilon}}$ of a function $f: X \to Y$ is given by setting 
\[f^{\dagger_{W_\epsilon}}(x) = f(x)\] 
and 
\[f^{\dagger_{W_\epsilon}} ((\ell,\op,i)(out,k, \gamma)) = \psi_{\ell,\op,i}(o, \comp{f^{\dagger_{F_\epsilon}}}{k})\]

The idea of defining monads $F_{\epsilon}$ as free algebras depending on varying families of operation signatures to give the semantics of effect calculi already appears in~\citet{FK19}; here we adapt the idea  to our slightly different treatment of effects.

Turning to $W_\epsilon(X) = F_\epsilon(\R \times X)$, we can again see this as a free algebra monad. Say that an \emph{action $\epsilon$-algebra} is an $\epsilon$-algebra  $(Y,\psi_{\ell,\op,i})$ together with an additive action $\cdot: \R \times Y \to Y$ commuting with the $\psi_{\ell,\gamma,\op,i}$  (by an additive action we mean one such that $0\cdot y = y$ and $r \cdot (s\cdot y) = (r + s)\cdot y$). Then $F_\epsilon(\R \times X)$ is the free such algebra with operations $\varphi^{\R \times X}_{\ell,\op,i}$ and action $\cdot: \R \times F_\epsilon(\R \times X) \to F_\epsilon(\R \times X)$ given by:
\[r\cdot u \eqdef \mlet{F_\epsilon}{s\in \R, x\in X}{u}{(r+s,x)}\] 
The unit is $\eta_{W_\epsilon}(x) = (0,x)$, and if we have another such algebra $(Y,\psi_{\ell,\op,i}, \cdot)$, the unique homomorphic extension $f^{\dagger_{W_\epsilon}}$ of a function $f: X \to Y$ is given by $f^{\dagger_{F_\epsilon}}(r,x) = r\cdot f(x)$ and:
\[f^{\dagger_{W_\epsilon}} ((\ell,\op,i)(out,k)) = \psi_{\ell,\op,i}(o,\comp{f^{\dagger_{W_\epsilon}}}{k})\]

 We now turn to the augmented selection monad $S_\epsilon(X) = (X \to R_\epsilon) \to W_\epsilon(X)$. The unit is given by $(\eta_{S_\epsilon})_X(x) = \lambda \gamma 
\in X  \to R_\epsilon.\, (\eta_{W_\epsilon})_X(x)$ (and recall that $(\eta_{W_\epsilon})_X(x) = (0,x)$). For the Kleisli extension,
rather than  follow the definitions in, e.g.,~\cite{AP21} via a $W_\epsilon$-algebra on $R_\epsilon$ we give definitions that are a little easier to read. 

First, $R_\epsilon$ is an action $\epsilon$-algebra, with 
$\psi_{\ell,\op,i} :  \ssem{\outt}  \times R_\epsilon^{\ssem{\inn}}   \to R_\epsilon$
given by 
$\psi_{\ell,\op,i}(o,k) = 
\varphi^\R_{\ell,\op,i}(o,k)$
and action $\R \times R_\epsilon \to R_\epsilon$ given by:
$r\cdot u \eqdef \mlet{F_{\Onebb,\epsilon}}{s\in \R}{u}{r+s}$. 

Next,  the loss $\ER{\epsilon}{F}{\gamma} \in R_\epsilon$ associated to a selection function $F \in S_\epsilon(Y)$ and loss function $\gamma\type Y \to R_\epsilon$  is
\[\ER{\epsilon}{F}{\gamma} \eqdef \gamma^{\dagger_{W_\epsilon}}(F(\gamma))\]
where we use the fact that $R_\epsilon$ 
is an action $\epsilon$-algebra.

Then, finally, the Kleisli extension $f^{\dagger_{S_\epsilon}}\type S_\epsilon(X) \to S_\epsilon(Y)$  of a function $f\type X \to S_\epsilon(Y)$ is defined~by:
\begin{equation} \label{eqn:Kleisli}
    f^{\dagger_{S_\epsilon}}(F) = \lambda \gamma \in Y \to R_\epsilon.\,\mlet{W_\epsilon}{x \in X}{F(\lambda x\in X.\ER{\epsilon}{fx}{\gamma})}{fx\gamma}
\end{equation}

\renewcommand{\CL}[1]{\mathbf{L}[#1]}

Turning to the denotational semantics, given an environment $\myG = x_1\type \sigma_1,\dots,x_n:\sigma_n$ we take
$\ssem{\myG}$ to be the functions (called \emph{environments}) $\rho$ on $\Dom(\myG)$ such that  
$\rho(x_i)  \in \ssem{\sigma_i}$, for $i = 1,\dots,n$.

Then the denotational semantics of typed expressions and handlers are of the following types:
\begin{mathpar}
\inferrule {\myG \vdash e\type\sigma\etype \epsilon}
  {\ssem{e} :  \ssem{\myG} \to S_\epsilon(\ssem{\sigma})}
\and
\inferrule{\myG \vdash h\type \parr,  \sigma \etype \epsilon\effl \Rightarrow \sigma'\etype\epsilon }
 {\ssem{h} :  \ssem{\myG} \to (\ssem{\parr} \times S_{\epsilon\ell}(\ssem{\sigma}))\to S_{\epsilon}(\ssem{\sigma'})}
\end{mathpar}

The semantics of expressions is given in Figure~\ref{fig:Aesem}.  We assume given semantics $\sem{c} \in \sem{b}$, for constants $c\type b$, and $\sem{f}\type \sem{\sigma}  \to \sem{\tau}$ for basic function symbols $f\type \sigma \to \tau$ (with $0$ and $+$ given their standard meanings).
We make use of an auxiliary ``loss function" semantics 
$\lsem{v}\type \ssem{\myG} \to R_\epsilon$ defined on 
functional values $\myG\vdash v \type \sigma \to \real\etype \epsilon$.
For $v = \lambda^\epsilon x\type \sigma.\,e$ we set:
\[\lsem{\lambda^\epsilon x\type \sigma.\,e}(\rho) = 
\lambda a \in \ssem{\sigma}.\, 
   \mlet{F_\epsilon}{r_1,r_2 \in \R}{\ssem{e}(\rho[a/x]) (\lambda r \in \R.\,0)}{r_2}\]

\begin{figure}[h]

{
\small{\[\begin{array}{lcl}
\ssem{x}(\rho)  & = & \eta_{S_\epsilon}(\rho(x))\\[0.1em]
\ssem{\con}(\rho)  & = & \eta_{S_\epsilon}(\sem{\con}) \\[0.1em]
\ssem{f(e)}(\rho) & = & 
  \mlet{{S_\epsilon}}{a \in \ssem{\sigma}}{\ssem{e}(\rho)}
     {\eta_{S_\epsilon}(\sem{f}(a))} \\
&&\hspace{100pt} (f\type \sigma \to \tau)\\
\ssem{(e_1,\ldots,e_n)}(\rho) & = &
 \mlet{{S_\epsilon}}{a_1 \in \ssem{\sigma_1}}{\ssem{e_1}(\rho)}
                                                       {\\&& \dots \\&& \mlet{{S_\epsilon}}{a_n \in \ssem{\sigma_n}}{\ssem{e_n}(\rho)}{\\&&\eta_{S_\epsilon}((a_1,\dots,a_n))}\\&&}
                                                       
                                      \hspace{100pt}  (\sigma = (\sigma_1,\dots,\sigma_n))
                                                       \\
\ssem{e.i}(\rho)  & = & S_\epsilon(\pi_i)(\ssem{e}(\rho))\\ 
\ssem{\inl{\sigma}{\tau}(e)}(\rho) & = & 
S_\epsilon(\lambda a \in \ssem{\sigma}.\, (0,a))(\ssem{e}(\rho))\\
\ssem{\inr{\sigma}{\tau}(e)}(\rho) & = & 
S_\epsilon(\lambda b \in \ssem{\tau}.\, (1,b))(\ssem{e}(\rho))\\
\ssem{\mycases{e}{x_1\type \sigma_1}{e_1}{x_2\type \sigma_2}{e_2}}(\rho)    & = & 
\mlet{S_\epsilon}{a \in \ssem{\sigma_1} + \ssem{\sigma_2}}{\ssem{e}(\rho)}{\\&&\;[\lambda b_1 \in \ssem{\sigma_1}.\, \ssem{e_1}(\rho[b_1/x_1]),
\\&&\;\;\,\lambda b_2 \in \ssem{\sigma_2}.\, \ssem{e_2}(\rho[b_2/x_2])](a)}\\ 
\ssem{\zero}(\rho) & = & \eta_{S_\epsilon}(0)\\
\ssem{\mysucc{e}}(\rho) & = & \mlet{{S_\epsilon}}{n \in \N}{\ssem{e}(\rho)}
                                   {\eta_{S_\epsilon}(n + 1)}\\ 
\ssem{\iter{e_1}{e_2}{e_3}}(\rho) & = & 
\mlet{{S_\epsilon}}{n \in \N}{\ssem{e_1}(\rho)}
     {\\&&  \mlet{{S_\epsilon}}{a \in \ssem{\sigma}}{\ssem{e_2}(\rho)}{\\&&\mlet{S_\epsilon}{\varphi \in \ssem{\sigma} \to S_\epsilon(\ssem{\sigma})}{\ssem{e_3}(\rho)}
     {\\&& (\varphi^{\dagger_{S_\epsilon}})^n(\eta_{S_\epsilon}(a))}}}\\
\ssem{\nil_\sigma}(\rho) & = & \eta_{S_\epsilon}(\varepsilon)\\
\ssem{\cons{e_1}{e_2}}(\rho) & = & 
\mlet{{S_\epsilon}}
 {a \in \ssem{\sigma_1}}{\ssem{e_1}(\rho)}
     {\\&&  \mlet{{S_\epsilon}}{l \in \ssem{\sigma_1}^*}{\ssem{e_2}(\rho)}{\\&& \eta_{S_\epsilon}(al)}}\\
         &&\hspace{100pt} (\sigma = \mylist(\sigma_1))\\
\ssem{\fold{e_1}{e_2}{e_3}}(\rho) & = &
\mlet{{S_\epsilon}}{l \in \ssem{\sigma_1}^*}{\ssem{e_1}(\rho)}
     {\\&&  \mlet{{S_\epsilon}}{a \in \ssem{\sigma}}{\ssem{e_2}(\rho)}{\\&&\mlet{S_\epsilon}{\varphi \in \ssem{\sigma_1}\times \ssem{\sigma} \to S_\epsilon(\ssem{\sigma})}{\ssem{e_3}(\rho)}}
     {\\&&\myfold{l}{ \eta_{S_\epsilon}(a)}
    {\lambda a_1\in \ssem{\sigma_1}, b \in S_\epsilon(\ssem{\sigma}).\,\varphi(a_1,-)^{\dagger_{S_\epsilon}}(b)}}}\\
         &&\hspace{100pt} (\myG \vdash e_1\type \mylist(\sigma_1)\etype \epsilon)\\
\ssem{\lambda^{\epsilon_1} x\type \sigma.\, e}(\rho) & = & 
\eta_{S_\epsilon}(\lambda a \in \ssem{\sigma}.\, \ssem{e}(\rho[a/x]))\\

\ssem{e_1 ~ e_2}(\rho) & = &
 \mlet{{S_\epsilon}}
 {\varphi \in \ssem{\sigma_1} \to S_\epsilon(\ssem{\sigma})}{\ssem{e_1}(\rho)}
     {\\&&  \mlet{{S_\epsilon}}{a \in \ssem{\sigma_1}}{\ssem{e_2}(\rho)}{\\&&\varphi(a)}}\\
     &&\hspace{100pt} (\Gamma \vdash e_1 \type \sigma_1 \to \sigma\etype \epsilon)\\[0.1em] 

\ssem{\op(e)}(\rho)  & = &
 \mlet{S_\epsilon}
         {a \in \ssem{\outt}}
         {\ssem{e}(\rho)}
         {\\&&\quad\lambda \gamma \in R_\epsilon^{\ssem{\inn}}.\,\varphi^{\R \times \ssem{\sigma}}_{\ell,\op,\epsilon(\ell)} 
           (a, (\eta_{W_\epsilon})_{\ssem{\inn}})}\\
&&\hspace{100pt} (\op\type \outt \xrightarrow{\effl} \inn)\\
\ssem{\wph{h}{e_1}{e_2}}(\rho) & = & 
\mlet{{S_\epsilon}}{a \in \ssem{par}}{\ssem{e_1}(\rho)}{\ssem{h}(\rho)(a,\ssem{e_2}(\rho))}\\
 &&\hspace{100pt}  (\myG \vdash e_1\type par)  \\[0.1em]
\ssem{\then{\epsilon_1}{e_1}
  {\lambda^\epsilon x\type \sigma_1.\, e_2}}(\rho) 
&=& \lambda \gamma \in  R_\epsilon^{\ssem{\sigma}}.\,\\
&&\quad\;\mlet{F_\epsilon}{r_1 \in \R, a \in \ssem{\sigma_1}}{\ssem{e_1}(\rho)(\lsem{\lambda^\epsilon x\type \sigma_1.\, e_2}(\rho))}
{\\&&\;\quad \mlet{F_\epsilon}
                {r_2,r_3\in \R}
                {\ssem{e_2}(\rho[a/x])(\lambda r \in \R.\,0)}
                {(r_2,r_1 + r_3})}\\

\ssem{\glocal{\epsilon_1}{e}{\myg}}(\rho) & = & 
\lambda \gamma \in   R_\epsilon^{\ssem{\sigma}}.\,\ssem{e}(\rho)\lsem{\myg}(\rho)\\
\ssem{\reset{e}}(\rho) & = & 
\lambda \gamma \in  R_\epsilon^{\ssem{\sigma}}.\,%
\mlet{F_\epsilon}
            {r_1 \in \R, a \in \ssem{\sigma}}
            {\ssem{e}(\rho)(\gamma)}
            {(0,a)}
\end{array}\]}

}

\caption{Semantics of expressions}
\label{fig:Aesem}
\end{figure}

\subsection{Semantics of handlers} \label{sec:Ahandsem} 

We build up the semantics of handlers in stages. Consider a handler $h$:
 \[\left \{\begin{array}{l}
                     \op_1 \mapsto \lambda^\epsilon z\type (\parr, \outt_1,
                     (\parr,\inn_1) \to \real \etype \epsilon,
                     (\parr,\inn_1) \to \sigma' \etype \epsilon).~ e_1,\dots, \\
                     \op_n \mapsto \lambda^\epsilon z\type (\parr, \outt_n,
                     (\parr,\inn_n) \to \real \etype \epsilon,
                     (\parr,\inn_n) \to \sigma' \etype \epsilon).~ e_n,\\
                   \return \mapsto \lambda^\epsilon z\type (\parr,  \sigma).\, e_r
                            \end{array}\right \} \]
where $\myG \vdash h\type \parr, \sigma\etype \epsilon\effl \Rightarrow \sigma'\etype\epsilon$, and fix  $\rho \in \ssem{\myG}$ and $\gamma \in \R^\ssem{\sigma'}$.
We first construct  an $\epsilon\ell$-algebra 
\[A = (W_\epsilon(\ssem{\sigma'})^{\ssem{\parr}},\psi_{\ell,\op,i})\]

So for $\ell_1 \in \epsilon\ell$, $\op\type \outt \xrightarrow{\ell_1} \inn$, and $0 < i \leq (\epsilon\ell)\ell_1$ we need functions 
\[\psi_{\ell,\op,i}\type \ssem{\outt} \times \left (W_\epsilon(\ssem{\sigma'})^{\ssem{\parr}}\right )^{\ssem{\inn}}  \to W_\epsilon(\ssem{\sigma'})^{\ssem{\parr}}\]

For $\ell_1 \in \epsilon$, $\op\type \outt \xrightarrow{\ell_i} \inn$, and $0 < i \leq \epsilon(\ell_1)$, we set

\[\psi_{\ell_1,\op,i}(o,k) = 
 \lambda p \in \ssem{\parr}.\, 
   ((\ell_1,\op,i), (o, \lambda a \in \ssem{\inn}.\, kap))\]

 For $\op_j \mapsto   \lambda^\epsilon z\type (\parr, \outt_j,
                     (\parr,\inn_j) \to \real \etype \epsilon,
                     (\parr,\inn_j) \to \sigma' \etype \epsilon).~ e_j\in h$ and  $i = \epsilon(\ell) +1$ we set
 \[\psi_{\ell,\op_j,i}(o,k) = 
    \lambda p\in \ssem{\parr}.\, \ssem{e_j}(\rho[(p,o,l_1,k_1)/z])\gamma
 \]

 where
 \[k_1 = \lambda p\in \ssem{\parr}, a \in \ssem{\inn_j}.\, \lambda \gamma_1 \in  \R^{\ssem{\sigma'}}.\,kap\]
 and
 \[l_1 = \lambda p\in \ssem{\parr}, a \in \ssem{\inn_j}.\, \lambda \gamma_1 \in  \R^{\ssem{\sigma'}}.\, \delta_0(\gamma^{\dagger_{W_\epsilon}}(kap))\]

where in the definition of $l_1$ we use the fact that $R_\epsilon$ 
is an action $\epsilon$-algebra, and where 
$\delta_\epsilon\type F_\epsilon(\R) \to F_\epsilon(\R \times \R)$
is the evident conversion function $F_\epsilon(\lambda r\in \R.\, (0,r))$

  We will use this algebra to extend the map  $s\type \R \times \ssem{\sigma} \to A$ defined by
  \[s(r,a) = \lambda p\in \ssem{\parr}.\, r\cdot(\ssem{e_r}(\rho[(p,a)/z])\gamma)\] 
  (Recall that $ \return \mapsto \lambda^\epsilon z\type (\parr,  \sigma).\, e_r$ is in $h$.)
  The semantics of the handler $h$ is then given by:
 \[\ssem{h}(\rho)(p,G)(\gamma) = s^{\dagger_{F_{\epsilon\ell}}}(G(\lambda a \in \ssem{\sigma}.\, \ER{\epsilon}{\ssem{e}(\rho[(p,a)/z])}{\gamma}))(p)\]
 Note that
 \[s^{\dagger_{F_{\epsilon\ell}}}: F_{\epsilon\ell}(\R \times \ssem{\sigma}) 
 \to F_{\epsilon}(\R \times \ssem{\sigma'})^{\ssem{\parr}}\] 
 depends on the choices of $\rho$ and $\gamma$, though that is not reflected in the notation.
 \subsection{Proof of theorems}\label{sec:Aproofs}
 
\begin{lemma}(Substitution)\label{lem:sublem}
Suppose that $\myG \vdash v\type \sigma$ and 
$\myG, x\type \sigma\vdash e\type \tau\etype\epsilon$.
Then
$\myG\vdash e[v/x]\type \tau\etype\epsilon$.
\end{lemma}

The following ``value semantics"  for values allows us to state our soundness and adequacy results. It follows the following scheme:
\[\frac{\myG\vdash v\type \sigma}{\vsem{v}\type \ssem{\myG} \to \ssem{\sigma}}\]
{\small \[\begin{array}{lcl}
\vsem{x}(\rho) &=& \rho(x)\\
\vsem{c}(\rho) &=& \sem{c}\\
\vsem{(v_1,\dots,v_n)}(\rho) &=& (\vsem{v_1}(\rho),\dots,\vsem{v_n}(\rho))\\
\vsem{\inl{\sigma}{\tau}(v)}(\rho) &=& (0,\vsem{v}(\rho))\\
\vsem{\inr{\sigma}{\tau}(v)}(\rho) &=& (1,\vsem{v}(\rho))\\
\vsem{\zero}(\rho) &=& 0\\
\vsem{\mysucc{v}}(\rho) &=& \vsem{v}(\rho) +1\\
\vsem{\nil}(\rho) &=& \varepsilon\\
\vsem{\cons{v_1}{v_2}}(\rho) &=& \vsem{v_1}(\rho)\vsem{v_2}(\rho)\\
\vsem{\lambda^{\epsilon_1} x\type \sigma_1.e}(\rho) &=& \lambda a \in \ssem{\sigma}.\, \ssem{e}(\rho[a/x])\\
\end{array}\]}

Below we may omit $\rho$ in $\ssem{e}(\rho)$ (or $\vsem{v}(\rho)$) when $e$ (respectively $v$) is closed. 
 
 \begin{lemma} \label{Alem:valsem} For any value $\myG \vdash v\type \sigma \etype \epsilon$ we have:
 \[\ssem{v}(\rho) = \eta_{S_\epsilon}(\vsem{v}(\rho))\]
 \end{lemma}
\begin{proof}
The proof is by structural induction, split into cases according to the form of $v$:
\begin{enumerate}
    \item  Suppose that $e$ has the form $x$. Then we calculate:
    \[\begin{array}{lcl}
       \ssem{x}(\rho)  & = & \eta_{S_\epsilon}(\rho(x))\\
         & = & \eta_{S_\epsilon}(\vsem{x}(\rho))\\
    \end{array}\]
    
    \item  Suppose that $e$ has the form $c$. Then we calculate:
    \[\begin{array}{lcl}
       \ssem{c}(\rho)  & = & \eta_{S_\epsilon}(\sem{\con})\\
         & = & \eta_{S_\epsilon}(\vsem{c}(\rho))\\
    \end{array}\]
    
    \item  Suppose that $e$ has the form $(v_1,\dots,v_n)$. Then we calculate:
    \[\begin{array}{lcl}
       \ssem{(v_1,\dots,v_n)}(\rho)  & = & 
 \mlet{{S_\epsilon}}{a_1 \in \ssem{b_1}}{\ssem{v_1}(\rho)}
                                                       {\\&& \dots \\&& \mlet{{S_\epsilon}}{a_n \in \ssem{b_n}}{\ssem{v_n}(\rho)}{\\&&\eta_{S_\epsilon}((a_1,\dots,a_n))}} \\
         & = & 
 \mlet{{S_\epsilon}}{a_1 \in \ssem{b_1}}{\eta_{S_\epsilon}(\vsem{v_1}(\rho))}
                                                       {\\&& \dots \\&& \mlet{{S_\epsilon}}{a_n 
                    \ssem{b_n}}{\eta_{S_\epsilon}(\vsem{v_n}(\rho))}{\\&&\eta_{S_\epsilon}((a_1,\dots,a_n))}} \\
                    & = &\eta_{S_\epsilon}( (\vsem{v_1}(\rho),\dots,\vsem{v_n}(\rho)))\\
                    & = & \eta_{S_\epsilon}(\vsem{(v_1,\dots,v_n)})
    \end{array}\]
    
    \item  Suppose that $e$ has the form $\inl{\sigma}{\tau}(v)$. Then we calculate:
    \[\begin{array}{lcl}
       \ssem{\inl{\sigma}{\tau}(v)}(\rho)  & = & 
          S_\epsilon(\lambda a \in \ssem{\sigma}.\, (0,a))(\ssem{v}(\rho))\\
         & = & S_\epsilon(\lambda a \in \ssem{\sigma}.\, (0,a))(\eta_{S_\epsilon}(\vsem{v}(\rho)))\\
         & = &\eta_{S_\epsilon}((0,\vsem{v}(\rho))) \\
         & = & \eta_{S_\epsilon}(\vsem{\inl{\sigma}{\tau}(v)}(\rho))
    \end{array}\]
    The case where $e$ has the form $\inr{\sigma}{\tau}(v)$ is similar.
    \item  Suppose that $e$ has the form $\zero$. Then we calculate:
    \[\begin{array}{lcl}
       \ssem{\zero}(\rho)  & = & \eta_{S_\epsilon}(0)\\
         & = & \eta_{S_\epsilon}(\vsem{\zero}(\rho))\\
    \end{array}\]
    \item  Suppose that $e$ has the form $\mysucc{v}$. Then we calculate:
    \[\begin{array}{lcl}
       \ssem{\mysucc{v}}(\rho)  & = & \mlet{{S_\epsilon}}{n \in \N}{\ssem{v}(\rho)}{\eta_{S_\epsilon}(n + 1)}\\
       & = &  \mlet{{S_\epsilon}}{n \in \N}{\eta_{S_\epsilon}(\vsem{v}(\rho))}
                   {\eta_{S_\epsilon}(n + 1)}\\
         & = & \eta_{S_\epsilon}(\vsem{v}(\rho) + 1)\\
         & = & \eta_{S_\epsilon}(\vsem{\mysucc{v}}(\rho))
    \end{array}\]
    
    \item  Suppose that $e$ has the form $\nil$. Then we calculate:
    \[\begin{array}{lcl}
       \ssem{\nil}(\rho)  & = & \eta_{S_\epsilon}(\varepsilon)\\
         & = & \eta_{S_\epsilon}(\vsem{\nil}(\rho))\\
             \end{array}\]
    
    \item  Suppose that $e$ has the form $\cons{v_1}{v_2}$. Then we calculate:
    \[\begin{array}{lcl}
      \ssem{\cons{v_1}{v_2}}(\rho)  & = & 
        \mlet{{S_\epsilon}}
 {a \in \ssem{\sigma_1}}{\ssem{v_1}(\rho)}
     {\\&&  \mlet{{S_\epsilon}}{l \in \ssem{\sigma_1}^*}{\ssem{v_2}(\rho)}{\\&& \eta_{S_\epsilon}(al)}}
     \\
         & = & 
\mlet{{S_\epsilon}}
 {a \in \ssem{\sigma_1}}{\eta_{S_\epsilon}(\vsem{v_1}(\rho))}
     {\\&&  \mlet{{S_\epsilon}}{l \in \ssem{\sigma_1}^*}{\eta_{S_\epsilon}(\vsem{v_2}(\rho))}{\\&& \eta_{S_\epsilon}(al)}}
     \\
     & = & \eta_{S_\epsilon}(\vsem{v_1}(\rho)\vsem{v_2}(\rho))\\
     & = & \eta_{S_\epsilon}(\vsem{\cons{v_1}{v_2}}(\rho))
    \end{array}\]
    
    \item  Suppose that $e$ has the form 
    $\lambda^{\epsilon_1} x\type \sigma_1.e$. Then we calculate:
    \[\begin{array}{lcl}
       \ssem{\lambda^{\epsilon_1} x\type \sigma_1.e}(\rho)  
         & = & \eta_{S_\epsilon}(\lambda a \in \ssem{\sigma}.\, \ssem{e}(\rho[a/x]))\\
         & = & \eta_{S_\epsilon}(\vsem{\lambda^{\epsilon_1} x\type \sigma_1.e}(\rho))\\ 

    \end{array}\]

\end{enumerate}

\end{proof}
\begin{lemma} \label{lem:actc}
With the notation of the handler semantics, the following diagram commutes:
\[\begin{tikzcd}
	{F_{\epsilon\ell}(\R \times \ssem{\sigma})} &&& {F_{\epsilon}(\R \times \ssem{\sigma'})^{\ssem{\parr}}} \\
	\\
	\\
	{F_{\epsilon\ell}(\R \times \ssem{\sigma})} &&& {F_{\epsilon}(\R \times \ssem{\sigma'})^{\ssem{\parr}}}
	\arrow["{s^{\dagger_{F_{\epsilon\ell}}}}", from=1-1, to=1-4]
	\arrow["{r\cdot -}"', from=1-1, to=4-1]
	\arrow["{(r\cdot -)^P}", from=1-4, to=4-4]
	\arrow["{s^{\dagger_{F_{\epsilon\ell}}}}"', from=4-1, to=4-4]
\end{tikzcd}\]
 
\end{lemma}
\begin{proof} The map $\comp{s^{\dagger_{F_{\epsilon\ell}}}}{(r\cdot -)}$ 
is the extension of the map 
\[\comp{s}{(r\cdot -)}\type \R \times \ssem{\sigma} \to F_{\epsilon}(\R \times \ssem{\sigma'})^{\ssem{\parr}} \]
and the map $\comp{(r\cdot -)}{s^{\dagger_{F_{\epsilon\ell}}}}$
is the extension of the map 
\[\comp{(r\cdot -)^{\ssem{\parr}}}{s}\type \R \times \ssem{\sigma} \to F_{\epsilon}(\R \times \ssem{\sigma'})^{\ssem{\parr}} \].
We prove those two maps are equal.
We have: 
\[(\comp{s}{(r\cdot -)})(r',a) = s(r+r',a) = \lambda p\in \ssem{\parr}.\, (r+r')\cdot(\ssem{e_r}(\rho[(p,a)/z])\gamma)\]
and we have:
\[\begin{array}{lcl}
 (\comp{(r\cdot -)^{\ssem{\parr}}}{s})(r',a)
 & = &
 (r\cdot -)^{\ssem{\parr}}(\lambda p\in \ssem{\parr}.\, r'\cdot(\ssem{e_r}(\rho[(p,a)/z])\gamma))\\
 & = &
 \lambda p\in \ssem{\parr}.\, r\cdot (r'\cdot(\ssem{e_r}(\rho[(p,a)/z])\gamma))\\
 & = &
 \lambda p\in \ssem{\parr}.\, (r + r')\cdot(\ssem{e_r}(\rho[(p,a)/z])\gamma)\\
 \end{array}\]
recalling for the last equality that $F_{\epsilon}(\R \times \ssem{\sigma'})$ is an action $\epsilon$-algebra.
\end{proof}

\begin{lemma} \label{lem:smalllem}
\myskip
\begin{enumerate}
\item For any $\myG \vdash v\type \sigma$ and 
     $\myG \vdash  e\type \tau\etype \epsilon$ we have:
     \[\ssem{(v,e)}(\rho) = \mlet{S_\epsilon}{b\in \ssem{\tau}}{\ssem{e}(\rho)}{\eta_{S_\epsilon}((\vsem{v}(\rho),b))}\]
\item For any $\myG \vdash v\type \sigma \to \tau\etype \epsilon$ and $\myG \vdash v\type e\type \sigma\etype \epsilon$ we have:
\[\ssem{v ~ e}(\rho)  = \mlet{{S_\epsilon}}
        {a \in \ssem{\sigma}}{\ssem{e}(\rho)}{\vsem{v}(\rho)(a)}\]
\end{enumerate} 
\end{lemma}
\begin{proof}%
\myskip
\begin{enumerate}
\item
We calculate:
\[\begin{array}{lcl}
\ssem{(v,e)}(\rho) 
 & = & 
 \mlet{{S_\epsilon}}{a \in \ssem{\sigma}}{\ssem{v}(\rho)}
                                                       {\\&&  \mlet{{S_\epsilon}}{b \in \ssem{\tau}}{\ssem{e}(\rho)}{\\&&\eta_{S_\epsilon}((a,b))}} \\
 & = &
   \mlet{{S_\epsilon}}{a \in \ssem{\sigma}}{\eta_{S_\epsilon}(\vsem{v}(\rho))}
                                                       {\\&&  \mlet{{S_\epsilon}}{b \in \ssem{\tau}}{\ssem{e}(\rho)}{\\&&\eta_{S_\epsilon}((a,b))}}\\
& = &

\mlet{{S_\epsilon}}
     {b \in \ssem{\tau}}
     {\ssem{e}(\rho)}
     {\\&&\eta_{S_\epsilon}((\vsem{v}(\rho),b))}

\end{array}\]
\item
We calculate:
\[\begin{array}{lcl}\ssem{v ~ e}(\rho) & = &
 \mlet{{S_\epsilon}}
 {\varphi \in \ssem{\sigma} \to S_\epsilon(\ssem{\sigma})}{\ssem{v}(\rho)}
     {\\&&  \mlet{{S_\epsilon}}{a \in \ssem{\sigma}}{\ssem{e}(\rho)}{\varphi(a)}}\\
     & = &
     \mlet{{S_\epsilon}}
 {\varphi \in \ssem{\sigma} \to S_\epsilon(\ssem{\sigma})}{\eta_{S_\epsilon}(\vsem{v}(\rho))}
     {\\&&  \mlet{{S_\epsilon}}{a \in \ssem{\sigma}}{\ssem{e}(\rho)}{\varphi(a)}}\\
     & = & \mlet{{S_\epsilon}}{a \in \ssem{\sigma}}{\ssem{e}(\rho)}{\vsem{v}(\rho)(a)}
     \end{array}\]
\end{enumerate}
\end{proof}

\begin{lemma} \label{lem:theng}
\myskip
\begin{enumerate}
    \item 
For $\myG \vdash e \type \sigma\etype \epsilon$
and $\myG \vdash \myg\type \sigma \to \real\etype \epsilon_1$ 
with $\epsilon_1 \subseteq \epsilon$ we have:
\[\ssem{\then{}{e}{\myg}}(\rho)\gamma_1 = \delta_\epsilon(\ER{\epsilon}{\ssem{e}(\rho)}{\lsem{\myg}(\rho)})\]
\item For $\myG \vdash \myg\type \sigma \to \real\etype \epsilon$
we have 
\[\comp{\delta_\epsilon}{\lsem{\myg}(\rho)} = \lambda a \in \ssem{\sigma}.\,\vsem{\myg}(\rho) a \gamma_1\]
\item
For $\myG,x\type\sigma \vdash e\type\tau\etype\epsilon$ and $\myG,x\type\sigma \vdash \myg\type \tau \to \real\etype\epsilon_1$ with $\epsilon_1 \subseteq \epsilon$ we have:
\[\lsem{\lambda^\epsilon x\type \sigma.\,\then{\epsilon}{e}{\myg}}(\rho) = 
\lambda a \in \ssem{\sigma}.\, \ER{\epsilon}{\ssem{e}(\rho[a/x])}{\lsem{\myg}(\rho[a/x])}
 \]
\end{enumerate}

\end{lemma}
\begin{proof}

We prove the first two parts by a mutual induction.

For the first part we calculate:
\[\begin{array}{lcll}
     \ssem{\then{}{e}{\myg}}(\rho)\gamma_1
     & = &
\mlet{F_\epsilon}{r_1 \in \R, a \in \ssem{\sigma}}{\ssem{e}(\rho)(\lsem{\myg}(\rho))}
{\\&&\;\quad \mlet{F_\epsilon}
                {r_2,r_3\in \R}
                {\vsem{\myg}(\rho)a(\lambda r \in \R.\,0)}
                {(r_2,r_1 + r_3})} \\     
& = &
\mlet{F_\epsilon}{r_1 \in \R, a \in \ssem{\sigma}}{\ssem{e}(\rho)(\lsem{\myg}(\rho))}
{\\&&\;\quad \mlet{F_\epsilon}
                {r_2,r_3\in \R}
                {\delta_\epsilon(\lsem{\myg}(\rho)a)}
                {(r_2,r_1 + r_3})} & (\mbox{using the induction hypothesis})\\
& = &
\mlet{F_\epsilon}{r_1 \in \R, a \in \ssem{\sigma}}{\ssem{e}(\rho)(\lsem{\myg}(\rho))}
{\\&&\;\quad \mlet{F_\epsilon}
                {r_2,r_3\in \R}
                {\delta_\epsilon(\lsem{\myg}(\rho)a)}
                {(0,r_1 + r_3})}\\
& = &
\mlet{F_\epsilon}{r_1 \in \R, a \in \ssem{\sigma}}{\ssem{e}(\rho)(\lsem{\myg}(\rho))}
{\\&&\;\quad \mlet{F_\epsilon}
                {r_3\in \R}
                {\lsem{\myg}(\rho)a}
                {(0,r_1 + r_3})} \\
                & = &
\delta_\epsilon(\mlet{F_\epsilon}{r_1 \in \R, a \in \ssem{\sigma}}{\ssem{e}(\rho)(\lsem{\myg}(\rho))}
{\\&&\;\quad \mlet{F_\epsilon}
                {r_3\in \R}
                {\lsem{\myg}(\rho)a}
                {r_1 + r_3}}) \\
& = &
\delta_\epsilon(\ER{\epsilon}{\ssem{e}(\rho)}{\lsem{\myg}(\rho)})
                \end{array}\]
For the second part, the case where $\myg = \lambda^{\epsilon} x\in \sigma.\, 0$ is evident. For the case where $\myg$ has the form $\lambda^{\epsilon} x \in \sigma.\then{}{e}{\myg_1}$, we calculate:
\[\begin{array}{lcl}
    \delta_\epsilon(\lsem{\myg}(\rho)a)
     & = & 
    \delta_\epsilon(\mlet{F_\epsilon}{r_1,r_2\in \R}{\vsem{\myg}(\rho)a (\lambda r \in \R.\,0)}{r_2}) \\
    & = & 
    \mlet{F_\epsilon}{r_1,r_2\in \R}{\vsem{\myg}(\rho)a (\lambda r \in \R.\,0)}{(0,r_2)} \\
    & = & 
    \mlet{F_\epsilon}{r_1,r_2\in \R}{\ssem{\then{}{e}{\myg_1}}(\rho[a/x])(\lambda r \in \R.\,0)}{(0,r_2)} \\
    & = & 
    \mlet{F_\epsilon}{r_1,r_2\in \R}{\delta_\epsilon(\ER{\epsilon}{\ssem{e}(\rho[a/x])}{\lsem{\myg}(\rho[a/x])})}{(0,r_2)} \\\\
    & = & 
    \mlet{F_\epsilon}
          {r_1,r_2\in \R}
          {\\&&\quad (\delta_\epsilon(
                \mlet{F_\epsilon}{r'_1 \in \R, a \in \ssem{\sigma}}{\ssem{e}(\rho[a/x])(\lsem{\myg_1}(\rho[a/x]))}
{\\&&\;\quad \mlet{F_\epsilon}
                {r'_3\in \R}
                {\lsem{\myg_1}(\rho[a/x])a}
                {r'_1 + r'_3}}) )\\ &&\!\!}
          { (0,r_2)} \\\\
& = & 
    \mlet{F_\epsilon}
          {r_1,r_2\in \R}
          {\\&&\quad (
                \mlet{F_\epsilon}{r'_1 \in \R, a \in \ssem{\sigma}}{\ssem{e}(\rho[a/x])(\lsem{\myg_1}(\rho[a/x]))}
{\\&&\;\quad \mlet{F_\epsilon}
                {r'_3\in \R}
                {\lsem{\myg_1}(\rho[a/x])a}
                {(0,r'_1 + r'_3)}} )\\ &&\!\!}
          { (0,r_2)} \\\\
&= &
\mlet{F_\epsilon}{r'_1 \in \R, a \in \ssem{\sigma}}{\ssem{e}(\rho[a/x])(\lsem{\myg_1}(\rho[a/x]))}
{\\&&\;\quad \mlet{F_\epsilon}
                {r'_3\in \R}
                {\lsem{\myg_1}(\rho[a/x])a}
                {\\&&\quad\quad \mlet{F_\epsilon}{r_1,r_2\in \R}{(0,r'_1 + r'_3)}{(0,r_2)}}} \\\\
&= &
\mlet{F_\epsilon}{r'_1 \in \R, a \in \ssem{\sigma}}{\ssem{e}(\rho[a/x])(\lsem{\myg_1}(\rho[a/x]))}
\\&&\;\quad \mlet{F_\epsilon}
                {r'_3\in \R}
                {\lsem{\myg_1}(\rho[a/x])a}
                {(0,r'_1 + r'_3)} \\
&= &
\mlet{F_\epsilon}{r'_1 \in \R, a \in \ssem{\sigma}}{\ssem{e}(\rho[a/x])(\lsem{\myg_1}(\rho[a/x]))}
\\&&\;\quad \mlet{F_\epsilon}
                {r'_2,r'_3\in \R}
                {\vsem{\myg_1}(\rho[a/x])a (\lambda r\in \R.\, 0))}
                {(r'_2,r'_1 + r'_3)}\\ && \qquad(\mbox{using the induction hypothesis})\\
& = &
\ssem{\then{}{e}{\myg_1}}(\rho[a/x])\gamma_1\\
&=&
\ssem{\myg}(\rho)a\gamma_1\\
\end{array}\]

The third part follows from the first two:
\[\begin{array}{lcl}
   \comp{\delta_\epsilon}{ \lsem{\lambda^\epsilon x\type \sigma.\,\then{\epsilon}{e}{\myg}}(\rho)} 
     &= & 
     \lambda a \in \ssem{\sigma}.\,\vsem{\lambda^\epsilon x\type \sigma.\,\then{\epsilon}{e}{\myg}}(\rho) a \gamma_1\\
     & = &
     \lambda a \in \ssem{\sigma}.\,\ssem{\then{\epsilon}{e}{\myg}}(\rho[a/x])  \gamma_1\\
     & = &
     \lambda a \in \ssem{\sigma}.\,\delta_\epsilon(\ER{\epsilon}{\ssem{e}(\rho[a/x])}{\lsem{\myg}(\rho[a/x])})\\
     & = &
     \comp{\delta_\epsilon}{\lambda a \in \ssem{\sigma}.\,\ER{\epsilon}{\ssem{e}(\rho[a/x])}{\lsem{\myg}(\rho[a/x])}}\\
\end{array}\]
and we can cancel the $\delta_\epsilon$'s.
\end{proof}

\begin{lemma} \label{lem:Fsem} For $e\type \sigma\etype\epsilon$ and $F[e]\type \tau\etype\epsilon$ we have:
 \[\ssem{F[e]}(\rho) = \mlet{S_\epsilon}{a \in \ssem{\sigma}}{\ssem{e}(\rho)}{\ssem{F[x]}(\rho[x/a])}\]
 \end{lemma}
 \begin{proof}
 The proof is by cases on the form of $F$. We consider 
 some illustrative examples.
 
 \begin{enumerate}
     \item For $F = f(\square)$ with $f\type \sigma \to \tau $ 
     we have
     \[\begin{array}{lcl} \ssem{f(x)}(\rho)
     & = & 
     \mlet{{S_\epsilon}}{b \in \ssem{\sigma}}{\ssem{x}(\rho[a/x])}
     {\eta_{S_\epsilon}(\sem{f}(b))}\\
     & = &
     \mlet{{S_\epsilon}}{b \in \ssem{\sigma}}{a}
     {\eta_{S_\epsilon}(\sem{f}(b))}\\
     & = & 
     {\eta_{S_\epsilon}(\sem{f}(a))}
     \end{array}\]
     and so:
     \[\begin{array}{lcl}
     \ssem{f(e)} (\rho)
        & = & \mlet{{S_\epsilon}}{a \in \ssem{\sigma}}{\ssem{e}(\rho)}
     {\eta_{S_\epsilon}(\sem{f}(a))}\\
        & = & \mlet{{S_\epsilon}}{a \in \ssem{\sigma}}{\ssem{e}(\rho)}
     {\ssem{f(x)}(\rho[a/x]x)}
     \end{array}\]
     
    as required
     \item For $F = (v_1,\dots, v_k, \square, e_{k+2},\dots, e_n)$ we calculate:
      \[\begin{array}{lcl}
    \ssem{(v_1,\dots, v_k, e, e_{k+2},\dots, e_n)}(\rho) & = & 
    \mlet{{S_\epsilon}}{a_1 \in \ssem{\sigma_1}}{\ssem{v_1}(\rho)}
     {\\&& \dots \\&& \mlet{{S_\epsilon}}{a_k\in \ssem{\sigma_k}}{\ssem{v_k}(\rho)}
     {\\&& \mlet{S_\epsilon}{a_{k+1}\in \ssem{\sigma_{k+1}}}{\ssem{e}(\rho)}
     {\\&&\mlet{{S_\epsilon}}{a_{k+2} \in \ssem{\sigma_{k+2}}}{\ssem{e_{k+2}}(\rho)}
     {\\&& \dots 
      \\&& \mlet{{S_\epsilon}}{a_n \in \ssem{\sigma_n}}{\ssem{e_n}(\rho)}
      {\\&&\eta_{S_\epsilon}((a_1,\dots,a_n))}}}}}\\
      
     &=&
     \mlet{S_\epsilon}{a_{k+1}\in \ssem{\sigma_{k+1}}}{\ssem{e}(\rho)}
     {\\&&\mlet{{S_\epsilon}}{a_{k+2} \in \ssem{\sigma_{k+2}}}{\ssem{e_{k+2}}(\rho)}
     {\\&& \dots 
      \\&& \mlet{{S_\epsilon}}{a_n \in \ssem{\sigma_n}}{\ssem{e_n}(\rho)}
      {\\&&\eta_{S_\epsilon}((\vsem{v_1},\dots,\vsem{v_k},a_{k+1},\dots, a_n))}}}\\
     & = &
     \mlet{S_\epsilon}{a_{k+1} \in \ssem{\sigma_{k+1}}}{\ssem{e}(\rho)}
     {\ssem{(v_1,\dots, v_k, x, e_{k+2},\dots, e_n)}(\rho[a_{k+1}/x])}\\
     \end{array}\]

     \item 
     
     For $F = \square.i$ we calculate:
     \[\begin{array}{lcl}
    
     \ssem{e.i}(\rho)  
       & = & 
        S_\epsilon(\pi_i)(\ssem{e}(\rho))\\ 
        & = &
        \mlet{S_\epsilon}{a\in (\sigma_1,\dots,\sigma_n)}{\ssem{e}(\rho)}{\pi_i(a)}\\
        & = &
        \mlet{S_\epsilon}{a\in (\sigma_1,\dots,\sigma_n)}{\ssem{e}(\rho)}{\ssem{\pi_i(x)}(\rho[a/x])}
     \end{array}\]
     
 \end{enumerate}
 \end{proof}
 
 \begin{lemma} \label{lem:Fsemg} For an expression $F[e]\type \tau$, 
 where $e\type \sigma$, and a loss continuation 
 $\myg\type \tau \to \real\etype\epsilon$ we have:
 \[\begin{array}{lcl}\ssem{F[e]}\lsem{\myg} 
  & = &
 \mlet{W_\epsilon}{a \in \ssem{\sigma}\\&&\quad}
      {\ssem{e}\lsem{\lambda^\epsilon x\type \sigma.\,\then{\epsilon}{F[x]}{\myg}}\\&&\quad}
      {\ssem{F[x]}(x\mapsto a)\lsem{\myg}}
     \end{array}\]
 
 \end{lemma}
 \begin{proof}
 \[\begin{array}{lcll}
   \ssem{F[e]}\lsem{\myg}  
    & = &  
    (\mlet{S_\epsilon}{a \in \ssem{\sigma}\\&&\quad}
    {\ssem{e}\\&&\quad}{\ssem{F[x]}(x\mapsto a)})\lsem{\myg}
    & (\mbox{by Lemma~\ref{lem:Fsem}})\\
    & = & 
    \mlet{W_\epsilon}{a \in \ssem{\sigma}\\&&\quad}
     {\ssem{e}(\lambda a\in \ssem{\sigma}.\ER{\epsilon}{\ssem{F[x]}(x\mapsto a)}{\lsem{\myg}})\\&&\quad}
     {\ssem{F[x]}(x\mapsto a)\lsem{\myg}}
     & (\mbox{by Equation~\ref{eqn:Kleisli}})\\
      & = & 
    \mlet{W_\epsilon}{a \in \ssem{\sigma}\\&&\quad}
      {\ssem{e}\lsem{\lambda^\epsilon x\type \sigma.\,\then{\epsilon}{F[x]}{\myg}}\\&&\quad}
      {\ssem{F[x]}(x\mapsto a)\lsem{\myg}}
      & (\mbox{by Lemma~\ref{lem:theng}})\\
 \end{array}\]
\end{proof}

 \begin{lemma} \label{Alem:Kop} For $K[\op(v)]\type \sigma\etype\epsilon$ with $\op \notin \heff(K)$ and $\op\type \outt \xrightarrow{\ell} \inn$
 we have:
 \[\ssem{K[\op(v)]}(\rho)(\gamma) = 
 \varphi^{\R \times \ssem{\sigma}}_{\ell,\op,\epsilon(\ell)}
 (\vsem{v}(\rho),\lambda a \in \ssem{\inn}.\,
     \ssem{K[x]}(\rho[a/x])(\gamma))\]
 \end{lemma}
 \begin{proof}

 The proof is by induction on the size of $K$. If 
 $K = \square$ then we calculate:
 \[\begin{array}{lcl}\ssem{\op(v)}(\rho)(\gamma)  
 & = & 
{\varphi}^{\R \times\ssem{\inn}}_{\ell,\op,\epsilon(\ell)}(\vsem{v}(\rho), (\eta_{W_\epsilon})_{\ssem{\inn}})\\
 & = & 
{\varphi}^{\R \times\ssem{\inn}}_{\ell,\op,\epsilon(\ell)}(\vsem{v}(\rho), \lambda a \in \ssem{\inn}. (0,a))\\
& = & 
{\varphi}^{\R \times\ssem{\inn}}_{\ell,\op,\epsilon(\ell)}(\vsem{v}(\rho), \lambda a \in \ssem{\inn}. \ssem{x}(\rho[a/x])(\gamma))\end{array}
 \]

 Otherwise the proof splits into cases:
     \begin{enumerate}
     \item Suppose that $K$ has the form $F[K_1]$ with $K_1[\op(v)]\type \tau\etype\epsilon$

     Then, using the algebraicity of the operations at the level of the selection monad and Lemma~\ref{lem:Fsem}, we calculate:
     \[\begin{array}{lcll}
       \ssem{F[K_1[\op(v)]]}(\rho)
       & = & 
       \mlet{S_\epsilon}{a \in \ssem{\tau}}{\ssem{K_1[\op(v)]}(\rho)}{\ssem{F[x]}
                                                 (\rho[a/x])}\\
       & = & 
       \mlet{S_\epsilon}{a \in \ssem{\tau}}
            {\\&&\; \varphi^{\R \times \ssem{\tau}}_{\ell,\op,\epsilon(\ell)}
 (\vsem{v},\lambda b \in \ssem{\inn}.\,
     \ssem{K_1[y]}(\rho[b/y]))}{\ssem{F[x]}(\rho[a/x])}\\
     & = &
     \varphi^{\R \times \ssem{\sigma}}_{\ell,\op,\epsilon(\ell)}
 (\vsem{v},\\&&\hspace{42pt}\lambda b \in \ssem{\inn}.\,
    \mlet{S_\epsilon}{a \in \ssem{\tau}}{ \ssem{K_1[y]}(\rho[b/y])}{\ssem{F[x]}(\rho[a/x])})\\
    
     & = &
     \varphi^{\R \times \ssem{\sigma}}_{\ell,\op,\epsilon(\ell)}
 (\vsem{v},\lambda b \in \ssem{\inn}.\,
     \ssem{F[K_1][y]}(\rho[b/y]))\\
          
     \end{array}\]
  \item Suppose that $K$ has the form 
         $\wph{h}{v}{K_1}$.
     where $h$ handles $\ell'$. By assumption $h$ does not handle $\op$, so $\ell' \neq \ell$.    Setting 
        \[T[\gamma] = \lambda a \in \ssem{\sigma}.\, 
\ER{\epsilon}{\ssem{e_{r}}(\rho[(\vsem{v}(\rho),a)/z])}{\gamma}\]
where the return clause of the handler is $\return \mapsto e_r$, we calculate:
     \[\begin{array}{lcl}
        \ssem{\wph{h}{v}{K_1[\op(w)]}}(\rho)(\gamma) 
        & = & 
        \ssem{h}(\rho)(\vsem{v}(\rho),\ssem{K_1[\op(w)]}(\rho))(\gamma)\\
        & = &
        s^{\dagger_{F_{\epsilon\ell}}}(\ssem{K_1[\op(w)]}(\rho)
        \\&&  \qquad (\lambda a \in \ssem{\sigma}.\, \ER{\epsilon}{\ssem{e_r}(\rho[(\vsem{v}(\rho),a)/z])}{\gamma}))
        \\&& \qquad (\vsem{v}(\rho))\\
        & = &
        s^{\dagger_{F_{\epsilon\ell}}}
         (\ssem{K_1[\op(w)]}(\rho)
        (T(\gamma)))(\vsem{v}(\rho))\\
        & = &
        s^{\dagger_{F_{\epsilon\ell}}}
         (
 \varphi^{\R \times \ssem{\sigma}}_{\ell,\op,\epsilon\ell'(\ell)}
 (\vsem{w}(\rho),\lambda a \in \ssem{\inn}.\,
     \ssem{K_1[x]}(\rho[a/x])(T(\gamma)))
         \\&& \qquad(\vsem{v}(\rho))\\
     & = &
     \psi_{\ell,\op,\epsilon(\op)}
        (\vsem{w}(\rho),\\&&
        \qquad \lambda a \in \ssem{\inn}.\,
            s^{\dagger_{F_{\epsilon\ell}}}(
                  \ssem{K_1[x]}(\rho[a/x])T(\gamma))) (\vsem{v}(\rho)) \\
    & = &
      \varphi_{\ell,\op,\epsilon(\op)}
        (\vsem{w}(\rho),\\&&
        \qquad \lambda a \in \ssem{\inn}.\,
            s^{\dagger_{F_{\epsilon\ell}}}(
                  \ssem{K_1[x]}(\rho[a/x])T(\gamma)) (\vsem{v}(\rho))) \\
     & = &
      \varphi_{\ell,\op,\epsilon(\op)}
        (\vsem{w}(\rho),\\&&
        \qquad \lambda a \in \ssem{\inn}.\,
            \ssem{\wph{h}{v}{K_1[x]}}(\rho[a/x])\gamma) \\
        \end{array}\]

   \item 
   
   Suppose that $K$ has the form $\glocal{\epsilon}{K_1}{\myg}$.
     We calculate:
     
     \[\begin{array}{lcl}
        \ssem{\glocal{\epsilon_1}{K_1[\op(v)]}{\myg}}(\rho)\gamma 
        & = &
        \ssem{K_1[\op(v)]}(\rho)\lsem{\myg}(\rho)\\
        & = &
        \varphi^{\R \times\ssem{\sigma}}_{\ell,\op,\epsilon(\ell)}
 (\vsem{v}(\rho),\lambda a \in \ssem{\inn}.\,
     \ssem{K_1[x]}(\rho[a/x])\lsem{\myg}(\rho))\\
     & = &
        \varphi^{\R \times\ssem{\sigma}}_{\ell,\op,\epsilon(\ell)}
 (\vsem{v}(\rho),\lambda a \in \ssem{\inn}.\,
     \ssem{\glocal{\epsilon_1}{K_1[x]}{\myg}}(\rho[a/x])\\
\end{array}\]

     \item

     Suppose that $K$ has the form 
     $\then{}{K_1}{\lambda^{\epsilon_1} x\type \tau.\, e_1}$.
      Then we calculate:
     \[\begin{array}{lcl}
       \ssem{\then{}{K_1[\op(v)]}{\lambda^{\epsilon_1} x\type \tau.\, e_1}}(\rho)(\gamma)
       & = & 
    \mlet{F_\epsilon}{r_1 \in \R, a \in \ssem{\tau}}{\ssem{K_1[\op(v)]}(\rho)(\lsem{\lambda^\epsilon x\type \tau.\, e_1}(\rho))}
{\\&&\;\quad \mlet{F_\epsilon}
                {r_2,r_3\in \R}
                {\ssem{e_1}(\rho[a/x])(\lambda r \in \R.\,0)}
                {(r_2,r_1 + r_3})}   \\
& = & 
 \mlet{F_\epsilon}{r_1 \in \R, a \in \ssem{\tau}}
      {\\&& \quad
         \varphi^{\R \times \ssem{\sigma}}_{\ell,\op,\epsilon(\ell)}
 (\vsem{v}(\rho),\lambda a \in \ssem{\inn}.\,
     \ssem{K_1[x]}(\rho[a/x])(\lsem{\lambda^\epsilon x\type \tau.\, e_1}(\rho)))
      \\&&}
{\quad \mlet{F_\epsilon}
                {r_2,r_3\in \R}
                {\ssem{e_1}(\rho[a/x])(\lambda r \in \R.\,0)}
                {(r_2,r_1 + r_3})}   \\
 & = &
 \varphi^{\R \times \ssem{\sigma}}_{\ell,\op,\epsilon(\ell)}
 (\vsem{v}(\rho),\lambda a \in \ssem{\inn}.\,
   \\&& \mlet{F_\epsilon}{r_1 \in \R, a \in \ssem{\tau}}
      {\ssem{K_1[x]}(\rho[a/x])(\lsem{\lambda^\epsilon x\type \tau.\, e_1}(\rho))}
{\\&&\quad \mlet{F_\epsilon}
                {r_2,r_3\in \R}
                {\ssem{e_1}(\rho[a/x])(\lambda r \in \R.\,0)}
                {(r_2,r_1 + r_3})} 
  )\\
& = &
 \varphi^{\R \times \ssem{\sigma}}_{\ell,\op,\epsilon(\ell)}
 (\vsem{v}(\rho),\lambda a \in \ssem{\inn}.\,
    \ssem{\then{}{K_1[x]}{\lambda^{\epsilon_1} x\type \tau.\, e_1}}(\rho[a/x])(\gamma)
  )\\

 \end{array}\]

\item  Suppose that $K$ has the form $\reset K_1$. Then we calculate:
 \[\begin{array}{lcl}
\ssem{\reset K_1[\op(v)]}(\rho)(\gamma) 
& = &
\mlet{F_\epsilon}
            {r_1 \in \R, a \in \ssem{\sigma}}
            {\ssem{K_1[\op(v)]}(\rho)(\gamma)}
            {(0,a)}\\
& = &
\mlet{F_\epsilon}
            {r_1 \in \R, a \in \ssem{\sigma}\\&&}
            {\varphi^{\R \times \ssem{\sigma}}_{\ell,\op,\epsilon(\ell)}
             (\vsem{v}(\rho),
              \lambda a \in \ssem{\inn}.\,\ssem{K_1[x]}(\rho[a/x])(\gamma))\\&&}
            {(0,a)}\\
& = &
\varphi^{\R \times \ssem{\sigma}}_{\ell,\op,\epsilon(\ell)}
             (\vsem{v}(\rho),
              \lambda a \in \ssem{\inn}.\,\\&&\quad
                  \mlet{F_\epsilon}
            {r_1 \in \R, a \in \ssem{\sigma}}
            {\ssem{K_1[x]}(\rho[a/x])(\gamma)}
            {(0,a)})\\
& = &
\varphi^{\R \times \ssem{\sigma}}_{\ell,\op,\epsilon(\ell)}
             (\vsem{v}(\rho),
              \lambda a \in \ssem{\inn}.\,
                  \ssem{\reset K_1[x]}(\rho[a/x])(\gamma)\\
\end{array}\]

 \end{enumerate}
 \end{proof}
For our soundness theorem we assume that the semantics of basic functions is sound w.r.t.\, the operational semantics, i.e.
$f(v) \rightarrow v' \implies \sem{f}(\sem{v}) = \sem{v'}$. 
The main result we need is that the small step operational semantics is sound:

\begin{theorem}[Small-step Soundness] \label{Athm:smallsound} Suppose we have an expression $e\type \sigma\etype \epsilon$ and a loss continuation $\myg \type \sigma \to \real\etype \epsilon_1$ with $\epsilon_1 \subseteq \epsilon$. Then:
\[\myg \vdash_\epsilon e \xrightarrow{r} e' \;\; \implies \;\; 
    \ssem{e}\lsem{\myg} = r\cdot(\ssem{e'}\lsem{\myg})\]
\end{theorem}
\begin{proof} we split into cases according to the various possible transitions.
We use  Lemma~\ref{Alem:valsem} without comment.
\begin{enumerate}
    \item  %
    Suppose the transition is 
    \[f(v) \evalto{0}  v'\] 
    for $f\type \sigma \to \tau$, where $f(v) \to v'$.
    Then we calculate:
    \[\begin {array}{lcll}
    \ssem{f(v)}
        & = & 
       \mlet{{S_\epsilon}}{a \in \ssem{\sigma}}{\ssem{v}}
     {\eta_{S_\epsilon}(\sem{f}(a))} \\
       & = &
       \mlet{{S_\epsilon}}{a \in \ssem{\sigma}}{\eta_{S_\epsilon}(\vsem{v})}
     {\eta_{S_\epsilon}(\sem{f}(a))} & (\mbox{by Lemma~\ref{Alem:valsem}})\\
      & = & \eta_{S_\epsilon}(\sem{f}(\vsem{v}))\\
      & = & \eta_{S_\epsilon}(\vsem{v'})
      
    \end{array}\]
    and this  gives the required result as it is equivalent to
    \[\ssem{f(v)}\gamma = 0\cdot \ssem{v'}\gamma\]
    for all relevant $\gamma$.
    \item  %
    \[\myg \vdash_\epsilon  (v_1,\dots,v_n).i  \evalto{0}  v_i \] 
    Then we calculate:
    
    \[\begin {array}{lcll}
    \ssem{(v_1,\dots,v_n).i} & = & S_\epsilon(\pi_i)(\ssem{(v_1,\dots,v_n)})\\
                             & = & S_\epsilon(\pi_i)(\eta_{S_\epsilon}((\vsem{v_n},\dots,\vsem{v_n})))& \mbox{(by Lemma~\ref{Alem:valsem})}\\
                             & = & \eta_{S_\epsilon}(\vsem{v_1})\\
                             & = & \ssem{v_i}
    \end{array}\]
    and this  gives the required result as it is equivalent to
    \[\ssem{(v_1,\dots,v_n).i}\gamma = 0\cdot \ssem{v_i}\gamma\]
    for all relevant $\gamma$.
    \item %
    \[\myg \vdash_\epsilon   \mycases{\inl{\sigma_1}{\sigma_2}(v)}{x_1\type \sigma_1}{e_1}{x_2\type \sigma_2}{e_2}  \evalto{0}  e_1[v/x_1]\]

    Again using Lemma~\ref{Alem:valsem}, we calculate:
    \[\begin{array}{lcll}
        \ssem{ \mycases{\inl{\sigma_1}{\sigma_2}(v)}{x_1\type \sigma_1}{e_1}{x_2\type \sigma_2}{e_2}} 
        & = &
        \mlet{S_\epsilon}{a \in \ssem{\sigma_1} + \ssem{\sigma_2}}{\ssem{\inl{\sigma_1}{\sigma_2}(v)}}{\\&&\;[\lambda b_1 \in \ssem{\sigma_1}.\, \ssem{e_1}(x_1 \mapsto b_1),
\\&&\;\;\,\lambda b_2 \in \ssem{\sigma_2}.\, \ssem{e_2}(x_2 \mapsto b_2)](a)}\\
         & = &
    \mlet{S_\epsilon}{a \in \ssem{\sigma_1} + \ssem{\sigma_2}}{\eta_{S_\epsilon}((0,\vsem{v}))}{\\&&\;[\lambda b_1 \in \ssem{\sigma_1}.\, \ssem{e_1}(x_1 \mapsto b_1),
\\&&\;\;\,\lambda b_2 \in \ssem{\sigma_2}.\, \ssem{e_2}(x_2 \mapsto b_2)](a)}\\             & = &
  [\lambda b_1 \in \ssem{\sigma_1}.\, \ssem{e_1}(x_1 \mapsto b_1),
\\&&\;\;\,\lambda b_2 \in \ssem{\sigma_2}.\, \ssem{e_2}(x_2 \mapsto b_2)]((0,\vsem{v}))\\
     & = &
  \ssem{e_1}(x_1 \mapsto \vsem{v})\\
    & = & \ssem{e_1[v_1/x_1]}\qquad (\mbox{by Lemma~\ref{lem:sublem}})
\end{array}\]
The case $\inr{\sigma_1}{\sigma_2}(v)$ is similar.
    \item %
    \[ \myg \vdash_\epsilon   \iter{0}{v_2}{v_3}  \evalto{0}  v_2\]

We have 
\[\begin{array}{lcl}
   \ssem{\iter{0}{v_2}{v_3}}  & = & 
       \mlet{{S_\epsilon}}{n \in \N}{\ssem{0}}
     {\\&&  \mlet{{S_\epsilon}}{a \in \ssem{\sigma}}{\ssem{v_2}}{\\&&\mlet{S_\epsilon}{\varphi \in \ssem{\sigma} \to S_\epsilon(\ssem{\sigma})}{\ssem{v_3}}
     {\\&& (\varphi^{\dagger_{S_\epsilon}})^n(\eta_{S_\epsilon}(a))}}}\\
     & = &
      \mlet{{S_\epsilon}}{n \in \N}{\eta_{S_\epsilon}(\vsem{0})}
     {\\&&  \mlet{{S_\epsilon}}{a \in \ssem{\sigma}}{\eta_{S_\epsilon}(\vsem{v_2})}{\\&&\mlet{S_\epsilon}{\varphi \in \ssem{\sigma} \to S_\epsilon(\ssem{\sigma})}{\eta_{S_\epsilon}(\vsem{v_3})}
     {\\&& (\varphi^{\dagger_{S_\epsilon}})^n(\eta_{S_\epsilon}(a))}}}\\
& = & (\vsem{v_3}^{\dagger_{S_\epsilon}})^0(\eta_{S_\epsilon}(\vsem{v_2}))
     \\
& = &         \eta_{S_\epsilon}(\vsem{v_2})\\
& = &         \ssem{v_2}
\end{array}\]
    
    \item %
    \[ \myg \vdash_\epsilon   \iter{\mysucc{v_1}}{v_2}{v_3}  \evalto{0} 
v_3 \iter{v_1}{v_2}{v_3}\]
    
   We calculate:
\[\begin{array}{lcll}
   \ssem{\iter{\mysucc{v_1}}{v_2}{v_3}}  & = & 
       \mlet{{S_\epsilon}}{n \in \N}{\ssem{\mysucc{v_1}}}
     {\\&&  \mlet{{S_\epsilon}}{a \in \ssem{\sigma}}{\ssem{v_2}}{\\&&\mlet{S_\epsilon}{\varphi \in \ssem{\sigma} \to S_\epsilon(\ssem{\sigma})}{\ssem{v_3}}
     {\\&& (\varphi^{\dagger_{S_\epsilon}})^n(\eta_{S_\epsilon}(a))}}}\\
     & = &
      \mlet{{S_\epsilon}}{n \in \N}{\eta_{S_\epsilon}(\vsem{v_1} + 1)}
     {\\&&  \mlet{{S_\epsilon}}{a \in \ssem{\sigma}}{\eta_{S_\epsilon}(\vsem{v_2})}{\\&&\mlet{S_\epsilon}{\varphi \in \ssem{\sigma} \to S_\epsilon(\ssem{\sigma})}{\eta_{S_\epsilon}(\vsem{v_3})}
     {\\&& (\varphi^{\dagger_{S_\epsilon}})^n(\eta_{S_\epsilon}(a))}}}\\
& = & (\vsem{v_3}^{\dagger_{S_\epsilon}})^{\vsem{v_1} + 1}(\eta_{S_\epsilon}(\vsem{v_2}))
     \\
& = &     \vsem{v_3}^{\dagger_{S_\epsilon}}(\ssem{\iter{v_1}{v_2}{v_3}})   \\
& = &         \ssem{v_3~\iter{v_1}{v_2}{v_3}} & (\mbox{by Lemma~\ref{lem:smalllem}})
\end{array}\]
    
    \item %
    \[\myg \vdash_\epsilon  \fold{\nil_{\sigma_1}}{v_2}{v_3}    \evalto{0}    v_2\]
    
    Then we have 
\[\begin{array}{lcl}
     \ssem{\fold{\nil_{\sigma_1}}{v_2}{v_3}} 
     & = &
\mlet{{S_\epsilon}}
     {l \in \ssem{\sigma_1}^*}{\ssem{\nil_{\sigma_1}}}
     {\\&&  \mlet{{S_\epsilon}}
            {a \in \ssem{\sigma}}
            {\ssem{v_2}}
              {\\&&\mlet{S_\epsilon}
{\varphi \in \ssem{\sigma_1}\times \ssem{\sigma} \to S_\epsilon(\ssem{\sigma})}}
{\ssem{v_3}}
     {\\&&\myfold{l}{ \eta_{S_\epsilon}(a)}
    {\lambda a_1\in \ssem{\sigma_1}, b \in S_\epsilon(\ssem{\sigma}).\,\varphi(a_1,-)^{\dagger_{S_\epsilon}}(b)}}}\\
    & = &
\mlet{{S_\epsilon}}
     {l \in \ssem{\sigma_1}^*}{\eta_{S_\epsilon}(\varepsilon)}
     {\\&&  \mlet{{S_\epsilon}}
            {a \in \ssem{\sigma}}
            {\eta_{\epsilon}(\vsem{v_2})}
              {\\&&\mlet{S_\epsilon}
{\varphi \in \ssem{\sigma_1}\times \ssem{\sigma} \to S_\epsilon(\ssem{\sigma})}}
{\eta_{\epsilon}(\vsem{v_3})}
     {\\&&\myfold{l}{ \eta_{S_\epsilon}(a)}
    {\lambda a_1\in \ssem{\sigma_1}, b \in S_\epsilon(\ssem{\sigma}).\,\varphi(a_1,-)^{\dagger_{S_\epsilon}}(b)}}}\\
& = &
\myfold{\epsilon}{ \eta_{S_\epsilon}(\vsem{v_2})}
    {\lambda a_1\in \ssem{\sigma_1}, b \in S_\epsilon(\ssem{\sigma}).\,\vsem{v_3}(a_1,-)^{\dagger_{S_\epsilon}}(b)}\\
& = & \eta_{S_\epsilon}(\vsem{v_2})\\
& = & \ssem{v_2}
\end{array}\]

    \item %
    \[\myg \vdash_\epsilon  \fold{\cons{v_1}{v_2}}{v_3}{v_4}  \evalto{0}   
    v_4(v_1,\fold{v_2}{v_3}{v_4})\]
    
    Then, on the one hand, we have:
\[\begin{array}{lcl}
     \ssem{\fold{\cons{v_1}{v_2}}{v_3}{v_4}} & = &
\mlet{{S_\epsilon}}
     {l \in \ssem{\sigma_1}^*}{\ssem{\cons{v_1}{v_2}}}
     {\\&&  \mlet{{S_\epsilon}}
            {a \in \ssem{\sigma}}
            {\ssem{v_3}}
              {\\&&\mlet{S_\epsilon}
{\varphi \in \ssem{\sigma_1}\times \ssem{\sigma} \to S_\epsilon(\ssem{\sigma})}}
{\ssem{v_4}}
     {\\&&\myfold{l}{ \eta_{S_\epsilon}(a)}
    {\lambda a_1\in \ssem{\sigma_1}, b \in S_\epsilon(\ssem{\sigma}).\,(\varphi(a_1,-))^{\dagger_{S_\epsilon}}(b)}}}\\
& = &
\mlet{{S_\epsilon}}
     {l \in \ssem{\sigma_1}^*}{\eta_{S_\epsilon}(\vsem{v_1}\vsem{v_2})}
     {\\&&  \mlet{{S_\epsilon}}
            {a \in \ssem{\sigma}}
            {\eta_{S_\epsilon}(\vsem{v_3})}
              {\\&&\mlet{S_\epsilon}
{\varphi \in \ssem{\sigma_1}\times \ssem{\sigma} \to S_\epsilon(\ssem{\sigma})}}
{\eta_{S_\epsilon}(\vsem{v_4})}
     {\\&&\myfold{l}{ \eta_{S_\epsilon}(a)}
    {\lambda a_1\in \ssem{\sigma_1}, b \in S_\epsilon(\ssem{\sigma}).\,(\varphi(a_1,-))^{\dagger_{S_\epsilon}}(b)}}}\\
& = &
\myfold{\vsem{v_1}\vsem{v_2}}{\eta_{S_\epsilon}(\vsem{v_3})}
    {\\&& \; \lambda a_1\in \ssem{\sigma_1}, b \in S_\epsilon(\ssem{\sigma}).\,(\vsem{v_4}(a_1,-))^{\dagger_{S_\epsilon}}(b)}\\
   & = & (\vsem{v_4}(\vsem{v_1},-))^{\dagger_{S_\epsilon}}\\
            && \quad(\myfold{\vsem{v_2}}{\eta_{S_\epsilon}(\vsem{v_3})}
    {\\&& \qquad \lambda a_1\in \ssem{\sigma_1}, b \in S_\epsilon(\ssem{\sigma}).\,(\vsem{v_4}(a_1,-))^{\dagger_{S_\epsilon}}(b)})\\
     & = & (\vsem{v_4}(\vsem{v_1},-))^{\dagger_{S_\epsilon}}\\
            && \quad(\ssem{\fold{\cons{v_1}{v_2}}{v_3}{v_4}})
    \end{array}\]
and, on the other hand, using Lemma~\ref{lem:smalllem} we have:
\[\begin{array}{lcll}
    \ssem{v_4~(v_1,\fold{v_2}{v_3}{v_4})} & = & 
     \mlet{{S_\epsilon}}
          {a \in \ssem{(\sigma_1,\sigma)}}
          {\ssem{(v_1,\fold{v_2}{v_3}{v_4})}}
          {\vsem{v_4}(a)}   \\
     & = &
     \mlet{{S_\epsilon}}
          {a \in \ssem{(\sigma_1,\sigma)}}
          {\\&& \quad (\mlet{{S_\epsilon}}{b\in \ssem{\sigma}}{\ssem{\fold{v_2}{v_3}{v_4}}}{\eta_{S_\epsilon}(\vsem{v_1},b)})\\&&}
          {\vsem{v_4}(a)}   \\
          & = & 
           \mlet{{S_\epsilon}}
                {b\in\ssem{\sigma}}
                {\ssem{\fold{v_2}{v_3}{v_4}}\\&&}
                {\mlet{{S_\epsilon}}
                      {a \in \ssem{(\sigma_1,\sigma)}}
                      {\eta_{S_\epsilon}(\vsem{v_1},b)\\&&\quad}
                      {\vsem{v_4}(a)}}\\
        & = &
             \mlet{{S_\epsilon}}
                {b\in\ssem{\sigma}}
                {\ssem{\fold{v_2}{v_3}{v_4}}\\&&}
                {\vsem{v_4}((\vsem{v_1},b))}\\   
          
\end{array}\]

\item %
    \[\myg \vdash_\epsilon (\lambda^\epsilon x\type \sigma.\, e_1)~v   \evalto{0}  e_1[v/x] \]
    
    Then we have 
\[\begin{array}{lcll}
   \ssem{(\lambda^\epsilon x\type \sigma.\, e_1)~v}  
     & = & \mlet{{S_\epsilon}}
        {a \in \ssem{\sigma}}{\ssem{v}}{\vsem{\lambda^\epsilon x\type \sigma.\, e_1}(a)}& 
        (\mbox{by Lemma~\ref{lem:smalllem}})\\
     & = &  \vsem{\lambda^\epsilon x\type \sigma.\, e_1}(\vsem{v})\\
     & = & \lambda a \in \ssem{\sigma}.\, \ssem{e_1}(x \mapsto a)\vsem{v}\\
     & = & \ssem{e_1}(x \mapsto \vsem{v})\\
     & = & \ssem{e_1[v/x]} & (\mbox{by Lemma~\ref{lem:sublem}})
\end{array}\]

\item %
    \[\myg \vdash_\epsilon  \loss(r)   \evalto{r}   ()\]
    
    Then for any $\gamma\type \Onebb \to R_\epsilon$ we have 

\[\begin{array}{lcl}
    \ssem{\loss(r)}(\gamma) & = & (\mlet{{S_\epsilon}}{a \in \R}{\ssem{r}}
    {\lambda \gamma \in \mathbbm{1} \to R_\epsilon.\, (a,())})(\gamma)\\
     & = & (\lambda \gamma \in \mathbbm{1} \to R_\epsilon.\, (r,()))(\gamma)\\
     & = & (r,())\\
     & = & r\cdot (0,())\\
     & = & r\cdot (\ssem{()}\gamma)\\
\end{array}\]

\item
Suppose the transition is
\[\myg \vdash_\epsilon\wph{h}{v_1}{K[\op(v_2)]}  \evalto{0} v_o(v_1,v_2,f_l,f_k)\]
where 
  $v_1:par$, 
  $\op \notin \hop(K)$, 
  $\op\type \outt \xrightarrow{\ell} \inn$, 
  $\op \mapsto v_o \in h$,
 and
 \[f_k = \lambda^\epsilon(p,y)\type (par,in).\,\glocal{\epsilon}{\wph{h}{p}{K[y]}}{\myg}\]
 and 
 \[f_l = \lambda^\epsilon (p,y)\type (par,in).\,\then{\epsilon}{(\wph{h}{p}{K[y]})}{\myg}\]
 We have 
 \[v_0 = \lambda^\epsilon z\type (\parr, \outt_1,
                     (\parr,\inn) \to \real \etype \epsilon,
                     (\parr,\inn) \to \sigma' \etype \epsilon).~ e_o\]
and
\[\return \mapsto \lambda^\epsilon z\type (\parr,  \sigma).\, e_r \in h\]
for some $e_o$ and $e_r$.

Setting 
$T[\myg] = \lambda a \in \ssem{\sigma_1}.\, 
\ER{\epsilon}{\ssem{e_{r}}[(\vsem{v_1},a)/z]}{\lsem{\myg}}
$
we calculate, using Lemma~\ref{Alem:Kop}:
\[\begin{array}{lcll}
  \ssem{\wph{h}{v_1}{K[\op(v_2)]}}\lsem{\myg} \\\\
   =   
  \quad \mlet{{S_\epsilon}}{a \in \ssem{\parr}}{\ssem{v_1}}
   {\\ \quad\ssem{h}(a,\ssem{K[\op(v_2)})}\lsem{\myg}\\\\
   =    
   \quad \ssem{h}(\vsem{v_1},\ssem{K[\op(v_2)})\lsem{\myg} \\\\
   = 
\quad   s^{\dagger_{F_{\epsilon\ell}}}((\ssem{K[\op(v_2)]})(\lambda a \in \ssem{\sigma_1}.\, 
\ER{\epsilon}{\ssem{e_{r}}[(\vsem{v_1},a)/z]}{\lsem{\myg}}))\vsem{v_1}\\\\
   = 
\quad   s^{\dagger_{F_{\epsilon\ell}}}((\ssem{K[\op(v_2)]})(T(\myg)))\vsem{v_1}\\\\
=
\quad s^{\dagger_{F_{\epsilon\ell}}}
 (\varphi^{\R \times \ssem{\sigma_1}}_{\ell,\op,i}
  (\vsem{v_2},
   \lambda a \in \ssem{\inn}.\, \ssem{K[y]}(y \mapsto a)T(\myg)))\vsem{v_1} \\
     (\mbox{where $i = 1 + \epsilon(\op)$})\\\\
= \quad \psi_{\ell,\op,i}
        (\vsem{v_2},
        \lambda a \in \ssem{\inn}.\,
            s^{\dagger_{F_{\epsilon\ell}}}(
                  \ssem{K[y]}(y \mapsto a)T(\myg))) \vsem{v_1} \\\\
= \quad \ssem{e_o}(z \mapsto (\vsem{v_1},\vsem{v_2},l_1,k_1))\lsem{\myg}\\
\end{array}\]

 where
 \[k_1 = \lambda p\in \ssem{\parr}, a \in \ssem{\inn_j}.\, \lambda \gamma_1 \in  \R^{\ssem{\sigma'}}.\,s^{\dagger_{F_{\epsilon\ell}}}(
                  \ssem{K[y]}(y \mapsto a)T(\myg))p\]
 and
 \[l_1 = \lambda p\in \ssem{\parr}, a \in \ssem{\inn_j}.\, \lambda \gamma_1 \in  \R^{\ssem{\sigma'}}.\, \delta_\epsilon(\lsem{\myg}^{\dagger_{W_\epsilon}}(s^{\dagger_{F_{\epsilon\ell}}}(
                  \ssem{K[y]}(y \mapsto a)T(\myg))p))\]

 and we calculate:
\[\begin{array}{lcl}
\ssem{f_k}(p,a)\gamma_1
&= &
\ssem{\lambda^\epsilon(x,y)\type (par,in).\,\glocal{\epsilon}{\wph{h}{x}{K[y]}}{\myg}}(p,a)\gamma_1\\
&=& 
\ssem{\glocal{\epsilon}{\wph{h}{x}{K[y]}}{\myg}}(x \mapsto p, y \mapsto a)\gamma_1\\
&= & 
\ssem{\wph{h}{p}{K[y]}}(y \mapsto a)\lsem{\myg}\\
& = &
\ssem{h}(p,\ssem{K[y]}(y \mapsto a))\lsem{\myg}\\
& = &
s^{\dagger_{F_{\epsilon\ell}}}(\ssem{K[y]}(y \mapsto a)(\lambda a \in \ssem{\sigma}.\, \ER{\epsilon}{\ssem{e_r}(\rho[(p,a)/z])}{\lsem{\myg}}))p\\
&=&
s^{\dagger_{F_{\epsilon\ell}}}(\ssem{K[y]}(y \mapsto a)
(T(\myg)))p\\
& = &
k_1(p,a)\gamma_1
\end{array}\]
and also:
\[\begin{array}{lcll}
   \ssem{f_l}(p,a)\gamma_1
   & = &
   \ssem{\then{\epsilon}{(\wph{h}{x}{K[y]})}{\myg}}(x \mapsto p, y \mapsto a)\gamma_1\\
   & = &
   \delta_\epsilon(\ER{}{\ssem{\wph{h}{x}{K[y]}}(x \mapsto p, y \mapsto a)}{\lsem{\myg}})\\ && \hspace{150pt}(\mbox{by Lemma~\ref{lem:theng}})\\
   &=& \delta_\epsilon(\lsem{\myg}^{\dagger_{W_\epsilon}}(\ssem{\wph{h}{x}{K[y]}}(x \mapsto p, y \mapsto a)\lsem{\myg}))\\
   
   & = &
   \delta_\epsilon(\lsem{\myg}^{\dagger_{W_\epsilon}}(s^{\dagger_{F_{\epsilon\ell}}}(\ssem{K[y]}(y \mapsto a)
(T(\myg)))p))\\
& = &
l_1(p,a)\gamma_1
\end{array}\]

\[\myg \vdash_\epsilon\wph{h}{v_1}{K[\op(v_2)]}  \evalto{0} v_o(v_1,v_2,f_l,f_k)\]
where 
  $v_1:par$, 
  $\op \notin \hop(K)$, 
  $\op\type \outt \xrightarrow{\ell} \inn$, 
  $\op \mapsto v_o \in h$,\\ 
  \[f_k = \lambda^\epsilon(p,y)\type (par,in).\,\wph{h}{p}{K[y]}\] 
  and \[f_l = \lambda^\epsilon (p,y)\type  
            (par,in).\,\then{\epsilon}{f_k(p,y)}{\myg}\]

we have, setting $T[\gamma] = \lsem{\lambda^\epsilon x\type \sigma_1.\,\then{\epsilon}{v_r(v,x)}{\gamma}}$:

\[s(r,a) = \lambda p\in \ssem{\parr}.\, r\cdot(\ssem{e}(\rho)[(p,a)/z])\] 
  \[\ssem{h}(p,G)(\gamma) = s^{\dagger_{F_{\epsilon\ell}}}(G(\lambda a \in \ssem{\sigma}.\, 
  \ER{\epsilon}{\ssem{e}[(p,a)/z]}{\gamma}))(p)(\gamma)\]

\[\begin{array}{lcl}
  \ssem{\wph{h}{v_1}{K[\op(v_2)]}}\lsem{\myg} \\
   =   
  \quad (\mlet{{S_\epsilon}}{a \in \ssem{\parr}}{\ssem{v_1}}
   {\\ \quad\ssem{h}(a,\ssem{K[\op(v_2)}})\lsem{\myg}\\
   =    
   \quad \ssem{h}(\vsem{v_1},\ssem{K[\op(v_2)})\lsem{\myg} \\
   = 
\quad   s^{\dagger_{F_{\epsilon\ell}}}((\ssem{K[\op(v_2)]})(\lambda a \in \ssem{\sigma}.\, 
\ER{\epsilon}{\ssem{e_{r}}(\rho)[(\vsem{v_1},a)/z]}{\lsem{\myg}}))\vsem{v_1} \lsem{\myg}\\
   = 
\quad   s^{\dagger_{F_{\epsilon\ell}}}((\ssem{K[\op(v_2)]})(T(\myg)))\vsem{v_1} \lsem{\myg}\\
=
\quad s^{\dagger_{F_{\epsilon\ell}}}(\varphi^{\R \times \ssem{\sigma_1}}_{\ell,\op,i}(\vsem{v},\lambda a \in \ssem{\inn}.\,
     \ssem{K[y]}(y \mapsto a)T(\myg)))\vsem{v_1} \lsem{\myg}\\
= \quad \psi_{\ell,\op_j,i}
        (\vsem{v},
        \lambda a \in \ssem{\inn}.\,
            s^{\dagger_{F_{\epsilon\ell}}}(
                  \ssem{K[y]}(y \mapsto a)T(\myg)))) \vsem{v_1} \lsem{\myg}\\
= \quad \ssem{e_j}(\rho[(\vsem{v_1},l_1,k_1)/z])
\end{array}\]

\[\psi_{\ell,\op_j,i}(o,k) = \lambda p\in \ssem{\parr}.\, \ssem{e_j}(\rho[(p,l_1,k_1)/z])
 \]
 where
 \[k_1 = \lambda p\in \ssem{\parr}, a \in \ssem{\inn_j}.\, kap\]
 and
 \[l_1 = \lambda p\in \ssem{\parr}, a \in \ssem{\inn_j}.\, \ER{\epsilon}{kap}{-}\]
 where
 
 \[\begin{array}{lcl}k = \lambda a \in \ssem{\inn}.\,
            s^{\dagger_{F_{\epsilon\ell}}}(
                  \ssem{K[y]}(y \mapsto a)T(\myg))
\end{array}\]

Then we have:
\[\begin{array}{lcl}
k_1(p,a) &=& k(a,p)\\ 
&=& s^{\dagger_{F_{\epsilon\ell}}}(
                  \ssem{K[y]}(y \mapsto a)T(\myg))p \\                  
&=& \lambda \gamma. s^{\dagger_{F_{\epsilon\ell}}}(
                  \ssem{K[y]}(y \mapsto a)T(\myg))p\gamma \\
                  
                  \end{array}\]
                  
\textcolor{red}{BUT}

\[\begin{array}{lcl}
\ssem{f_k}(p,a)
&= &
\ssem{\lambda^\epsilon
   x\type \parr,y \type \inn.\,\wph{h}{x}{K[y]}}(p,a)\\
& = & \ssem{\wph{h}{x}{K[y]}} (x \mapsto p, y \mapsto a)  \\
& = &
\lambda \gamma.\, s^{\dagger_{F_{\epsilon\ell}}}(\ssem{K[y]}y \mapsto a)(T(\gamma))p \gamma
\end{array}\]

\item %

Suppose the transition is:
    \[ \myg \vdash_\epsilon  \wph{h}{v_1}{v_2}  \evalto{0}  v_r(v_1,v_2) \\\]
    where $\return \mapsto v_r =  \lambda^\epsilon z\type (\parr,  \sigma).\, e_{r}$ is the return clause of $h$.
   
   For some $\parr$, $\sigma_1$, and $\ell$ we have 
   $h\type \sigma_1\etype \epsilon\ell \Rightarrow \sigma\etype\epsilon$,
   $v_1\type\parr$, $v_2\type \sigma_1$, and $z\type (\parr,  \sigma_1) \vdash e_{br}\type \sigma\etype\epsilon$.
   
    We have :

\[\begin{array}{lcl}
    \ssem{\wph{h}{v_1}{v_2}} \gamma
    & = & 
(\mlet{{S_\epsilon}}{a \in \ssem{\parr}}{\ssem{v_1}}{\ssem{h}(a,\ssem{v_2})}\gamma\\
  & = &
  \ssem{h}(\vsem{v_1},\ssem{v_2})\gamma\\
  & = & 
 s^{\dagger_{F_{\epsilon\ell}}}((\ssem{v_2})(\lambda a \in \ssem{\sigma}.\, \ER{\epsilon}{\ssem{e_{r}}[(\vsem{v_1},a)/z]}{\gamma}))\vsem{v_1} \\%
 & = &
 s^{\dagger_{F_{\epsilon\ell}}}((0,\vsem{v_2}))\vsem{v_1}\\%
 & = &
  s(0,\vsem{v_2})\vsem{v_1}\\%
 & = & (\lambda p\in \ssem{\parr}.\, 0\cdot(\ssem{e_{r}}(z \mapsto (p,\vsem{v_2})\gamma))\vsem{v_1}\\
 & = &
  \ssem{e_r}(z \mapsto (\vsem{v_1},\vsem{v_2}))\gamma\\
 & = &
  \ssem{e_r[(v_1,v_2)/z]}\gamma\\
 & = &
  \ssem{v_r~(v_1,v_2)}\gamma
  \end{array}\]

\item %
\[\myg \vdash_{\epsilon} \glocal{\epsilon_1}{v}{\myg}  \evalto{0}  v 
\]    
where $\epsilon_1 \subseteq \epsilon$. Then we have 
\[\ssem{\glocal{\epsilon_1}{v}{\myg}}\gamma = \ssem{v}\lsem{\myg} = \ssem{v}\gamma\]

\item

Suppose the transition is
\[\myg \vdash_{\epsilon} \then{\epsilon_1}{v}{\lambda^{\epsilon_1} x\type \sigma. e}
  \; \evalto{0} \; \glocal{\epsilon_1}{e[v/x]}{\lambda^{\epsilon_1} x: \sigma.\, 0} \]
where $\epsilon_1 \subseteq \epsilon$. In this case 
$\sigma = \real$.
  
  We have:

\[\begin{array}{lcl}
     \ssem{\then{\epsilon_1}{v}{\lambda^{\epsilon_1} x\type \sigma_1. \, e_2}}(\gamma)
&=&  \mlet{F_\epsilon}{r_1 \in \R, a \in \ssem{\sigma_1}}{\ssem{v}(\lsem{\lambda^{\epsilon_1} x\type \sigma_1. \, e_2})}
{\\&&\;\quad \mlet{F_\epsilon}
                {r_2,r_3\in \R}
                {\ssem{e_2}(x \mapsto a)(\lambda r \in \R.\,0)}
                {(r_2,r_1 + r_3})} \\
 &=& \mlet{F_\epsilon}{r_1 \in \R, a \in \ssem{\sigma_1}}{(0,\vsem{v})}
{\\&&\;\quad \mlet{F_\epsilon}
                {r_2,r_3\in \R}
                {\ssem{e_2}(x \mapsto a)(\lambda r \in \R.\,0)}
                {(r_2,r_1 + r_3})} \\
 &=&  \mlet{F_\epsilon}
                {r_2,r_3\in \R}
                {\ssem{e_2}(x \mapsto \vsem{v})(\lambda r \in \R.\,0)}
                {(r_2,r_3}) \\
 &=& 
   \ssem{e_2}(x \mapsto \vsem{v})(\lambda r \in \R.\,0)\\
 &=& 
   \ssem{e_2[v/x]}\lsem{\lambda^{\epsilon_1} x \type \real.\,0}\\
 &= &\ssem{\glocal{\epsilon_1}{e_2[v/x]}{\lambda^{\epsilon_1} x :  \real.\,0}}\gamma
   \end{array}\]
\item
Suppose the transition is 
\[\reset{v} \xrightarrow{0} v\]
Then we have:
\[\begin{array}{lcl}
\ssem{\reset{v}}\lsem{\myg}
& = &
\mlet{F_\epsilon}
            {r_1 \in \R, a \in \ssem{\sigma}}
            {\ssem{v}\lsem{\myg}}
            {(0,a)}\\
& = &
\mlet{F_\epsilon}
            {r_1 \in \R, a \in \ssem{\sigma}}
            {(0,\vsem{v})}
            {(0,a)}\\
& = & (0,\vsem{v})\\
& = & \ssem{v}\lsem{\myg}

\end{array}\]

\item 

Suppose the transition is generated by the rule
\[ {\inferrule{
                     \lambda^\epsilon x\type \tau.\, (\then{\epsilon}{F[x]}{\myg}) \vdash_\epsilon e_1 \evalto{r} e_1' }
{\myg \vdash_\epsilon F[e_1] \evalto{r} F[e_1']}}\]

By the induction hypothesis we have:
\[\ssem{e_1}\lsem{\lambda^\epsilon x\type \tau.\, (\then{\epsilon}{F[x]}{\myg})}
=
r\cdot \ssem{e'_1}\lsem{\lambda^\epsilon x\type \tau.\, (\then{\epsilon}{F[x]}{\myg})}\]

We then see:

\[\begin{array}{lcll}\ssem{F[e_1]}\lsem{\myg} 
  & = &
 \mlet{W_\epsilon}{a \in \ssem{\tau}\\&&\quad}
      {\ssem{e_1}\lsem{\lambda^\epsilon x\type \tau.\,\then{\epsilon}{F[x]}{\myg}}\\&&\quad}
      {\ssem{F[x]}(x\mapsto a)\lsem{\myg}}
      & (\mbox{by Lemma~\ref{lem:Fsemg}})\\
& = &
 \mlet{W_\epsilon}{a \in \ssem{\tau}\\&&\quad}
      {r\cdot \ssem{e_1}\lsem{\lambda^\epsilon x\type \tau.\,\then{\epsilon}{F[x]}{\myg}}\\&&\quad}
      {\ssem{F[x]}(x\mapsto a)\lsem{\myg}}\\
& = &
 r\cdot (\mlet{W_\epsilon}{a \in \ssem{\tau}\\&&\quad}
      {\ssem{e_1}\lsem{\lambda^\epsilon x\type \tau.\,\then{\epsilon}{F[x]}{\myg}}\\&&\quad}
      {\ssem{F[x]}(x\mapsto a)\lsem{\myg}})\\
      & = &
 r\cdot \ssem{F[e_1]}\lsem{\myg} & (\mbox{by Lemma~\ref{lem:Fsemg}})\\
      \end{array}\]
\item 

Suppose the transition is generated by the rule
\[{\inferrule*[narrower=0.7]{\return \mapsto v_r \in h \quad v_r\type (\parr,\sigma_1) \to \sigma\etype \epsilon \\
                     \lambda^\epsilon x\type \sigma_1.\, (\then{\epsilon}{v_r(v,x)}{\myg}) \vdash_{\epsilon\ell} e_2 \evalto{r} e_2' }
{\myg \vdash_\epsilon \wph{h}{v}{e_2} \evalto{r} \wph{h}{v}{e_2'}}}\]

We have $h$ handles $\ell$, for some $\ell$,
and $v_r = \lambda^{\epsilon}\, z\type (\parr,\sigma_1).\,e_r$
for some $e_r$.

By the induction hypothesis we have:
\[\ssem{e_2}\lsem{\lambda^\epsilon x\type \sigma_1.\,\then{\epsilon}{v_r(v,x)}{\myg}} = r\cdot \ssem{e'_2}\lsem{\lambda^\epsilon x\type \sigma_1.\,\then{\epsilon}{v_r(v,x)}{\myg}}\]

We have:
\[\begin{array}{lcll}
\lambda a \in \ssem{\sigma}.\, \ER{\epsilon}{\ssem{e_r}[(\vsem{v},a)/z]}{\lsem{\myg}}))
& = &
\lambda a \in \ssem{\sigma_1}.\, 
 \ER{\epsilon}{\ssem{v_r(v,x)}[x\mapsto a]}{\lsem{\myg}}))\\
& = &
\lsem{\lambda^\epsilon x\type \sigma_1.\,\then{\epsilon}{v_r(v,x)}{\myg}} \\
&& \qquad\qquad (\mbox{by Lemma~\ref{lem:theng}})
\end{array}\]

\[\lsem{\lambda^\epsilon x\type \sigma.\,\then{\epsilon}{e}{\myg}} = \lambda a \in \ssem{\sigma}.\, \ER{\epsilon}{\ssem{e}[x\mapsto a]}{\lsem{\myg}}))
 \]
 
Using this equality, for $e_2$ we have
\[\begin{array}{lcl}
    \ssem{\wph{h}{v}{e_2}}(\lsem{\myg}) 
       & = & (\mlet{{S_\epsilon}}{a \in \ssem{par}}{\ssem{v}}{\ssem{h}(a,\ssem{e_2})}(\lsem{\myg})\\
       & = & \ssem{h}(\vsem{v},\ssem{e_2})(\lsem{\myg})\\
       & = & s^{\dagger_{F_{\epsilon\ell}}}(\ssem{e_2}(\lambda a \in \ssem{\sigma}.\, \ER{\epsilon}{\ssem{e_r}[(\vsem{v},a)/z]}{\lsem{\myg}}))(\vsem{v})\\%
       & = &s^{\dagger_{F_{\epsilon\ell}}}(\ssem{e_2}(\lsem{\lambda^\epsilon x\type \sigma_1.\,\then{\epsilon}{v_r(v,x)}{\myg}}))(\vsem{v})\\%
\end{array}\]
and, similarly, for $e_2'$ we have:
\[\begin{array}{lcl}\ssem{\wph{h}{v}{e'_2}}(\lsem{\myg}) 
        =  s^{\dagger_{F_{\epsilon\ell}}}(\ssem{e'_2}(\lsem{\lambda^\epsilon x\type \sigma_1.\,\then{\epsilon}{v_r(v,x)}{\myg}}))(\vsem{v})\\%
        \end{array}\]

We conclude by linking these equalities using the induction hypothesis:

\[\begin{array}{lcl}
     s^{\dagger_{F_{\epsilon\ell}}}(\ssem{e_2}(\lsem{\lambda^\epsilon x\type \sigma_1.\,\then{\epsilon}{v_r(v,x)}{\myg}}))(\vsem{v})\\%
    \qquad = s^{\dagger_{F_{\epsilon\ell}}}(r\cdot \ssem{e'_2}(\lsem{\lambda^\epsilon x\type \sigma_1.\,\then{\epsilon}{v_r(v,x)}{\myg}}))(\vsem{v})\\%
    \qquad = (r\cdot -)^{\ssem{\parr}}(s^{\dagger_{F_{\epsilon\ell}}}(\ssem{e'_2}(\lsem{\lambda^\epsilon x\type \sigma_1.\,\then{\epsilon}{v_r(v,x)}{\myg}})))(\vsem{v})\\%
     \qquad = r\cdot (s^{\dagger_{F_{\epsilon\ell}}}(\ssem{e'_2}(\lsem{\lambda^\epsilon x\type \sigma_1.\,\then{\epsilon}{v_r(v,x)}{\myg}}))(\vsem{v}))\\%
\end{array}\]
with the second-last equality holding by Lemma~\ref{lem:actc}.

\item 
Suppose the transition is generated by the rule
\[{\inferrule{ \vdash e \type \sigma\etype \epsilon_1 
\\ \myg_1 \vdash_{\epsilon_1} e \evalto{r} e'}
{\myg \vdash_{\epsilon} \glocal{\epsilon_1}{e}{\myg_1}\evalto{r} \glocal{\epsilon_1}{e'}{\myg}}}\]

From the induction hypothesis, we have:
\[\ssem{e}\lsem{\myg_1}= r\cdot(\ssem{e'}\lsem{\myg_1})\]
Using this, we calculate:
\[\begin{array}{lcl}
\ssem{\glocal{\epsilon_1}{e}{\myg_1}}\lsem{\myg} 
& = & \ssem{e}\lsem{\myg_1}\\
& = & r\cdot(\ssem{e'}\lsem{\myg_1})\\
& = & r \cdot \ssem{\glocal{\epsilon_1}{e}{\myg_1}}\lsem{\myg} 
\end{array}\]

\item

Suppose the transition is generated by the rule
\[{\inferrule{\myg_1 \vdash_\epsilon e_1 \evalto{r_0} e_1'}
         {\myg \vdash_\epsilon  (\then{\epsilon}{e_1}{\myg_1}) \evalto{\; 0 \;} r_0 +  
            (\then{\epsilon_1}{e_1'}{\myg_1})}}\]

By the induction hypothesis we have:
\[\ssem{e_1}\lsem{\myg_1} = r_0\cdot\ssem{e'_1}\lsem{\myg_1}\]
Suppose that:
\[\myg_1 = \lambda^{\epsilon_1} x\type \sigma_1. \, e_2\]
with $\epsilon_1 \subseteq \epsilon$.

                We then have:
 \[\begin{array}{lcl}
      \ssem{\then{\epsilon_1}{e_1}
  {\lambda^\epsilon x\type \sigma_1.\, e_2}}\gamma 
&=&\mlet{F_\epsilon}{r_1 \in \R, a \in \ssem{\sigma_1}}{\ssem{e_1}(\lsem{\lambda^\epsilon x\type \sigma_1.\, e_2})}
{\\&&\;\quad \mlet{F_\epsilon}
                {r_2,r_3\in \R}
                {\ssem{e_2}(x \mapsto a)(\lambda r \in \R.\,0)}
                {(r_2,r_1 + r_3})}
\end{array}\]

\[\begin{array}{lcl}
\ssem{\then{\epsilon_1}{e_1}
        {\lambda^{\epsilon_1} x\type \sigma_1. \, e_2}} 
        \lsem{\myg}
&=& 
\mlet{F_\epsilon}{r_1 \in \R, a \in \ssem{\sigma_1}}{\ssem{e_1}(\lsem{\myg_1})}
{\\&&\;\quad \mlet{F_\epsilon}
                {r_2, r_3\in \R}
                {\ssem{e_2}(x \mapsto a)(\lambda r \in \R.\,0)}
                {(r_2,r_1 + r_3})}\\

&=& 
\mlet{F_\epsilon}{r_1 \in \R, a \in \ssem{\sigma_1}}{r_0\cdot \ssem{e'_1}(\lsem{\myg_1})}
{\\&&\;\quad \mlet{F_\epsilon}
                {r_2,r_3\in \R}
                {\ssem{e_2}(x \mapsto a)(\lambda r \in \R.\,0)}
                {(r_2, r_1 + r_3})}\\
&=& 
\mlet{F_\epsilon}{r_1 \in \R, a \in \ssem{\sigma_1}}{\ssem{e'_1}(\lsem{\myg_1})}
{\\&&\;\quad \mlet{F_\epsilon}
                {r_2,r_3\in \R}
                {\ssem{e_2}(x \mapsto a)(\lambda r \in \R.\,0)}
                {(r_2,r_0 + r_1 + r_3})}\\
&=&         
\ssem{r_0 + (\then{\epsilon_1}{e_1}
        {\lambda^{\epsilon_1} x\type \sigma_1. \, e_2})} 
        \lsem{\myg}
\end{array}\]

\item Suppose the transition is generated using the rule

\[\inferrule{\myg \vdash_\epsilon e_1 \evalto{r} e'_1}
         {\myg \vdash_\epsilon  \reset{e_1} \evalto{\; 0 \;} \reset{e'_1}}\] 

By the induction hypothesis we have:
\[\ssem{e_1}\lsem{\myg} = r\cdot \ssem{e'_1}\lsem{\myg}\]
So then:
\[\begin{array}{lcl}
\ssem{\reset{e_1}}\lsem{\myg}
 &=& \mlet{F_\epsilon}
            {r_1 \in \R, a \in \ssem{\sigma}}
            {\ssem{e_1}\lsem{\myg}}
            {(0,a)}\\
&=& \mlet{F_\epsilon}
            {r_1 \in \R, a \in \ssem{\sigma}}
            {r\cdot\ssem{e'_1}\lsem{\myg}}
            {(0,a)}\\
&=& \mlet{F_\epsilon}
            {r_1 \in \R, a \in \ssem{\sigma}}
            {\ssem{e'_1}\lsem{\myg}}
            {(0,a)}\\
& = & \ssem{\reset{e'_1}}\lsem{\myg}\\
& = & 0\cdot\ssem{\reset{e'_1}}\lsem{\myg}\\ 
\end{array}\]

\end{enumerate}

\end{proof}

Soundness for the big-step operational semantics  then follows:
\begin{theorem} [Evaluation Soundness] \label{Athm:sound}
For all expressions $e\type \sigma\etype \epsilon$ and loss continuations $\myg \type \sigma \to \real\etype \epsilon'$ with $\epsilon' \subseteq \epsilon$ we have:
\[\myg \vdash_\epsilon e \xRightarrow{r} v \;\; \implies \;\; 
\ssem{e}\lsem{\myg} = (r,\vsem{v}) \]
and
\[\begin{array}{l}
\myg \vdash_\epsilon e \xRightarrow{r} K[\op(v)] \;\; \implies \\\\ 
\hspace{30pt} \ssem{e}\lsem{\myg} = \varphi^{\R \times \ssem{\sigma}}_{\ell,\op,\epsilon(\ell)}
 (\vsem{v},\lambda a \in \ssem{\inn}.\,
   r\cdot  \ssem{K[x]}(x \mapsto a)\lsem{\myg}
     \end{array}\]
     \end{theorem}
\begin{proof}
Both statements are proved by induction on the proof of the evaluation relation. For the first one, in the base case we have $e = v$ and $r=0$ and, using Lemma~\ref{Alem:valsem} we have:
\[\ssem{e}\lsem{\myg} = \ssem{v}\lsem{\myg} = (0,\vsem{v}) \]
In the other case we have $\myg \vdash_\epsilon e \xrightarrow{r_1} e'$ and 
$\myg \vdash_\epsilon e' \xRightarrow{r_2} v$ and $r = r_1 + r_2$. So using Theorem~\ref{Athm:smallsound} we obtain:
\[\ssem{e}\lsem{\myg} = r_1\cdot (\ssem{e'}\lsem{\myg}) = r_1 \cdot (r_2,\vsem{v}) = (r,\vsem{v})\]

For the second one, in the base case we have $e = K[\op(v)]$ and $r=0$ and, using Lemma~\ref{Alem:Kop} we have:
\[\ssem{e}\lsem{\myg} = \ssem{K[\op(v)]}\lsem{\myg} =  
 \varphi^{\R \times \ssem{\sigma}}_{\ell,\op,\epsilon(\ell)}
 (\vsem{v},\lambda a \in \ssem{\inn}.\,
     \ssem{K[x]}(x \mapsto a)\lsem{\myg})\]
In the other case we have $\myg \vdash_\epsilon e \xrightarrow{r_1} e'$ and 
$\myg \vdash_\epsilon e' \xRightarrow{r_2} v$ and $r = r_1 + r_2$. So we obtain:
\[\begin{array}{lcl}
\ssem{e}\lsem{\myg} 
& = &
r_1\cdot (\ssem{e'}\lsem{\myg})\\
& = &
r_1\cdot (\varphi^{\R \times \ssem{\sigma}}_{\ell,\op,\epsilon(\ell)}
 (\vsem{v},\lambda a \in \ssem{\inn}.\,
     r_2\cdot\ssem{K[x]}(x \mapsto a)\lsem{\myg}))\\
& = &
(\varphi^{\R \times \ssem{\sigma}}_{\ell,\op,\epsilon(\ell)}
  \vsem{v},\lambda a \in \ssem{\inn}.\,
     r\cdot\ssem{K[x]}(x \mapsto a)\lsem{\myg})
\end{array}
\]

\end{proof}

\begin{theorem}[Adequacy] \label{Athm:adequate} For all expressions
$e\type \sigma\etype \epsilon$ and loss continuations $\myg \type \sigma \to \real\etype \epsilon'$ with $\epsilon' \subseteq \epsilon$ we have:
\[\ssem{e}\lsem{\myg} = (r,a) \implies 
    \exists v.\,  \myg \vdash_\epsilon e \xRightarrow{r} v \wedge \vsem{v} = a\]
and
\[\begin{array}{l}
\ssem{e}\lsem{\myg} = \varphi^{\R \times \ssem{\sigma}}_{\ell,\op,\epsilon(\ell)} (a,f)\;\; \implies \;\;   \exists K[\op(v)].\,  \myg \vdash_\epsilon e \xRightarrow{r} K[\op(v)]\; \wedge \\\\
\hspace{20pt} a = \vsem{v} 
 \;\wedge \; f = \lambda b \in \ssem{\inn}.\,
    r\cdot \ssem{K[x]}(x \mapsto b)\lsem{\myg}
 \end{array}\]
\end{theorem}
\begin{proof}
\myskip

Suppose that $\ssem{e}\lsem{\myg} = (r,a)$. By the termination theorem, Theorem~\ref{Athm:term}, either there are $r'$ and $v$ such that
$\myg \vdash_\epsilon e \xRightarrow{r'} v$ or else there are $r'$ and $K[\op(v)]$ such that
$\myg \vdash_\epsilon e \xRightarrow{r'} K[\op(v)]$. 
Using the second part of Theorem~\ref{Athm:sound} we see that the latter cannot happen, as then $\ssem{e}\lsem{\myg}$ would not ahve the form $(r,a)$.
So there are $r'$ and $v$ such that
$\myg \vdash_\epsilon e \xRightarrow{r'} v$.
Then, by evaluation soundness, we have $\ssem{e}\lsem{\myg} = (r',\vsem{v})$. So $(r,\vsem{v}) = (r',a)$ and the conclusion follows.

The second case is proved similarly.
\end{proof}

For the following corollary we assume the interpretation of constants is 1-1, I.e. that $\sem{c} = \sem{c'}$ implies $c = c'$. 
\begin{corollary}[First-order adequacy] 
Suppose we have a first-order type $\sigma$, an expression $e\type \sigma\etype \emptyset$, and a loss continuation $\myg \type \sigma \to \real\etype \emptyset$. Then for any $v\type \sigma$ we have:
\[\ssem{e}\lsem{\myg} = (r,\vsem{v}) \iff
    \myg \vdash_\epsilon e \xRightarrow{r} v\]
\end{corollary}
\begin{proof} This follows from the adequacy theorem, together with the fact that as the interpretation of constants is 1-1 so is that of first-order values, i.e., if $\vsem{v} = \vsem{v'}$, for $v,v'\type \sigma$ then $v = v'$.
\end{proof}

Turning to giant step operational semantics, fix $\sigma$ and $\epsilon$. Define the set $\EV$ of effect values by:
\[\EV\; = \; 
\left ( \sum_{\ell \in \epsilon, \op: \outt \xrightarrow{\ell} \inn} 
  V_\outt \times \EV^{V_\inn}  \right ) \,+ \,R \times V_\sigma \]
where $V_\tau = \{v|v\type\tau\}$.
Next, fix $\myg\type \sigma \to \loss\etype \epsilon'$ (with $\epsilon'\subseteq\epsilon$)
and define a partial function 
$\mathrm{Eval}: \mathrm{E} \rightharpoonup \EV$ on expressions $e\type\sigma\etype \epsilon$ by:
{\small \[\mathrm{Eval}(e) = 
   \left \{
   \begin{array}{ll}
      (r,v)  & (\myg \vdash_\epsilon e \xRightarrow{r} v) \\
       ((\ell,\op),(v, \lambda w \in \V_{\inn}.\, r\cdot\mathrm{Eval}(K[w])) & 
        (\myg \vdash_\epsilon e \xRightarrow{r} K[\op(v)], \op: \outt \xrightarrow{\ell} \inn)
   \end{array}
   \right .\]}
using the evident $R$-action on $\EV$.
Using $\{\myg\}$-induction, one sees that $\mathrm{Eval}$ is  total.
Next, inductively define a relation $\mysim$ between $\EV$ and $W_\epsilon(\ssem{\sigma})$ by:
\[(r,v) \mysim (s,a) \iff r = s \;\wedge\; \vsem{v} = a\]
and:
\[((\ell,\op),(v,f)) \mysim ((\ell',\op',n'),(a,g))
\iff \ell = \ell' \,\wedge\, \op = \op' \,\wedge\, n' = \epsilon(\ell)\,
\wedge \, \vsem{v} = a \,\wedge \, \forall w \in V_\inn. f(w) \mysim g(\ssem{w}) 
\]
Notice that in case all the $\inn$ and $\outt$ such that $\op\type \outt \xrightarrow{\ell} \inn$ for some $\ell \in \epsilon$ are first-order, and the semantics of constants is the identity, this relation is a bijection. 

\begin{theorem} \label{Athm:giant} For all $e\type \sigma\etype \epsilon$ we have: $\mathrm{Eval}(e) \mysim \ssem{e}\lsem{\myg}$.
\end{theorem}
\begin{proof}
We proceed by well-founded-tree induction on  $\mathrm{Eval(e)}$.
By Theorem~\ref{Athm:bterm} there are two cases:
\begin{enumerate}
\item Here $\myg \vdash_\epsilon e  \xRightarrow{r} v$. 
We have $\mathrm{Eval}(e) = (r,v)$ 
and, by Theorem~\ref{Athm:sound},  $\ssem{e}\lsem{\myg} = (r,\vsem{v})$.
\item Here $\myg \vdash_\epsilon e \xRightarrow{r} K[\op(v)]$.
We have 
\[\mathrm{Eval}(e) = ((\ell,\op),(v, \lambda w \in \V_{\inn}.\, r\cdot\mathrm{Eval}(K[w]))\]
and, by Theorem~\ref{Athm:sound} again, we have:
\[\begin{array}{lcl}
\ssem{e}\lsem{\myg} & = & \varphi^{\R \times \ssem{\sigma}}_{\ell,\op,\epsilon(\ell)}
 (\vsem{v},\lambda a \in \ssem{\inn}.\,
   r\cdot  \ssem{K[x]}(x \mapsto a)\lsem{\myg}\\
   & = & ((\ell,\op,\epsilon(\ell)),
 (\vsem{v},\lambda a \in \ssem{\inn}.\,
   r\cdot  \ssem{K[x]}(x \mapsto a)\lsem{\myg}))
     \end{array}\]
and the conclusion follows, since, using the induction hypothesis, 
when $a = \vsem{w}$ for $w\type\inn$ we have:
\[\mathrm{Eval}(K[w]) \mysim \ssem{K[w]}\lsem{\myg} = \ssem{K[x]}(x \mapsto \vsem{w})\lsem{\myg} 
= \ssem{K[x]}(x \mapsto a)\lsem{\myg}\]

\end{enumerate}
\end{proof}

%

%
%
%
%
%
 %
%
%
%
%
%
%
%
%
%

\end{document}